\documentclass[11pt]{article} 

\usepackage[utf8]{inputenc}
\usepackage[T1]{fontenc}
\usepackage[pdftex,left=1in,top=1in,bottom=1in,right=1in]{geometry}

\usepackage{amsmath, amssymb, dsfont, mathrsfs,amsthm}

\usepackage{graphicx}
\usepackage{caption}
\usepackage{bm}
\usepackage{xcolor}
\usepackage[colorlinks=true,citecolor=blue,linkcolor=blue]{hyperref}
\usepackage{tikz}
\usetikzlibrary{calc,positioning,shapes,backgrounds, decorations.pathreplacing}
\usepackage{pgfplots}
\usepackage{subcaption}
\pgfplotsset{compat=1.18}

\usepackage{thm-restate}
\hypersetup{
    colorlinks,
    linkcolor={red!50!black},
    citecolor={blue!50!black},
    urlcolor={blue!80!black}
}
\usepackage[nameinlink, noabbrev, capitalize]{cleveref}
\AddToHook{cmd/appendix/before}{\crefalias{section}{appendix}}

\newcommand{\C}{\mathbb{C}}
\newcommand{\N}{\mathbb{N}}

\newcommand{\R}{\mathbb{R}}
\newcommand{\E}{\mathds{E}}
\newcommand{\cA}{\mathcal{A}}

\newcommand{\cC}{\mathcal{C}}

\newcommand{\cS}{\mathcal{S}}
\newcommand{\cT}{\mathcal{T}}

\newcommand{\cX}{\mathcal{X}}
\newcommand{\cY}{\mathcal{Y}}

\newcommand{\supp}{\operatorname{supp}}

\newcommand{\eps}{\varepsilon}

\newtheorem{theorem}{Theorem}[section]

\newtheorem{claim}[theorem]{Claim}
\newtheorem{lemma}[theorem]{Lemma}
\newtheorem{conjecture}[theorem]{Conjecture}
\newtheorem{corollary}[theorem]{Corollary}

\newtheorem{definition}[theorem]{Definition}
\newtheorem{remark}[theorem]{Remark}

\newcommand{\Ch}{\mathsf{Ch}}
\newcommand{\Geom}{\operatorname{Geom}}
\newcommand{\Capa}{\operatorname{Cap}}
\newcommand{\CPUC}{\operatorname{CPUC}}

\newcommand{\bits}{\{0,1\}}

\newcommand{\Mod}[1]{\ \mathrm{mod}\ #1}

\newcommand{\cN}{\mathcal{N}}

\newcommand{\Bin}{\mathsf{Bin}}
\newcommand{\Ber}{\mathsf{Ber}}

\newcommand{\KL}{D}

\newcommand{\calX}{\mathcal{X}}
\newcommand{\calY}{\mathcal{Y}}

\let\originalleft\left
\let\originalright\right
\renewcommand{\left}{\mathopen{}\mathclose\bgroup\originalleft}
\renewcommand{\right}{\aftergroup\egroup\originalright}

\newcommand{\capBAS}{C(p)}
\newcommand{\capBASfinite}{C(p, \tau)}

\newcommand{\ANSC}{\mathrm{ANSC}}
\newcommand{\ANMC}{\mathrm{ANMC}}

\newcommand{\dd}{\textnormal{d}}

\makeatletter
\def\blfootnote{\xdef\@thefnmark{}\@footnotetext}
\makeatother

\title{Capacity of Additive-Noise Sticky Channels\thanks{This is the extended version of a paper accepted to appear at ITW 2026.}
}

\allowdisplaybreaks

\author{Cécile Bouette\thanks{Equal contribution.} \thanks{Instituto de Telecomunicações and Departamento de Matemática, Instituto Superior Técnico, Universidade de Lisboa. \texttt{\{cecile.bouette,samuelpearson,jribeiro\}@tecnico.ulisboa.pt}.} \and Samuel Pearson\footnotemark[1] \footnotemark[2]  \and Roni Con\thanks{School of Electrical \& Computer Engineering, Tel Aviv
University. \texttt{ronicon@tauex.tau.ac.il}.} \and João Ribeiro\footnotemark[2]}

\date{}

\begin{document}
	
\maketitle

\begin{abstract}
Sticky channels, which never destroy nor create runs, are some of the simplest types of channels with synchronization errors (such as deletions, insertions, and replications). Despite their simplicity, we know little about the capacity of the sticky channels considered in the literature so far beyond what can be gleaned from purely numerical methods. Towards a broader and systematic exploration of sticky channels, we initiate the study of \emph{additive-noise sticky channels}, a basic family of sticky channels that extend each input \emph{run} by an amount independently sampled from a fixed noise distribution. The capacity of these channels corresponds to the capacity per unit cost of additive-noise memoryless channels over the non-negative integers, and they capture some aspects of the loss of synchronization caused by homopolymer length miscalls in DNA sequencing.

We first focus on the setting of Bernoulli additive noise with parameter $p$, and uncover curious behavior of the capacity beyond what numerical methods can tell us. For example, the capacity is constant when $p\in [1/\varphi^2,1/2]$ with $\varphi\approx 1.618$ the golden ratio, achieved by zero-error runlength-constrained coding, and behaves differently when $p\approx 0$ vs.\ when $p\approx 1$. More generally, we determine the capacity exactly for all $p\leq 1/2$, and when $1/2<p<1$ we give analytical bounds that allow us to characterize the asymptotic behavior of the capacity as $p\to 1$. We also extend our analysis to the setting where input strings have bounded runlengths, motivated by DNA-based data storage.

Then, we study additive-noise sticky channels beyond Bernoulli noise. We derive capacity lower bounds for all additive-noise sticky channels whose noise distributions have a given mean and support. These lower bounds allow us to characterize the regime where zero-error runlength-constrained coding is never optimal.
\end{abstract}

\newpage
\tableofcontents
\newpage

\section{Introduction}\label{sec:intro}

Synchronization errors like deletions, insertions, and replications represent a fundamental challenge in information theory.
Our understanding of even very simple synchronization error models remains poor, and most of the information-theoretic techniques developed over the past decades are not appropriate for dealing with such errors~\cite{Mit09,CR21}.

Some of the most basic channels with synchronization errors studied so far are \emph{sticky repeat channels}~\cite{Mit08,MTL12,ISW16,CR19b}.
These are channels that independently replicate each input bit a number of times dictated by some replication distribution, but always at least once.
A particularly simple example is the duplication channel, which writes an input bit twice at the output with probability $p$ and once with probability $1-p$, for some duplication probability $p\in[0,1]$.
Because they do not introduce deletions, the capacity of every sticky repeat channel is easily seen to equal the capacity \emph{per unit cost} (with cost function $c(x)=x$) of an associated discrete memoryless channel over the positive integers, by encoding every input string as a vector of its runlengths\footnote{A \emph{run} in a string is a maximal substring of the form $b^\ell$ for alphabet symbol $b$ and integer $\ell\geq 1$. Every string is uniquesly decomposed into concatenation of runs $0^{\ell_1} \circ 1^{\ell_2}\circ \cdots\circ 0^{\ell_r}$ and then the vector of its runlengths is $(\ell_1, \ldots, \ell_r)$.} (and assuming that all inputs start with a $0$, which does not affect the capacity)~\cite{Mit08}.
The capacity per unit cost with cost function $c(x)$ of a DMC $\Ch$ with input alphabet $\N_{>0}=\{1,2,\dots\}$ is given by
\begin{equation*}
    \CPUC_c(\Ch):=\sup_{P_X:\E[X]<\infty} \frac{I(X;Y)}{\E[c(X)]}, 
\end{equation*}
where the supremum is taken over all input random variables $P_X$ supported on $\N_{>0}$ with finite expectation, $Y$ denotes the output of $\Ch$ on input $X$, and $I$ denotes mutual information.
In this work we will only consider the special case where $c(x)=x$, and will refer to this simply by ``capacity per unit cost'' and denote it by $\CPUC(\Ch)$ from here onward.

For example, the capacity of the duplication channel with duplication probability $p$ equals the capacity per unit cost of the channel that on input 
an integer $n\geq 1$ outputs $n+\Bin(n,p)$, where $\Bin(n,p)$ is a random variable following a binomial distribution with $n$ trials and success probability $p$.
Note that this DMC captures how the length of a given input run is changed by this noisy process.

Despite their apparent simplicity, our understanding of the capacity of sticky repeat channels is still very poor.
In particular, the DMCs over the positive integers resulting from the equivalence above do not have a zero cost symbol (because runlengths are strictly positive), and so 
Verdú's simplified formula for computing the capacity per unit cost~\cite{Ver90} does not apply here.
Our goal in this work is to kickstart a more systematic exploration of sticky channels.
We leverage the connection between sticky channels and DMCs over the positive integers to identify and study what we see as a basic family of sticky channels, that we call \emph{additive-noise sticky channels}.
In certain instances, we are able to exactly determine their capacity.
As we discuss below, the capacity of these channels exhibits curious behavior, and they also capture aspects of the loss of synchronization stemming from sequencing errors in DNA-based data storage.

\subsection{Additive-noise sticky channels}

As discussed above, there is an equivalence between the capacity of sticky channels and the capacity per unit cost of DMCs over the positive integers.
Prior work has used the equivalence to help analyze the capacity of sticky repeat channels.
However, we can also turn this around and use the equivalence to find new interesting types of sticky channels by focusing on core families of DMCs, and then study their properties as a reasonable starting point towards a broader and more systematic exploration of sticky channels.

Additive-noise channels are one of the most fundamental and well-studied classes of memoryless channels.
An additive-noise channel over the reals is a memoryless channel that given input $x$ from some input alphabet $\cX\subseteq \R$ outputs $x+Z$ where $Z$ is a random variable distributed according to a noise distribution $\cN$ supported on a subset of $\R$, and ``$+$'' is the usual addition in $\R$.
In this context, most prior works have focused on the case where $\cX$ is continuous (in particular, where $\cX$ is $\R$ or the set of non-negative real numbers) and $\cN$ is a continuous distribution.\footnote{There are also well-studied discrete channels that can be framed as additive-noise channels by changing the addition operation. For example, the binary symmetric channel can be seen as an additive-noise channel if we take $\cX=\bits$, $\cN$ to be Bernoulli, and ``$+$'' to mean addition modulo $2$. However, here we are specifically interested in the case where ``$+$'' means addition in $\R$.}
Here, we are interested in the discrete counterparts of these channels, where the input alphabet $\cX$ is countable and the noise distribution $\cN$ is discrete.
In particular, we focus on the case where $\cX=\N_{>0}$ and the support of $\cN$ is a subset of $\N=\{0,1,2,\dots\}$.

\emph{Additive-noise sticky channels} are the sticky channel counterparts of additive-noise channels with positive-integer input and non-negative-integer noise, in the sense that they are the sticky channels whose capacity equals the capacity per unit cost (with cost function $c(x)=x$) of such an additive-noise DMC.
More precisely, we introduce the following definition.
\begin{definition}[Additive-noise sticky channels]
    An \emph{additive-noise sticky channel} is characterized by a noise distribution $\cN$ supported on $\N$.
    Given a binary input string $x\in\bits^n$ with $r=r(x)$ runs of lengths $\ell_1,\dots,\ell_r$, the additive-noise sticky channel with noise distribution $\cN$, which we denote by $\ANSC_{\cN}$, samples $Z_1,\dots,Z_r$ i.i.d.\ according to $\cN$ and replaces the $i$-th run of $x$ by a run of the same symbol of length
    \begin{equation*}
        L_i = \ell_i + Z_i
    \end{equation*}
    for each $i\in\{1,2,\dots,r\}$.
\end{definition}

We are interested in studying the (coding) capacity of $\ANSC_{\cN}$.
This capacity equals the capacity per unit cost of an additive-noise memoryless channel,
as described more explicitly in the following lemma.
Its proof is analogous to that of Mitzenmacher~\cite[Theorem 2.1]{Mit08}.
For completeness, we provide a proof in \cref{sec:ANSC-to-ANMC}.
\begin{lemma}\label{lem:ANSC-to-ANMC}
    Let $\cN$ be a distribution supported on $\N$.
    Then, the coding capacity of $\ANSC_{\cN}$ equals the capacity per unit cost of the memoryless channel $\ANMC_{\cN}$ with input alphabet $\N_{>0}$ which on input $n$ outputs $n+Z$ with $Z$ distributed according to $\cN$. 
\end{lemma}

\subsection{Additive-noise sticky channels and homopolymer length miscalls in DNA-based data storage}

DNA-based data storage is an emerging technology boasting much higher storage density and durability than currently used systems.
In those systems, data (encoded in a $4$-letter alphabet) is synthesized into DNA strands, with each letter mapped to a \emph{nucleotide}.
Then, sequencing technologies are used to recover stored data from these strands.
DNA sequencing is an error-prone process that introduces at the base level various types of errors, from substitutions to errors that cause loss of synchronization, such as deletions, insertions, and replications.

Long runs of nucleotides (called \emph{homopolymers}, corresponding to long runs in the stored string) are particularly problematic for sequencing technologies.
Loss of synchronization due to homopolymer length miscalls is a well-documented error type across DNA sequencing platforms, predominantly under-calls in nanopore-based sequencing~\cite{YGM17,WGGH23}, length miscalls in both directions in flow-based platforms such as Ion Torrent (where it is the dominant type of error~\cite{BSBHT13}), and structurally important even on Illumina~\cite{WGGH23}.
For example, in nanopore-based sequencing stored symbols are inferred from an electrical current signal generated as a DNA strand sequentially passes through a biological pore. The measured signal is continuous and, at any given time, depends on a short window of consecutive nucleotides within the pore.
The window shifts by one position at random, and these shifts are detected by changes in the current signal.
Long homopolymers generate an almost constant signal, and their length must be inferred from the duration of this signal.
Since shifting is not a deterministic process, there is a chance of under- or over-estimating the length of the homopolymer (equivalently, the length of the corresponding run in the stored string), leading to loss of synchronization~\cite{RKR18}.
Overall, homopolymers increase the likelihood of synchronization errors, with error rates typically increasing with the run length~\cite{YGM17,WGGH23,CR26,CMRS26}.

We see additive-noise sticky channels as a toy model that captures important aspects of the loss of synchronization inherent to homopolymer length miscalls. 
It is not a faithful model of any end-to-end sequencing technology, but rather a tractable simplification that may be a reasonable starting point for studying the effect of homopolymer length miscalls on the system's performance.
More concretely, the model makes some simplifying assumptions compared to practice.
First, it only considers replications (i.e., overestimation of homopolymer lengths).
Second, it assumes that the number of errors within a run does not depend on its length nor on the number of errors in surrounding runs.
Third, it considers binary (instead of quaternary) input symbols.
Some other recent works have considered non-sticky synchronization errors and runlength-dependent error probabilities in the context of homopolymer length miscalls, but they either focus on the incomparable setting of adversarial errors~\cite{CMRS26} or only provide loose capacity lower bounds~\cite{CR26}.
As we shall see below, we are able to carry out a much sharper study of the properties of our channels (naturally, because our model is simpler).

\subsection{An overview of our contributions}

\subsubsection{The Bernoulli additive-noise sticky channel}

We begin by focusing on the simplest case of Bernoulli noise.
That is, we take $\cN$ to be a Bernoulli distribution which outputs $1$ with  probability $p\in[0,1]$ and $0$ with probability $1-p$.
We denote the corresponding channel by $\ANSC_{\Ber(p)}$ and its coding capacity by $C(p)$.
By \cref{lem:ANSC-to-ANMC}, $C(p)=\CPUC(\ANMC_{\Ber(p)})$, where we recall that $\ANMC_{\Ber(p)}$ is the additive-noise memoryless channel that on input $n\in\N_{>0}$ outputs $n+Z$ with $Z$ being a Bernoulli random variable with success probability $p$.

As with all sticky channels studied so far, it is not hard to numerically approximate $C(p)$ for any given $p$ with great accuracy (we discuss this in \cref{sec:numerical}).
However, such an approach provides limited insights about the channel. 
Our goal, then, is to develop an understanding of the capacity that is out of reach of numerical methods, either by determining $C(p)$ exactly or by deriving asymptotic results.

For the $\ANSC_{\Ber(p)}$ there is a simple baseline lower bound on $C(p)$ obtained by considering only input strings $x\in\bits^n$ with runs of odd length.
It is clear that such a code is zero-error over the $\ANSC_{\Ber(p)}$, since an input run of length $\ell$ is mapped to a run of length either $\ell$ or $\ell+1$ at the output.
As we show in \cref{sec:opt-rlc}, this naive \emph{runlength-constrained} coding strategy, which we will shorthand as \emph{``naive RLC coding''} throughout our work, yields the lower bound
\begin{equation*}
    C(p)\geq \log \varphi \approx 0.694
\end{equation*}
valid for any $p\in[0,1]$, where $\varphi=\frac{1+\sqrt{5}}{2}\approx 1.618$ is the golden ratio (the unique positive solution to $x^2-x-1=0$) and $\log$ is the base-$2$ logarithm.
Some questions arise:
\begin{enumerate}
    \item Is naive RLC coding ever optimal?

    \item How much better can we do than naive RLC coding?
\end{enumerate}

\paragraph{Optimality of naive RLC coding.}

We settle the first question by characterizing the values of $p$ for which naive RLC coding is optimal.

\begin{restatable}{theorem}{optrlc}\label{thm:opt-rlc}
    We have $C(p)\geq \log\varphi$ for all $p\in[0,1]$, and $C(p)=\log\varphi$ if and only if $p\in\left[\frac{1}{\varphi^2}\approx 0.382,\frac{1}{2}\right]$.
\end{restatable}

For more details, see \cref{sec:opt-rlc}.
In particular, this means that naive RLC coding is optimal even when the noise is somewhat biased towards $0$, but immediately becomes sub-optimal when the noise is biased towards $1$.

\paragraph{The $p\leq 1/2$ setting.}

With the first question answered, we move towards understanding the second question.
We go beyond \cref{thm:opt-rlc} and determine $C(p)$ (along with a capacity-achieving input distribution) for all $p\leq 1/2$.

\begin{restatable}{theorem}{capbelowhalf}\label{thm:cap-below-half}
    When $p\leq 1/2$ we have
    \begin{equation*}
        C(p) = \begin{cases}
            \log\varphi, &\textrm{ if $p\in\left[\frac{1}{\varphi^2},\frac{1}{2}\right]$,}\\
            \log x_p, &\textrm{ if $p\in[0,\frac{1}{\varphi^2})$,}
        \end{cases}
    \end{equation*}
    where $x_p$ is the unique solution to the equation $x - 2^{-h(p)} x^p - 1 = 0$
    with $h(p)=-p\log p-(1-p)\log(1-p)$ the binary entropy function.
\end{restatable}

For more details, see \cref{sec:genbound,sec:p-leq-half}.

\paragraph{The $p>1/2$ setting.}
In this regime we are not able to exactly determine $C(p)$, but we derive sharp bounds that, in particular, allow us to pinpoint the asymptotic behavior of $C(p)$.

\begin{restatable}{theorem}{caplarge}\label{thm:cap-ub-largep}
    Fix $p\in(1/2,1]$.
    For $\beta\in\R$, denote by $\lambda(p,\beta)$ the unique solution to the equation
    \begin{equation*}
        \sum_{n=1}^\infty 2^{-\lambda(p,\beta)(n-p)-h(p)-\left(-\frac{1-p}{p}\right)^{n-1} \beta}=1.
    \end{equation*}
    Then, $C(p)\leq \inf_{\beta\in\R}\lambda(p,\beta)$.
\end{restatable}

For more details, see \cref{sec:genbound,sec:p-geq-half}.
Based on a comparison with numerical approximations of $C(p)$ in \cref{sec:numerical}, we conjecture that the upper bound in \cref{thm:cap-ub-largep} is actually the capacity.
We formalize this in \cref{conj:cap-largep} and leave it as an interesting direction for future work.

We use \cref{thm:cap-ub-largep} to better understand the asymptotic behavior of $C(p)$ when $p\to 1$ and when $p\to \frac{1}{2}$ from above:
\begin{itemize}
    \item In the regime where $p=1-\eps$ with $\eps\to 0$, we can show using \cref{thm:cap-ub-largep} and a geometric input distribution that
\begin{equation*}
    C(1-\eps) = 1-(1+o(1))\frac{h(\eps)}{4},
\end{equation*}
where $o(1)\to 0$ as $\eps\to 0$.
For details, see \cref{sec:p-to-1}.
It is interesting to compare $C(1-\eps)$ and $C(\eps)$ when $\eps\to 0$, as perhaps one would expect the capacity to behave similarly in both of these regimes.
However, it is not hard to see that $C(\eps)=1-(1+o(1))\frac{h(\eps)}{2}$, and so $C(1-\eps)>C(\eps)$ for all sufficiently small $\eps$.

\item In the regime where $p=\frac{1}{2}+\eps$ with $\eps\to 0$ from above we also obtain two-sided bounds on $C(p)$ that, while not being sharp, give us concrete estimates on the strict improvement over naive RLC coding.
More precisely, we show that
\begin{equation*}
     \log \varphi \cdot \left( 1 + \log \varphi \cdot  \frac{4(3-\varphi)\ln 2}{5\sqrt{5}(\varphi-1)}\cdot \eps^{2} \right)  -O\left(\eps^{3}\right) \leq C\left(\frac{1}{2}+\eps\right) \leq \log\varphi\cdot \left(1+\frac{2\eps}{2+\varphi}\right)+O\left(\eps^2\right).
\end{equation*}
The lower bound is obtained by considering an appropriate convex combination of geometric-like input distributions.
For details, see \cref{sec:p-to-half}.
\end{itemize}

\paragraph{Capacity under bounded runlengths.}

Motivated by the connection to DNA data storage, where data has to be encoded into codewords with bounded runlengths due to various system constraints, we extend our study of the $\ANSC_{\Ber(p)}$ to the setting where inputs $x\in\bits^n$ only have runs of length at most some integer $\tau\geq 1$.
When the noise is Bernoulli with success probability $p$, we denote the capacity in this setting by $C(p,\tau)$.
As before, our goal is to obtain insights about $C(p,\tau)$ beyond what can be extracted from purely numerical methods.

When $\tau$ is an odd integer, we prove analogs of \cref{thm:opt-rlc,thm:cap-below-half} for $C(p,\tau)$, in particular characterizing the region where naive RLC coding is optimal. For $p \leq \frac{1}{2}$, we show that $C(p,\tau)$ approaches $C(p)$ exponentially fast in $\tau$.
When $\tau$ is even, we find that naive RLC coding is never optimal. Thus, we are not able to determine the exact capacity for all values of $\tau$ when $p$ belongs to the interval given in \Cref{thm:opt-rlc}. Moreover, for $p$ in the interior of that interval, we prove that the capacity-achieving distribution also does not have full support, as long as $\tau$ is even and large enough.
Details can be found in \cref{sec:cap-rlc}.

\subsubsection{General additive-noise sticky channels}

In the second part of our work we study additive-noise sticky channels beyond Bernoulli noise.
First, we give a universal lower bound on the capacity of any additive-noise sticky channel whose noise distribution has a prescribed mean $\mu$ that cannot be improved in general.
We then derive sharper lower bounds on the capacity of additive-noise sticky channels whose noise distribution has \emph{bounded support} and prescribed mean, and completely characterize the regime where naive RLC coding is never optimal for such channels.
We obtain such bounds by showing that the worst-case noise in each of these settings is of geometric type, in the spirit of the result (and argument) of Bedekar and Azizoglu for discrete-time queues~\cite{BA98}.

\paragraph{Capacity lower bounds for any additive-noise sticky channel with prescribed noise mean.}

Fix an arbitrary noise distribution $\cN$ with mean $\mu$ and denote the capacity of the associated $\ANSC_{\cN}$ by $\Capa(\cN)$.
The next theorem shows that geometric noise is worst-case and gives a universal lower bound on $\Capa(\cN)$ for all such $\cN$.
\cref{fig:lower_bound_capacity_infinite_support} plots this bound as a function of $\mu$.
For more details, see \cref{sec:universal-LB}.

\begin{restatable}{theorem}{lbmean}
\label{thm:lower_bound_capacity_noise_infinite_support}
Let $\Geom_\mu$ denote the geometric distribution supported on $\N$ of mean $\mu>0$.
Then, for all noise distributions $\cN$ of mean $\mu$ it holds that
\begin{align}
\label{eq:lower_bound_capacity_noise_infinite_support}
    \Capa(\cN)\geq \Capa(\Geom_\mu) = \alpha,
\end{align}
where $\alpha>0$ is the smallest positive solution to
\begin{align}
    \label{equation_tau_star}
    \frac{2^{\alpha\mu-(1+\mu)h\left(\frac{1}{1+\mu}\right)}}{2^\alpha-1} = 1.
\end{align}
\end{restatable}

\begin{figure}
		\centering	\begin{tikzpicture}
		
		\begin{axis}[
            width = 12cm,
            height=8cm,
    		scaled y ticks = true,
    		tick label style={/pgf/number format/fixed },
    		axis lines = left,
    		xlabel = $\mu$,
    		ylabel = {Rate},
    		xtick={0,1,2,3,4,5,6,7,8,9,10},
    		ytick={0.2,0.4,0.6,0.8,1}, 
            xmin=0, xmax=10,
            ymin=0, ymax=1,
    		legend pos=south east,
    		ymajorgrids=true,
    		grid style=dashed,
		]
            \addplot[
		color=blue,
		]
		coordinates {
            (0.0, 0.9999999999999991)
            (0.13513513513513514, 0.7744018242947653)
            (0.2702702702702703, 0.6663399112708642)
            (0.40540540540540543, 0.5923465131492864)
            (0.5405405405405406, 0.5366456133738614)
            (0.6756756756756757, 0.4924944694969425)
            (0.8108108108108109, 0.4562999592942418)
            (0.945945945945946, 0.4259035203753125)
            (1.0810810810810811, 0.39990394428651654)
            (1.2162162162162162, 0.3773402722394477)
            (1.3513513513513513, 0.3575256473547892)
            (1.4864864864864866, 0.33995294509286483)
            (1.6216216216216217, 0.32423778999688724)
            (1.7567567567567568, 0.31008244690989517)
            (1.891891891891892, 0.2972520264539331)
            (2.027027027027027, 0.28555828756638124)
            (2.1621621621621623, 0.2748483047409615)
            (2.2972972972972974, 0.26499635010806183)
            (2.4324324324324325, 0.2558979583483572)
            (2.5675675675675675, 0.24746550901573072)
            (2.7027027027027026, 0.23962488572414575)
            (2.837837837837838, 0.23231291367832435)
            (2.9729729729729732, 0.22547536905459334)
            (3.1081081081081083, 0.21906541474034058)
            (3.2432432432432434, 0.21304235820900974)
            (3.3783783783783785, 0.2073706557426565)
            (3.5135135135135136, 0.2020191071345759)
            (3.6486486486486487, 0.19696019917259014)
            (3.783783783783784, 0.1921695664205474)
            (3.9189189189189193, 0.187625545277432)
            (4.054054054054054, 0.18330880280734188)
            (4.1891891891891895, 0.17920202595225207)
            (4.324324324324325, 0.17528965984713937)
            (4.45945945945946, 0.17155768632380683)
            (4.594594594594595, 0.1679934355081)
            (4.72972972972973, 0.16458542482364763)
            (4.864864864864865, 0.16132322081457617)
            (5.0, 0.15819732006388926)
            (5.135135135135135, 0.15519904616826208)
            (5.27027027027027, 0.1523204602748748)
            (5.405405405405405, 0.1495542831226153)
            (5.540540540540541, 0.14689382688190245)
            (5.675675675675676, 0.14433293537264186)
            (5.810810810810811, 0.14186593147211388)
            (5.9459459459459465, 0.1394875707147919)
            (6.081081081081082, 0.13719300024245892)
            (6.216216216216217, 0.134977722392183)
            (6.351351351351352, 0.13283756231688568)
            (6.486486486486487, 0.13076863912245648)
            (6.621621621621622, 0.12876734008005344)
            (6.756756756756757, 0.12683029753480984)
            (6.891891891891892, 0.12495436818495437)
            (7.027027027027027, 0.12313661444991503)
            (7.162162162162162, 0.12137428768379553)
            (7.297297297297297, 0.11966481302274051)
            (7.4324324324324325, 0.11800577568214514)
            (7.567567567567568, 0.11639490854315335)
            (7.7027027027027035, 0.1148300808879893)
            (7.837837837837839, 0.11330928816101613)
            (7.972972972972974, 0.11183064264735525)
            (8.108108108108109, 0.1103923649738278)
            (8.243243243243244, 0.10899277634815856)
            (8.378378378378379, 0.10763029146215923)
            (8.513513513513514, 0.10630341199304708)
            (8.64864864864865, 0.10501072064446344)
            (8.783783783783784, 0.10375087567525874)
            (8.91891891891892, 0.10252260586972932)
            (9.054054054054054, 0.10132470590802288)
            (9.18918918918919, 0.10015603209978721)
            (9.324324324324325, 0.09901549844800153)
            (9.45945945945946, 0.09790207301335857)
            (9.594594594594595, 0.09681477455256875)
            (9.72972972972973, 0.09575266940663592)
            (9.864864864864865, 0.09471486861755554)
            (10.0, 0.09370052525394872)
        }; 
        
        \end{axis}
	\end{tikzpicture}
	\caption{A universal lower bound on the capacity of any additive-noise sticky channel as a function of the mean $\mu$ of its noise distribution (\cref{thm:lower_bound_capacity_noise_infinite_support}). This lower bound is tight for geometric noise.} \label{fig:lower_bound_capacity_infinite_support}
	\end{figure}
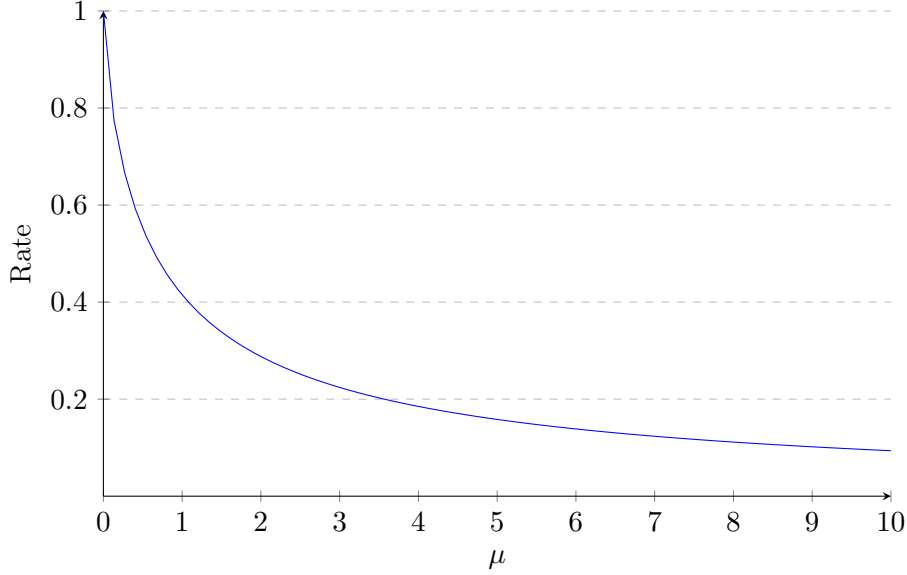

\paragraph{Improved capacity lower bounds for bounded-support noise with prescribed mean.}

In \cref{thm:lower_bound_capacity_noise_infinite_support} we gave a universal capacity lower bound that applies to all additive-noise sticky channels whose noise distributions have some fixed mean $\mu$.
The bound in that theorem cannot be improved in general, since it is sharp for the additive-noise sticky channel with $\Geom_\mu$ noise.
To complement this, we obtain sharper universal capacity lower bounds for channels whose noise distributions have \emph{bounded support} and a given mean.
More precisely, we consider an arbitrary additive-noise sticky channel whose noise distribution $\cN$ has mean $\mu$ and bounded support.
Our goal is to derive a lower bound on $\Capa(\cN)$ solely as a function of $\mu$ and the largest element in the support $m=\max\supp(\cN)$.

A first easy observation is that, just like for Bernoulli noise, there is a simple baseline lower bound obtained by considering naive RLC coding, which in this case amounts to coding only over positive integers $i$ satisfying $i\equiv 1 \pmod{m+1}$.
This strategy yields the general lower bound $\Capa(\cN)\geq \log \varphi$, with $\varphi$ now more generally the unique positive solution to $x^{m+1}-x^m-1=0$.
It is then natural to wonder for which choices of $m$ and $\mu$ it holds that naive RLC coding is optimal for the $\ANSC_{\cN}$ with $\cN$ some noise distribution with mean $\mu$ and $\max\supp(\cN)=m$, and how much better we can do than naive RLC coding otherwise.
Our next theorem makes progress in this direction.
It completely characterizes the choices of $m$ and $\mu$ for which naive RLC coding is never optimal, gives best possible bounds whenever $\mu<m/2$, and generalizes some of our prior results for Bernoulli noise from \cref{sec:bernoulli}.
For more details, see \cref{section_lower_bound_capacity_bounded_noise}.
This lower bound is compared with that of \cref{thm:lower_bound_capacity_noise_infinite_support} as a function of $\mu$ for $m=2,3$ in \cref{fig:lower_bound_capacity_finite_supporta,fig:lower_bound_capacity_finite_supportb}.

\begin{theorem}\label{theorem_lower_bound_capacity_noise_finite_support}
    Fix an integer $m>0$ and a real number $\mu>0$. 
    Consider an arbitrary noise distribution $\cN$ with mean $\mu$ and such that $\max\supp(\cN)=m$.
    Let $\varphi$ be the unique positive solution to $x^{m+1}-x^m-1=0$.
    If $\mu<m/2$, let $\delta$ be the unique solution in $(0,1)$ to the equation\footnote{Note that the right-hand side of \cref{relation_mu_delta} is continuous, converges to $0$ when $x\to 0$, converges to $m/2$ when $x\to 1$, and is increasing in $x$ for fixed $m$. Therefore, such a (unique) solution is guaranteed to exist when $\mu<m/2$.}
    \begin{align}
        \label{relation_mu_delta}
        \mu=\frac{x}{1-x}-\frac{(m+1)x^{m+1}}{1-x^{m+1}}.
    \end{align}
    Then, the following hold:
    \begin{enumerate}
        \item If $\mu\geq m/2$, then $\Capa(\cN)>\log\varphi$ unless $\cN$ is uniform over $\{0,1,\dots,m\}$. That is, in this regime naive RLC coding is optimal if and only if $\cN$ is uniform;

        \item If $\mu<m/2$ and $\delta\geq 1/\varphi$, then $\Capa(\cN)\geq \log\varphi$, and there exists a noise distribution $\cN$ fully supported on $\{0,1,\dots,m\}$ with mean $\mu$ such that $\Capa(\cN)=\log\varphi$ in this regime. In other words, naive RLC coding is optimal for some channel in this regime;

        \item If $\mu<m/2$ and $\delta<1/\varphi$, then
        \begin{equation}
        \label{eq:theorem_lower_bound_capacity_noise_finite_support}
            \Capa(\cN)\geq \alpha > \log\varphi,
        \end{equation}
        where $\alpha$ is the solution to
        \begin{equation*}
        \frac{2^{\alpha\mu+\log\left(\frac{1-\delta}{1-\delta^{m+1}}\right)+\mu\log\delta}}{2^{\alpha}-1}=1,
    \end{equation*}
        and there exists a noise distribution $\cN$ fully supported on $\{0,1,\dots,m\}$ with mean $\mu$ such that \cref{eq:theorem_lower_bound_capacity_noise_finite_support} is equal to $\alpha$.
        In particular, naive RLC coding is never optimal in this regime.
    \end{enumerate}
\end{theorem}

\begin{figure}
    \begin{subfigure}{0.48\textwidth}

		\centering	\begin{tikzpicture}
		
		\begin{axis}[
            width = 7.5cm,
            height=7cm,
    		scaled y ticks = true,
    		tick label style={/pgf/number format/fixed },
    		axis lines = left,
            x axis line style={-},
    		xlabel = $\mu$,
    		ylabel = {Rate},
            xtick={0,0.2,0.4,0.6,0.8, 1.0},
            xticklabels={0,0.2,0.4,0.6,0.8, 1.0},
    		ytick={0.2,0.4,0.551463,0.6,0.8,1},
            yticklabels={0.2,0.4,$\log\varphi$,0.6,0.8,1},
            xmin=0, xmax=1.,
            ymin=0.2, ymax=1,
            legend style={
                font=\scriptsize
            },
    		legend pos=north east,
    		ymajorgrids=true,
    		grid style=dashed,
		]
            \addplot[
		color=red,
		]
		coordinates {
            (1.000009999800001e-05, 0.9999140861223889)
            (0.009314147596458868, 0.9667267815007723)
            (0.018779016948452475, 0.9425277504061753)
            (0.028395471015937677, 0.921629949011816)
            (0.0381545602657606, 0.902835370428767)
            (0.0480475286570668, 0.8855871899783251)
            (0.05806581873577293, 0.8695617225352439)
            (0.0682010759719277, 0.8545489001472008)
            (0.07844515237465, 0.8404028571748481)
            (0.08879010941992248, 0.8270151190379719)
            (0.09922822032686661, 0.8143034973002952)
            (0.10975197171824493, 0.8022028388165461)
            (0.12035406470085362, 0.7906599302819929)
            (0.1310274154011986, 0.7796310245968398)
            (0.14176515499141287, 0.7690795072496247)
            (0.1525606292397905, 0.7589736281232207)
            (0.1634073976195989, 0.7492865506384316)
            (0.1742992320090065, 0.7399940041467741)
            (0.18523011501403913, 0.7310754438397685)
            (0.1961942379454763, 0.7225113229754032)
            (0.20718599847952548, 0.7142846705212089)
            (0.21819999803098875, 0.7063802839727575)
            (0.22923103886646723, 0.6987842936075848)
            (0.24027412098395173, 0.6914830810077854)
            (0.25132443878392835, 0.6844648406906184)
            (0.2623773775558985, 0.677718790981905)
            (0.2734285098029795, 0.6712345064413245)
            (0.2844735914260237, 0.6650018441713621)
            (0.29550855778747664, 0.6590121019298192)
            (0.3065295196739966, 0.6532562488739888)
            (0.3175327591756757, 0.6477267123637125)
            (0.32851472549855754, 0.6424159106453785)
            (0.3394720307260206, 0.6373162229003317)
            (0.35040144554350827, 0.632420395680882)
            (0.3612998949400417, 0.6277224633273006)
            (0.37216445389892755, 0.6232159864813411)
            (0.3829923430891071, 0.6188950062732175)
            (0.39378092456765584, 0.6147538639035142)
            (0.4045276975030479, 0.6107867112766819)
            (0.4152302939279578, 0.6069885560463698)
            (0.4258864745295521, 0.6033544814893147)
            (0.43649412448446445, 0.5998791815330909)
            (0.4470512493449252, 0.5965588727144707)
            (0.4575559709818182, 0.5933883729152627)
            (0.4680065235898112, 0.5903634032704421)
            (0.47840124975908654, 0.5874806232541316)
            (0.4887385966176357, 0.5847352795600945)
            (0.4990171120475571, 0.5821236423495649)
            (0.5092354409782895, 0.5796425637023681)
            (0.5193923217592602, 0.5772879951248766)
            (0.529486582614004, 0.5750565782855102)
            (0.5395171381773963, 0.5729454162996018)
            (0.5494829861172992, 0.5709509108301738)
            (0.5593832038415631, 0.5690703034404111)
            (0.5692169452910274, 0.5673005152885758)
            (0.5789834378188782, 0.5656385636444545)
            (0.5886819791564578, 0.5640820450082772)
            (0.5983119344653863, 0.5626283826292839)
            (0.6078727334756453, 0.5612746132026587)
            (0.6173638677090715, 0.5600184870739905)
            (0.6267848877875414, 0.5588576595873256)
            (0.6361354008249672, 0.557789815846119)
            (0.6454150679020865, 0.5568127406689548)
            (0.6546236016229022, 0.5559243048016336)
            (0.6637607637515256, 0.5551223737415157)
            (0.6728263629280704, 0.5544051983193123)
            (0.6818202524621769, 0.5537706851767306)
            (0.690742328202658, 0.5532170397048889)
            (0.6995925264817143, 0.5527424693502571)
            (0.708370822132105, 0.5523452228278163)
            (0.7170772265756227, 0.5520237650763633)
            (0.7257117859811955, 0.5517762099775689)
            (0.7342745794909029, 0.5516012672849778)
            (0.7427657175121805, 0.5514974105053035)
            (0.751185340074483, 0.5514630897455959)
            (1.0, 0.5514631)
        };
\addlegendentry{Bounded-support lower bound \eqref{eq:theorem_lower_bound_capacity_noise_finite_support}};
        
 \addplot[
		color=blue,
		]
		coordinates {
            (0.0, 0.9999999999999991)
            (0.02027027027027027, 0.9390865312855196)
            (0.04054054054054054, 0.898487662995086)
            (0.060810810810810814, 0.8652582898532442)
            (0.08108108108108109, 0.8366271123045965)
            (0.10135135135135136, 0.8112845820853204)
            (0.12162162162162163, 0.7884659348382063)
            (0.14189189189189189, 0.7676711179084729)
            (0.16216216216216217, 0.7485488875770824)
            (0.18243243243243246, 0.7308399265551447)
            (0.20270270270270271, 0.7143456143099957)
            (0.22297297297297297, 0.6989094496966872)
            (0.24324324324324326, 0.6844053106827833)
            (0.26351351351351354, 0.6707296746549076)
            (0.28378378378378377, 0.6577962623723802)
            (0.30405405405405406, 0.6455322325253827)
            (0.32432432432432434, 0.6338754055892196)
            (0.3445945945945946, 0.6227721924874821)
            (0.3648648648648649, 0.6121760188604397)
            (0.38513513513513514, 0.6020461059305078)
            (0.40540540540540543, 0.5923465131492862)
            (0.4256756756756757, 0.5830453764530367)
            (0.44594594594594594, 0.5741142949973391)
            (0.46621621621621623, 0.5655278321933116)
            (0.4864864864864865, 0.5572631058557024)
            (0.5067567567567568, 0.5492994486245618)
            (0.5270270270270271, 0.5416181243845405)
            (0.5472972972972974, 0.5342020897321882)
            (0.5675675675675675, 0.5270357919998847)
            (0.5878378378378378, 0.5201049971844667)
            (0.6081081081081081, 0.5133966425208129)
            (0.6283783783783784, 0.5068987095054737)
            (0.6486486486486487, 0.5006001139978824)
            (0.668918918918919, 0.49449061066767686)
            (0.6891891891891893, 0.4885607095604405)
            (0.7094594594594595, 0.4828016029532784)
            (0.7297297297297298, 0.47720510099010294)
            (0.75, 0.4717635748424093)
            (0.7702702702702703, 0.4664699063482745)
            (0.7905405405405406, 0.46131744325070195)
            (0.8108108108108109, 0.45629995929423933)
            (0.8310810810810811, 0.45141161855217593)
            (0.8513513513513514, 0.4466469434504022)
            (0.8716216216216217, 0.44200078603193366)
            (0.8918918918918919, 0.43746830207117515)
            (0.9121621621621622, 0.43304492770153546)
            (0.9324324324324325, 0.4287263582659641)
            (0.9527027027027027, 0.4245085291388003)
            (0.972972972972973, 0.4203875983002858)
            (0.9932432432432433, 0.4163599304731595)
            (1.0135135135135136, 0.4124220826547319)
        };
\addlegendentry{General lower bound \eqref{eq:lower_bound_capacity_noise_infinite_support}}
        
        \end{axis}
	\end{tikzpicture}
	\caption{A comparison between the general capacity lower bound (\cref{thm:lower_bound_capacity_noise_infinite_support}) and the capacity lower bound for $m=2$ (\cref{theorem_lower_bound_capacity_noise_finite_support}) as a function of $\mu$.} \label{fig:lower_bound_capacity_finite_supporta}

    \end{subfigure}
    \begin{subfigure}{0.48\textwidth}

		\centering	\begin{tikzpicture}
		
		\begin{axis}[
            width = 8.5cm,
            height=7cm,
    		scaled y ticks = true,
    		tick label style={/pgf/number format/fixed },
    		axis lines = left,
            x axis line style={-},
    		xlabel = $\mu$,
            xtick={0,0.2,0.4,0.6,0.8,1.0,1.2,1.4},
            xticklabels={0,0.2,0.4,0.6,0.8,1.0,1.2,1.4, },
    		ytick={0.2,0.4,0.465,0.6,0.8,1},
            yticklabels={0.2,0.4,$\log\varphi$,0.6,0.8,1},
            xmin=0, xmax=1.5,
            ymin=0.2, ymax=1,
            legend style={
                font=\scriptsize
            },
    		legend pos=north east,
    		ymajorgrids=true,
    		grid style=dashed,
		]
            \addplot[
		color=red,
		]
		coordinates {
            (1.000010000099997e-05, 0.9999140861223649)
            (0.009897256253909964, 0.9650825934448651)
            (0.019981464595814596, 0.9397457339810686)
            (0.030267276625123504, 0.9178452469905374)
            (0.04075870539774142, 0.8980996415047888)
            (0.05145913801898189, 0.8799067629857912)
            (0.062371348854094594, 0.8629203546775749)
            (0.07349751346629843, 0.8469143916279166)
            (0.08483922328346932, 0.8317326300468512)
            (0.09639750098633867, 0.8172604141512467)
            (0.10817281660324332, 0.803411575110945)
            (0.12016510428916444, 0.790118663554739)
            (0.132373779760024, 0.7773279065498145)
            (0.14479775834700176, 0.7649976253879169)
            (0.15743547363001073, 0.7530909477804265)
            (0.17028489660444013, 0.7415790010419235)
            (0.18334355533085625, 0.7304375353994773)
            (0.19660855501354613, 0.7196458117913538)
            (0.2100765984506029, 0.7091858998488918)
            (0.2237440067956756, 0.6990421300737677)
            (0.2376067405695386, 0.6892020192371464)
            (0.2516604208582586, 0.679653193394)
            (0.2659003506339366, 0.6703866730092773)
            (0.28032153613375066, 0.6613925318045878)
            (0.2949187082333097, 0.6526629164420257)
            (0.30968634375110016, 0.6441905309553582)
            (0.32461868662205623, 0.6359694957478809)
            (0.339709768879955, 0.6279932511444718)
            (0.35495343139041213, 0.6202566956522189)
            (0.3703433442786741, 0.612754449950821)
            (0.3858730269991481, 0.605482140048135)
            (0.40153586799662483, 0.5984353420723887)
            (0.4173251439124011, 0.5916095256776998)
            (0.43323403829195467, 0.5850016335433738)
            (0.4492556597544227, 0.5786070319628234)
            (0.46538305958785087, 0.5724225999980916)
            (0.481609248737972, 0.5664448822400127)
            (0.49792721416210817, 0.5606700530496718)
            (0.5143299345236345, 0.5550957810842672)
            (0.5308103952062645, 0.5497181216976801)
            (0.5473616026311791, 0.5445341315950087)
            (0.5639765978637191, 0.5395409570607013)
            (0.5806484694999378, 0.5347353662659999)
            (0.5973703658267769, 0.5301149376789287)
            (0.6141355062529477, 0.5256767477357209)
            (0.6309371920107494, 0.5214173329617956)
            (0.6477688161320538, 0.5173345821576496)
            (0.6646238727044765, 0.5134255981193128)
            (0.6814959654163703, 0.5096875877196245)
            (0.698378815401689, 0.5061179596512863)
            (0.7152662683979724, 0.5027141645880473)
            (0.7321523012327138, 0.49947371704154486)
            (0.7490310276551766, 0.49639382218519174)
            (0.7658967035323174, 0.4934724704479962)
            (0.7827437314288803, 0.49070729449053463)
            (0.7995666645929391, 0.4880954281679873)
            (0.8163602103691691, 0.48563502329084146)
            (0.8331192330629968, 0.48332337827773025)
            (0.8498387562794074, 0.48115889543821494)
            (0.866513964760731, 0.4791390890536733)
            (0.8831402057480425, 0.47726209321252344)
            (0.8997129898910444, 0.47552546725129663)
            (0.9162279917313416, 0.47392782226845553)
            (0.9326810497839884, 0.47246693900664005)
            (0.9490681662419914, 0.47114094537423534)
            (0.9653855063281971, 0.469948372325714)
            (0.9816293973186084, 0.46888723679347133)
            (0.9977963272607288, 0.4679561992189614)
            (1.0138829434099987, 0.46715360503906217)
            (1.029886050406792, 0.4664779167727033)
            (1.0458026082157958, 0.46592789527696776)
            (1.061629729848884, 0.465502064228752)
            (1.0773646788918683, 0.46519935459287337)
            (1.0930048668547196, 0.4650184714536744)
            (1.1085478503640784, 0.4649584172162091)
            (1.1085479, 0.464958)
            (1.49, 0.464958)
        };
\addlegendentry{Bounded-support lower bound \eqref{eq:theorem_lower_bound_capacity_noise_finite_support}};
        
 \addplot[
		color=blue,
		]
		coordinates {
            (0.0, 0.9999999999999991)
            (0.02027027027027027, 0.9390865312855196)
            (0.04054054054054054, 0.898487662995086)
            (0.060810810810810814, 0.8652582898532442)
            (0.08108108108108109, 0.8366271123045965)
            (0.10135135135135136, 0.8112845820853204)
            (0.12162162162162163, 0.7884659348382063)
            (0.14189189189189189, 0.7676711179084729)
            (0.16216216216216217, 0.7485488875770824)
            (0.18243243243243246, 0.7308399265551447)
            (0.20270270270270271, 0.7143456143099957)
            (0.22297297297297297, 0.6989094496966872)
            (0.24324324324324326, 0.6844053106827833)
            (0.26351351351351354, 0.6707296746549076)
            (0.28378378378378377, 0.6577962623723802)
            (0.30405405405405406, 0.6455322325253827)
            (0.32432432432432434, 0.6338754055892196)
            (0.3445945945945946, 0.6227721924874821)
            (0.3648648648648649, 0.6121760188604397)
            (0.38513513513513514, 0.6020461059305078)
            (0.40540540540540543, 0.5923465131492862)
            (0.4256756756756757, 0.5830453764530367)
            (0.44594594594594594, 0.5741142949973391)
            (0.46621621621621623, 0.5655278321933116)
            (0.4864864864864865, 0.5572631058557024)
            (0.5067567567567568, 0.5492994486245618)
            (0.5270270270270271, 0.5416181243845405)
            (0.5472972972972974, 0.5342020897321882)
            (0.5675675675675675, 0.5270357919998847)
            (0.5878378378378378, 0.5201049971844667)
            (0.6081081081081081, 0.5133966425208129)
            (0.6283783783783784, 0.5068987095054737)
            (0.6486486486486487, 0.5006001139978824)
            (0.668918918918919, 0.49449061066767686)
            (0.6891891891891893, 0.4885607095604405)
            (0.7094594594594595, 0.4828016029532784)
            (0.7297297297297298, 0.47720510099010294)
            (0.75, 0.4717635748424093)
            (0.7702702702702703, 0.4664699063482745)
            (0.7905405405405406, 0.46131744325070195)
            (0.8108108108108109, 0.45629995929423933)
            (0.8310810810810811, 0.45141161855217593)
            (0.8513513513513514, 0.4466469434504022)
            (0.8716216216216217, 0.44200078603193366)
            (0.8918918918918919, 0.43746830207117515)
            (0.9121621621621622, 0.43304492770153546)
            (0.9324324324324325, 0.4287263582659641)
            (0.9527027027027027, 0.4245085291388003)
            (0.972972972972973, 0.4203875983002858)
            (0.9932432432432433, 0.4163599304731595)
            (1.0135135135135136, 0.4124220826547319)
            (1.0337837837837838, 0.4085707908984097)
            (1.0540540540540542, 0.40480295821631973)
            (1.0743243243243243, 0.40111564348995327)
            (1.0945945945945947, 0.39750605128893546)
            (1.114864864864865, 0.39397152250951406)
            (1.135135135135135, 0.3905095257542808)
            (1.1554054054054055, 0.3871176493833861)
            (1.1756756756756757, 0.3837935941750521)
            (1.195945945945946, 0.3805351665399037)
            (1.2162162162162162, 0.37734027223944794)
            (1.2364864864864866, 0.3742069105642036)
            (1.2567567567567568, 0.37113316893153764)
            (1.2770270270270272, 0.3681172178672558)
            (1.2972972972972974, 0.3651573063385761)
            (1.3175675675675675, 0.3622517574092988)
            (1.337837837837838, 0.3593989641907149)
            (1.3581081081081081, 0.356597386064406)
            (1.3783783783783785, 0.3538455451552348)
            (1.3986486486486487, 0.3511420230348787)
            (1.418918918918919, 0.34848545763801214)
            (1.4391891891891893, 0.34587454037487614)
            (1.4594594594594597, 0.343308013425386)
            (1.4797297297297298, 0.3407846672012509)
            (1.5, 0.3383033379637269)
        };
\addlegendentry{General lower bound \eqref{eq:lower_bound_capacity_noise_infinite_support}}
        
        \end{axis}
	\end{tikzpicture}
	\caption{A comparison between the general capacity lower bound (\cref{thm:lower_bound_capacity_noise_infinite_support}) and the capacity lower bound for $m=3$ (\cref{theorem_lower_bound_capacity_noise_finite_support}) as a function of $\mu$.} \label{fig:lower_bound_capacity_finite_supportb}

\end{subfigure}
\end{figure}

\subsection{More related work}

As discussed above, studying the capacity of additive-noise sticky channels is equivalent to studying the capacity per unit cost of additive-noise memoryless channels with discrete alphabet $\N_{>0}$ and noise supported on $\N$.
In turn, the capacity per unit cost is closely connected to the average power-constrained capacity, where we focus on inputs $X$ such that $\E[X]\leq \mu$ for some $\mu$.
The average power-constrained capacity of additive-noise memoryless channels over the non-negative integers has appeared before in the study of discrete-time queuing channels~\cite{BA98,PG03}.
There, they only consider the case where the additive noise is geometric.

In an independent concurrent work, Kovačević~\cite{Kov26} constructed examples of sticky repeat channels for which the Shannon capacity equals the zero-error capacity, achieved by runlength-constrained coding. Although Kovačević's results are not directly related to the results obtained in our work, at a high level both sets of results highlight that zero-error runlength-constrained coding can sometimes be the optimal way of handling random sticky errors, even when they are biased (as in the case of the Bernoulli additive-noise sticky channel). 
We believe that the relationship between the capacity of sticky channels (broadly construed) and zero-error runlength-constrained coding deserves further study.

\section{Preliminaries}

\subsection{Notation}

We denote random variables by uppercase roman letters such as $X$, $Y$, and $Z$.
Sets are usually denoted by calligraphic uppercase letters such as $\cS$ and $\cT$.
We write $\N_{>0}$ for the set of positive integers and $\N$ for the set of non-negative integers.
We also write $[n]$ for the set $\{1,2,\dots,n\}$.
We denote the base-$2$ logarithm by $\log$ and the natural logarithm by $\ln$.

We will deal solely with discrete random variables.
Let $X$ be a random variable supported on $\cX\subseteq\R$ and some $x\in \cX$.
We write $P_X$ for the probability mass function of $X$ and $P_X(x)$ for the probability that $X$ takes on value $x$.
We denote the expectation of $X$ by $\E[X]$.
For two arbitrarily correlated random variables $X$ and $Y$ we will use the notation $P_{Y_x}$ as shorthand for $P_{Y|X=x}$.

\subsection{Characterization of the capacity per unit cost of memoryless channels}\label{subsec:AG}

Recall \cref{lem:ANSC-to-ANMC} relating the capacity of additive-noise sticky channels to the capacity per unit cost of additive-noise memoryless channels over the non-negative integers.
Motivated by this, the following framework due to Abdel-Ghaffar~\cite{AG93} will be useful to us.

\begin{restatable}[\cite{AG93}]{lemma}{capframework}\label{lem:CPUC-UB-framework}
    Let $\Ch$ be a memoryless channel with countable input and output alphabets $\cX,\cY\subseteq \N_{>0}$, respectively, and $c:\cX\to\R_{> 0}$ a positive cost function.
    Denote by $Y_x$ the output random variable of $\Ch$ on input $x\in\cX$.
    Then, we have
    \begin{equation*}
        \CPUC_c(\Ch) \leq \sup_{x\in\cX}\frac{\KL(P_{Y_x}\|Q)}{c(x)}
    \end{equation*}
    for any distribution $Q$ supported on $\cY$, where $\KL(P\|Q)=\sum_{x} P(x)\log \frac{P(x)}{Q(x)}$ is the Kullback-Leibler divergence between discrete distributions $P$ and $Q$.
    
    Furthermore, if there exists an input distribution $P_X$ such that $Q$ is its corresponding output distribution with respect to $\Ch$ and
    \begin{equation*}
        \frac{\KL(P_{Y_x}\|Q)}{c(x)}\leq\lambda
    \end{equation*}
    for all $x\in\cX$ with equality when $P_X(x)>0$, then $\CPUC_c(\Ch)=\lambda$ and $P_X$ is capacity-achieving.
\end{restatable}

The proof of \cref{lem:CPUC-UB-framework} presented by Abdel-Ghaffar~\cite{AG93} implicitly assumes that the input and output alphabets are finite.
However, it is not hard to see that it extends to countably infinite alphabets too.\footnote{For example, we can invoke a general result of Lapidoth and Moser~\cite[Theorem 5.1]{LM03}, which states that in the setting of \cref{lem:CPUC-UB-framework} (and well beyond it) we have $I(X;Y)\leq \sum_{x\in\cX} P_X(x) \KL(P_{Y_x}\|Q)$ for any input distribution $P_X$ and any distribution $Q$ over $\cY$, with $Y$ the corresponding channel output. Then, if $\KL(P_{Y_x}\|Q)\leq \lambda c(x)$ for all $x$, we get that
\begin{equation*}
    \frac{I(X;Y)}{\E[c(X)]}\leq \frac{\sum_{x\in\cX} P_X(x) \KL(P_{Y_x}\|Q)}{\E[c(X)]} \leq \frac{\lambda \sum_{x\in\cX} P_X(x) c(x)}{\E[c(X)]} = \lambda.
\end{equation*}
}
Under some additional mild assumptions, the sufficient conditions for optimality in \cref{lem:CPUC-UB-framework} are also necessary.
This is formalized in the next lemma, which follows from standard techniques.
For completeness, we provide a proof in \Cref{sec:AG-necessary}.

\begin{lemma}\label{lem:AG_necessary}
    Let $\Ch$ be a memoryless channel with input and output alphabets $\cX,\cY \subseteq \N_{>0}$ such that $H(Y_x) < \infty$ for every $x \in \cX$.
    Suppose that $P_X$ with corresponding output distribution $P_Y$ achieves the capacity per unit cost of $\Ch$ with cost function $c$, and that $\sum_{x \in \cX} P_X(x) H(Y_x) < \infty$.
    Then,
    \begin{equation*}
        \frac{\KL(P_{Y_x}\|P_Y)}{c(x)}\leq\CPUC_c(\Ch)
    \end{equation*}
    for all $x\in\cX$ with equality for all $x$ such that $P_X(x)>0$.
\end{lemma}

\begin{remark}
    For the additive-noise sticky channel $\ANSC_{\cN}$, we have $H(Y_x) = H(\cN)$ for all $x \in \cX$. Therefore, the conditions in \cref{lem:AG_necessary} are satisfied provided that $H(\cN)<\infty$.
\end{remark}

\section{The Bernoulli additive-noise sticky channel}\label{sec:bernoulli}

\subsection{When is naive RLC coding optimal?}\label{sec:opt-rlc}

In this section we prove \cref{thm:opt-rlc}, thus characterizing the optimality region of the the naive RLC coding strategy of taking only strings with runs of odd length as codewords.
We restate the theorem here to aid the reader.
\optrlc*

\begin{proof}
    We consider the input distribution
    \begin{equation*}
        P_X(n) = \begin{cases}
        \varphi^{-n}, &\textrm{if $n$ is odd,}\\
        0, &\textrm{otherwise.}
        \end{cases}
    \end{equation*}
    First, note that
    \begin{equation*}
        C(p)=\CPUC(\ANMC_{\Ber(p)})\geq \frac{I(X;Y)}{\E[X]} = \frac{H(X)}{\E[X]} = \log \varphi.
    \end{equation*}
    The second equality uses the fact that $H(X|Y)=0$, since $X$ is supported on the odd integers and $Y\in\{X,X+1\}$ with probability $1$.
    This proves the first part of the theorem.
    
    To see the second part, we use \cref{lem:CPUC-UB-framework} to determine when $P_X$ is capacity-achieving.
    Note that
    \begin{equation*}
        \sum_{x=1}^\infty P_X(x) =\sum_{k=1}^\infty \varphi^{-(2k-1)} = 1,
    \end{equation*}
    and so $P_X$ is a valid probability mass function.
    With this choice of $X$, the corresponding output $Y$ has probability mass function
    \begin{equation*}
        P_Y(y) = \begin{cases}
            (1-p)P_X(y), &\textrm{if $y$ is odd,}\\
            p \cdot P_X(y-1), &\textrm{otherwise.}
        \end{cases}
    \end{equation*}

    We now compute $\KL(P_{Y_x}\|P_Y)$ for an arbitrary integer $x>0$, where we recall that $P_{Y_x}$ denotes the channel output distribution on input $x$.
    For odd $x$, we get
    \begin{align*}
        \KL(P_{Y_x}\|P_Y) &= (1-p)\log\left(\frac{1}{P_Y(x)}\right) + p\log\left(\frac{1}{P_Y(x+1)}\right) - h(p)\\
        &=(1-p)\log\left(\frac{1}{(1-p)P_X(x)}\right) + p\log\left(\frac{1}{p\cdot P_X(x)}\right) -h(p)\\
        &= \log\left(\frac{1}{P_X(x)}\right)\\
        &= x\cdot \log\varphi.
    \end{align*}

    For even $x$, we get
    \begin{align*}
        \KL(P_{Y_x}\|P_Y) &= (1-p)\log\left(\frac{1}{P_Y(x)}\right) + p\log\left(\frac{1}{P_Y(x+1)}\right) - h(p)\\
        &=(1-p)\log\left(\frac{1}{p \cdot P_X(x-1)}\right) + p\log\left(\frac{1}{(1-p)P_X(x+1)}\right) - h(p)\\
        &= (1-p)(x-1)\log\varphi + p(x+1)\log\varphi -(1-p)\log p -p\log(1-p)-h(p)\\
        &=x\cdot \log\varphi +(2p-1)\log\varphi-(1-p)\log p -p\log(1-p)-h(p).
    \end{align*}
    Define $\Delta(p)=(2p-1)\log\varphi-(1-p)\log p -p\log(1-p)-h(p)$.
    We get that $\KL(P_{Y_x}\|P_Y)\leq x \log(\varphi)$ if and only if $\Delta(p)\leq 0$, in which case we conclude that $P_X$ is capacity-achieving.
    By routine algebraic manipulation, we conclude that this holds if and only if
    \begin{equation*}
        \frac{1}{\varphi^2}\leq p\leq \frac{1}{2}.
    \end{equation*}
    Consequently, for $p$ in this interval we get via \cref{lem:CPUC-UB-framework} that
    \begin{equation*}
        C(p)=\CPUC(\ANMC_{\Ber(p)})=\log\varphi.
        \qedhere
    \end{equation*}
\end{proof}

\subsection{A capacity upper bound for general $p$}\label{sec:genbound}

In this section we derive an upper bound on $C(p)$ that holds for arbitrary $p$.
For $p<1/\varphi^2$ we show that this upper bound actually \emph{is} the capacity, and so, together with \cref{thm:opt-rlc}, we determine $C(p)$ exactly for $p\in[0,1/2]$, leading to \cref{thm:cap-below-half}.
For $p>1/2$ we leave it as a conjecture that our upper bound is the capacity, supported by numerical evidence.
Nevertheless, we are able to use our upper bound in that setting to characterize the asymptotic behavior of $C(p)$ as $p\to 1$.

Our strategy is to ``guess'' that the capacity-achieving output $Y$ must have full support over $\N_{>0}$.
More precisely, with \cref{lem:CPUC-UB-framework} in mind, we will characterize the distributions $Q$ that satisfy
\begin{equation}\label{eq:constraint-Y}
    \KL(P_{Y_x}\|Q)=\lambda x, \quad \textrm{for all $x\in\N_{>0}$},
\end{equation}
for some $\lambda$.
By that lemma, we conclude that all such $\lambda$'s provide an upper bound on $C(p)=\CPUC(\ANMC_{\Ber(p)})$.

Without loss of generality, we write
\begin{equation*}
    Q(y) = 2^{-\lambda(y-p) - h(p) - g(y)}
\end{equation*}
for some $\lambda$ and function $g$.
Note that not all choices of $\lambda$ and $g$ are allowed, since it must be the case that $\sum_{y=1}^\infty Q(y)=1$.
Writing $Q$ in this way is useful because
\begin{align*}
    \KL(P_{Y_x}\|Q) &= (1-p) \log(1/Q(x)) + p \log(1/Q(x+1)) - h(p)\\
    &= \lambda x +(1-p)g(x) + p g(x+1).
\end{align*}
Therefore, \cref{eq:constraint-Y} is equivalent to $g$ satisfying
\begin{equation*}
    (1-p)g(x) + p g(x+1)=0
\end{equation*}
for all $x\geq 1$, which is equivalent to $g(x+1) = -\frac{1-p}{p}\cdot g(x)$ for all $x\geq 1$.
By induction, we conclude that for \cref{eq:constraint-Y} to hold it must be that
\begin{equation*}
    g(x) = \left(-\frac{1-p}{p}\right)^{x-1} g(1)
\end{equation*}
for all $x\geq 1$.

We summarize this discussion in the following theorem, writing $\beta=g(1)$.
\begin{theorem}\label{thm:genbound}
    Fix $p\in[0,1]$ and $\lambda,\beta\in\R$.
    Suppose that
    \begin{equation*}
        Q(y) = 2^{-\lambda(y-p)-h(p)-\left(-\frac{1-p}{p}\right)^{y-1} \beta}
    \end{equation*}
    satisfies $\sum_{y=1}^\infty Q(y)=1$.
    Then, $C(p)\leq \lambda$.
    Furthermore, if there exists an input distribution $P_X$ that yields $Q$ as the output distribution over the $\ANMC_{\Ber(p)}$, then $C(p)=\lambda$.
\end{theorem}

\subsection{The $p\leq 1/2$ setting}\label{sec:p-leq-half}

We will now use \cref{thm:genbound} to determine $C(p)$ for all $p\leq 1/2$, leading to \cref{thm:cap-below-half}, which we restate here.

\capbelowhalf*

\begin{proof}
    The part about $p\geq 1/\varphi^2$ was already established in \cref{thm:opt-rlc}, so we focus on the setting where $p<1/\varphi^2$.

    Note that when $p<1/2$ the quantity $\left(\frac{1-p}{p}\right)^{n-1}$ grows exponentially with $n$.
    Therefore, in this setting, the only choice of $\beta$ that makes $Q$ in \cref{thm:genbound} a valid probability mass function is $\beta=0$, and so we focus on
    \begin{equation}\label{eq:simple-PY}
        Q(y) = 2^{-\lambda(y-p) - h(p)}.
    \end{equation}
    Since it must be the case that $\sum_{y=1}^\infty Q(y)=1$, we conclude that $\lambda$ must be the (unique) solution to
    \begin{equation}\label{eq:lambda-constraint}
        \sum_{y=1}^\infty Q(y)=\frac{2^{\lambda\cdot p-h(p)}}{2^\lambda -1} =1,
    \end{equation}
    in which case $C(p)\leq \lambda$.
    
    By \cref{thm:genbound}, to prove that $C(p)=\lambda$ it suffices to show that $Q$ with $\beta=0$ and $\lambda$ the unique solution to \cref{eq:lambda-constraint} can be realized as an output distribution of the $\ANMC_{\Ber(p)}$.
    This holds if there exists $X$ such that we can write $Q$ as the distribution of $X+Z$
    with $Z$ an independent $\Ber(p)$ random variable.
    We can use this to obtain a candidate generating function $G_X$ for $P_X$.
    This candidate generating function is not guaranteed to be a \emph{probability generating function}, in the sense that its terms need not all be non-negative.
    We then show that this is indeed the case, and so conclude that the desired $X$ exists.

    Let $G$ denote the probability generating function of $Q$.
    For $Z\sim \Ber(p)$ independent of $X$ and $w\in\C$ we have
    \begin{equation*}
        G(w) = \E[w^Y]= \E[w^{X+Z}]=\E[w^X]\cdot \E[w^Z] = G_X(w)(1-p+pw).
    \end{equation*}
    Therefore, we must have
    \begin{equation*}
        G_X(w) = \frac{G(w)}{1-p+pw}
    \end{equation*}
    in the disk $\{w\in\C:|w|<1\}$ (recall that $p\leq 1/2$).

    Recalling \cref{eq:simple-PY}, it is not hard to see that $G$ admits the closed-form
    \begin{equation*}
        G(w) = \frac{2^{\lambda p-h(p)}w}{2^\lambda-w},
    \end{equation*}
    and so
    \begin{equation}\label{eq:pgf-x}
        G_X(w) = \frac{2^{\lambda p-h(p)}w}{(1-p+pw)(2^\lambda-w)}.
    \end{equation}
    Note that $G_X(w)$ is analytic in the disk $|w|<1$.
    To see that $P_X$ is a valid probability distribution it is enough to show that the coefficients of Taylor expansion of $G_X(z)$ around $w=0$ are all non-negative and that the order-$0$ coefficient is $0$, since we know that $\sum_{x\geq 1}P_X(x)=G_X(1)=\frac{G(1)}{1-p+p}=1$.
    To simplify the exposition, we ignore the $2^{\lambda p-h(p)}$ multiplicative factor in \cref{eq:pgf-x} as it does not depend on $x$ and is always positive.
    In other words, we have to show that the coefficients of the Taylor expansion of
    \begin{equation*}
        \frac{w}{(1-p+pw)(2^\lambda-w)}
    \end{equation*}
    around $w=0$ are always positive.
    It is not hard to see by induction that the order-$n$ coefficient of the Taylor expansion of $\frac{w}{(1-p+pw)(2^\lambda-w)}$ around $w=0$ is
    \begin{equation*}
        \alpha_n=\sum_{i=0}^{n-1} \frac{(-p)^i}{2^{\lambda(n-i)}(1-p)^{i+1}} = \frac{1}{2^{\lambda n}(1-p)} \sum_{i=0}^{n-1} \left(-\frac{p2^{\lambda}}{1-p}\right)^i = \frac{1-\left(-\frac{p 2^\lambda}{1-p}\right)^n}{2^{\lambda n}(1-p+p2^\lambda)}.
    \end{equation*} 
    Note that $\alpha_0=0$ and, for $n\geq 1$, we have $\alpha_n\geq 0$ (and hence $P_X(n)\geq 0$) if and only if
    \begin{equation*}
        1\geq \left(-\frac{p 2^\lambda}{1-p}\right)^n.
    \end{equation*}
    This is clearly true when $n$ is odd, because the right-hand side is always non-positive in that case.
    When $n$ is even, this is equivalent to
    \begin{equation*}
        0\leq\lambda \leq \log\left(\frac{1-p}{p}\right).
    \end{equation*}
To see this, recall that $\lambda$ is the unique solution to \cref{eq:lambda-constraint}.
Since $\lambda\mapsto\sum_{y\geq 1}Q(y)$ is a decreasing function of $\lambda>0$, to verify \cref{eq:lambda-constraint} it suffices to check that replacing $\lambda$ by $\log\left(\frac{1-p}{p}\right)$ makes the left-hand side at most $1$, i.e., that
\begin{equation*}
    \frac{2^{\log\left(\frac{1-p}{p}\right) p-h(p)}}{2^{\log\left(\frac{1-p}{p}\right)} -1} = \frac{p(1-p)}{1-2p} \leq 1.
\end{equation*}
This holds if $0\leq p\leq 1/\varphi^2$.
Finally, to get the desired statement in the theorem define $x_p=2^\lambda$ and rewrite \cref{eq:lambda-constraint} in terms of $x_p$.
\end{proof}

\subsection{The $p>1/2$ setting}\label{sec:p-geq-half}

Note that for any fixed $p\in[1/2,1]$ the term $\left(\frac{1-p}{p}\right)^{n-1}$ decreases exponentially fast to $0$ as $n\to\infty$.
Therefore, we get that for every $p\in[1/2,1]$ and $\beta\in\R$ there is a (unique) $\lambda$ such that
\begin{equation*}
    \sum_{y=1}^\infty Q(y) = \sum_{y=1}^\infty 2^{-\lambda(y-p)-h(p)-\left(-\frac{1-p}{p}\right)^{y-1}\beta}=1.
\end{equation*}
Therefore, invoking \cref{thm:genbound} with $p\in[1/2,1]$ directly yields \cref{thm:cap-ub-largep}, which we restate here.

\caplarge*

\subsubsection{Asymptotic behavior of $C(p)$ as $p\to 1$}\label{sec:p-to-1}

In this section we use \cref{thm:cap-ub-largep} to determine the asymptotic behavior of $C(p)$ as $p\to 1$ up to low-order terms, as described in the following theorem.
\begin{theorem}\label{thm:p-to-1}
    We have
    \begin{equation*}
        C(1-\eps)=1-(1+o(1))\frac{h(\eps)}{4},
    \end{equation*}
    where $o(1)\to 0$ as $\eps\to 0$.
\end{theorem}
\begin{proof}
    We prove this theorem in two stages.
    First, we show that for all $\eps>0$ small enough we have $C(1-\eps)\leq 1-\frac{h(\eps)}{4}$.
    Second, we analyze the rate achieved by a geometric input distribution, and show that it is $1-(1+o(1))\frac{h(\eps)}{4}$.
    Together, these results yield the desired theorem statement.

    To get the upper bound on $C(p)$ for $p=1-\eps$ large we instantiate \cref{thm:cap-ub-largep} with $\beta =\frac{p h(p)}{1-p}$.
    consider the function $f(\alpha,p)$ defined as
\begin{equation*}
    f(\lambda,p)=\sum_{n=1}^\infty 2^{-\lambda(n-p)-h(p)-g(n)} = \sum_{n=1}^\infty 2^{-\lambda(n-p)-h(p)-\left(-\frac{1-p}{p}\right)^{n-1}\frac{p h(p)}{1-p}}.
\end{equation*}
From \cref{thm:cap-ub-largep}, we know that for each $p\in(1/2,1)$ the unique $\lambda_p> 0$ such that $f(\lambda_p,p)=1$ is an upper bound on $C(p)$.
Furthermore, with $p$ fixed the map $\lambda\mapsto f(\lambda,p)$ is decreasing in $\lambda$.
Our goal is to show that
\begin{equation}\label{eq:ub-alphap}
    \lambda_p \leq 1-\frac{h(p)}{4}
\end{equation}
for all sufficiently large $p<1$, which implies the same upper bound on $C(p)$ via the discussion above.
We will derive \cref{eq:ub-alphap} in two steps:
\begin{enumerate}
    \item We upper bound $f(\lambda,p)$ by a more tractable function $\tilde{f}(\lambda,p)$.

    \item We show that  $\tilde{f}(1-h(p)/4,p)\leq 1$ for all sufficiently large $p<1$.
    This implies that $f(1-h(p)/4,p)\leq 1$ for all sufficiently large $p<1$ too.
    Then, since $\lambda\mapsto f(\lambda,p)$ is decreasing in $\lambda$, we can conclude that $\lambda_p\leq 1-h(p)/4$ for all sufficiently large $p<1$.
\end{enumerate}

We obtain $\tilde{f}$ via the following claim.
\begin{claim}
    For all $\lambda>0$ and $p\in (1/2,1)$ we have that
    \begin{equation*}
        f(\lambda,p)\leq 2^{-\lambda(1-p)-\frac{h(p)}{1-p}}+2^{-\lambda(2-p)} + 2^{-\lambda(3-p) - \frac{h(p)}{p}}+\frac{2^{-\lambda(3-p)-h(p)+\left(\frac{1-p}{p}\right)^2 h(p)}}{2^\lambda-1} =: \tilde{f}(\lambda,p).
    \end{equation*}
\end{claim}
\begin{proof}
    We may write
    \begin{equation}\label{eq:f-first-terms}
        f(\lambda,p)=2^{-\lambda(1-p)-\frac{h(p)}{1-p}}+2^{-\lambda(2-p)} + 2^{-\lambda(3-p) - \frac{h(p)}{p}} + \sum_{n=4}^\infty 2^{-\lambda(n-p)-h(p)-\left(-\frac{1-p}{p}\right)^{n-1}\frac{p h(p)}{1-p}}.
    \end{equation}
    Then, note that $\left(\frac{1-p}{p}\right)^{n-1}$ is decreasing in $n$ for fixed $p>1/2$.
    Therefore,
    \begin{equation*}
        2^{-\lambda(n-p)-h(p)-\left(-\frac{1-p}{p}\right)^{n-1}\frac{p h(p)}{1-p}} \leq 2^{-\lambda(n-p)-h(p)+\left(\frac{1-p}{p}\right)^{2}h(p)}
    \end{equation*}
    for all $n\geq 4$,
    and so
    \begin{equation}\label{eq:f-rem-terms}
        \sum_{n=4}^\infty 2^{-\lambda(n-p)-h(p)-\left(-\frac{1-p}{p}\right)^{n-1}\frac{p h(p)}{1-p}} \leq \sum_{n=4}^\infty 2^{-\lambda(n-p)-h(p)+\left(\frac{1-p}{p}\right)^{2}h(p)} = \frac{2^{-\lambda(3-p)-h(p)+\left(\frac{1-p}{p}\right)^2 h(p)}}{2^\lambda-1}.
    \end{equation}
    Combining \cref{eq:f-first-terms,eq:f-rem-terms} yields the desired upper bound.
\end{proof}

Define $\gamma(p)=\tilde{f}(1-h(p)/4,p)$.
Then, it is not hard to check that $\gamma'(1)=\ln 2 - 1/e$, and so
\begin{equation*}
    \gamma(p)= 1-(\ln 2 - 1/e)(1-p)+O((1-p)^2).
\end{equation*}
Since $\ln 2>1/e$, it follows that there exists $p_0<1$ such that $\gamma(p)=\tilde{f}(1-h(p)/4,p)<1$ for all $p\in(p_0,1)$.
As discussed above, this implies that
\begin{equation}\label{eq:ub-largep}
    C(p)\leq \lambda_p \leq 1-h(p)/4
\end{equation}
for all $p\in(p_0,1)$.

We now turn to lower bounding $C(p)$ for large $p$.
When $p=1$ the capacity-achieving input distribution $X$ is geometric with $P_X(x)=2^{-x}$.
Therefore, we should expect that this input distribution gives a fairly sharp lower bound when $p\to 1$.
Note that $\E[X]=2$, and so the achievable rate is
\begin{equation*}
    \frac{I(X;Y)}{\E[X]}= \frac{H(Y)-H(Y|X)}{2}=\frac{H(Y)-h(p)}{2}.
\end{equation*}
where $Y$ is the output satisfying $P_Y(1)=(1-p)P_X(1)=\frac{1-p}{2}$ and, for $n\geq 2$,
\begin{equation*}
    P_Y(y) = p \cdot P_X(y-1)+(1-p) P_X(y) = p 2^{-n+1}+(1-p)2^{-n} = 2^{-n}(1+p).
\end{equation*}
Therefore,
\begin{align*}
    H(Y) &= \frac{1-p}{2}\log\left(\frac{2}{1-p}\right)+\sum_{y\geq 2} P_Y(y)\log(1/P_Y(y)) \\
    &=\frac{1-p}{2}\log\left(\frac{2}{1-p}\right)+(1+p)\sum_{y\geq 2} 2^{-y}(y-\log(1+p))\\
    &= \frac{1-p}{2}\log\left(\frac{2}{1-p}\right)+(1+p)\left(\frac{3}{2}-\frac{\log(1+p)}{2}\right)\\
    &= 1+p + h\left(\frac{1-p}{2}\right).
\end{align*}
Consequently, we have the lower bound
\begin{equation}\label{eq:lb-geom2}
    C(p) \geq \frac{ 1+p + h\left(\frac{1-p}{2}\right)-h(p)}{2} = 1-(1+o(1))\frac{h(p)}{4}.
\end{equation}
The theorem statement follows by combining \cref{eq:ub-largep,eq:lb-geom2} and noting that for $p=1-\eps$ we have $h(p)=h(\eps)$.
\end{proof}

A plot of the closed-form upper and lower bounds obtained in the proof of \cref{thm:p-to-1} are given in \cref{fig:p-1-ublb} together with a numerical approximation of $C(p)$.

\subsubsection{Asymptotic behavior of $C(p)$ as $p\to 1/2$}\label{sec:p-to-half}

In this section we study the asymptotic behavior of $C(p)$ in the regime where $p>1/2$ and $p\to 1/2$. First we present an upper bound on the capacity.

\begin{theorem} \label{thm:UP-p-half}
    For $\eps\rightarrow 0$ from above, we have
    \begin{equation*}
        C\left(\frac{1}{2}+\eps\right) \leq \log\varphi\cdot \left(1+\frac{2\eps}{2+\varphi}\right)+O\left(\eps^2\right).
    \end{equation*}
\end{theorem}
\begin{proof}
Let $p=\frac{1}{2}+\eps$. To get the upper bound on $C\left(\frac{1}{2}+\eps\right)$, we instantiate \cref{thm:cap-ub-largep} with a careful choice of $\beta>0$. 
First, we upper bound $Q(n)$ by a more tractable expression. Note that for all odd $y\geq 1$ we have
\begin{align}
Q(y) &=2^{-\lambda(y-p) - h(p) - \left(\frac{1-p}{p}\right)^{y-1}\beta}\nonumber\\
&\leq 2^{-\lambda(y-p) - h(p) - \left(1-4\eps\right)^{y-1}\beta}\nonumber\\
&\leq 2^{-\lambda(y-p) - h(p)-\beta + 4(y-1)\eps\beta}\nonumber,
\end{align}
while for all even $y\geq 2$ we have
\begin{align}
Q(y)&=2^{-\lambda(y-p) - h(p) + \left(\frac{1-p}{p}\right)^{y-1}\beta}\nonumber\\
&\leq 2^{-\lambda(y-p)- h(p)-\beta +\frac{\beta}{p}}\nonumber.
\end{align}
Therefore
\begin{align}
   1=\sum_{y=1}^\infty Q(y)&\leq \sum_{y=1}^\infty 2^{-\lambda(2y-1-p) - h(p)-\beta + 4(2y-2)\eps\beta} + \sum_{y=1}^\infty 2^{-\lambda(2y-p) - h(p)-\beta +\frac{\beta}{p}}\nonumber\\
    \label{upper_bound_of_1}
    &= 2^{\lambda p - h(p)-\beta}\left( \frac{2^{\lambda -8\eps\beta}}{2^{2\lambda -8\eps\beta} - 1} +  \frac{2^{\frac{\beta}{p}}}{2^{2\lambda} - 1}\right).
\end{align}
We will now take
\begin{align}
\label{choice_g_1}
    \beta &= 1+\frac{\log \varphi}{2}-h(p).
\end{align}
To motivate this choice of $\beta$, note that for $\eps=0$ the upper bound in \cref{upper_bound_of_1} is minimized by the choice in \cref{choice_g_1}.
Now, denote
\begin{align}
\label{choice_lambda_p}
    \lambda(p)=\log \varphi + \Delta(p).
\end{align}
for some function $\Delta$  satisfying $\Delta(1/2) = 0$, and denote
\begin{equation*}
    f(p)= 2^{\lambda p - h(p)-\beta}\left( \frac{2^{\lambda -8\eps\beta}}{2^{2\lambda -8\eps\beta} - 1} +  \frac{2^{\frac{\beta}{p}}}{2^{2\lambda} - 1}\right).
\end{equation*}
Note that if $f(p)\leq 1$, then $\lambda\leq \lambda(p)$ for that choice of $p$ because
the right-hand side of \cref{upper_bound_of_1} is decreasing with $\lambda$ and greater than $1$.
Since we are interested in the behavior around $p=1/2$, we focus on the corresponding Taylor expansion
\begin{align*}
    f(p)
    &=1+\left(p-\frac{1}{2}\right) \Bigg(\left(1-2\varphi\right)\Delta'\left(\frac{1}{2}\right)+ 2(\varphi-1)\log \varphi\Bigg)\ln 2+O\left(\left(p-\frac{1}{2}\right)^2\right).
\end{align*}
This means that $f(p)\leq 1$ for $p$ sufficiently close to $1/2$ provided that
\begin{align*}
    \Delta'\left(\frac{1}{2}\right)\geq  \frac{2(\varphi-1)}{2\varphi-1} \log \varphi= \frac{2}{\varphi+2}\log \varphi.
\end{align*}
   Now, note that
   \begin{equation}
   \label{choice_delta_p}
        \Delta(p) = \left(\frac{2}{\varphi+2} \log \varphi+\delta\right)\left(p-\frac{1}{2}\right)
    \end{equation}
    satisfies $\Delta(1/2)=0$ and $\Delta'(1/2)\geq2\log \varphi/(\varphi+2)$ for any $\delta>0$.
    Injecting \cref{choice_delta_p} in \cref{choice_lambda_p}, we conclude that
    \begin{equation*}
        \lambda \leq \log \varphi \cdot \left(1 + \left(\frac{2}{\varphi+2}+\delta\right)\left(p-\frac{1}{2}\right)\right) +O\left(\left(p-\frac{1}{2}\right)^2\right),
    \end{equation*}
    for any $\delta>0$, provided that $p$ is sufficiently close to $1/2$ (depending on $\eps$).
\end{proof}

To complement the asymptotic upper bound above, we present an asymptotic lower bound on the capacity in the regime $p\rightarrow 1/2$.

\begin{theorem} \label{prop:LB-p-half}
    For $\eps>0$ we have
    \begin{equation*}
    C\left(\frac{1}{2}+\eps\right)\geq \log \varphi \cdot \left( 1 + \log \varphi \cdot  \frac{4(3-\varphi)\ln 2}{5\sqrt{5}(\varphi-1)}\eps^{2} \right)  - O\left(\eps^{3}\right).
    \end{equation*}
\end{theorem}

\begin{proof}
Fix $0<\eps\leq 1/2$ and define $p=1/2+\eps$.
Let $0<\gamma\leq 2$, to be chosen later.
We consider a mixture $P_X$ with parameter $\gamma \epsilon$ of two distributions $P_{X_1}$ and $P_{X_2}$, given by
\begin{equation*}
    P_X=(1-\gamma\eps)P_{X_1}+\gamma\eps P_{X_2}.
\end{equation*}
We choose the first distribution of the mixture to be the naive RLC coding input distribution
\begin{equation*}
    P_{X_1}(x) = \begin{cases}
        \varphi^{-x}, &\textrm{if $x$ is odd,}\\
        0, &\textrm{otherwise,}
    \end{cases}
\end{equation*}
which is capacity achieving for $1/\varphi^2\leq p\leq 1/2$ 
(see \cref{sec:opt-rlc}) and induces the output distribution
\begin{equation*}
    P_{Y_1}(y) = \begin{cases}
        (1-p)\varphi^{-y}, &\textrm{if $y$ is odd,}\\
        p \varphi ^{-(y-1)}, &\textrm{otherwise.}
    \end{cases}
\end{equation*}
The second input distribution of the mixture is
\begin{equation*}
    P_{X_2}(x) = \begin{cases}
        0 , &\textrm{if $x$ is odd,} \\
        2^{-\frac{x}{2} } &\textrm{otherwise,}
    \end{cases}
\end{equation*}
and induces the output distribution
\begin{equation*}
    P_{Y_2}(y) = \begin{cases}
        0, &\textrm{if $y=1$,}\\
        p 2^{-\frac{y-1}{2}}, &\textrm{if $y$ is odd, $y\geq3$,}\\
         (1-p)2^{-\frac{y}{2}}, &\textrm{otherwise.}
    \end{cases}
\end{equation*}
Recall that $p=\frac{1}{2}+\eps$, a Taylor expansion of the mutual information around $\eps=0$ gives 
\begin{align}
\label{Taylor_expansion_mutual_information_regime_near_one_half}
I(X;Y)
&= \frac{(\gamma \eps)^2}{2\ln 2} + (1-\gamma \eps)I(X_1;Y_1) - \gamma \eps h(p)
- \gamma \eps \sum_{y\geq 1} P_{Y_2}(y) \log P_{Y_1}(y)  \nonumber\\
&~~~- \frac{1}{2\ln 2}\frac{(\gamma \eps)^2}{1-\gamma \eps}\sum_{y\geq 1}  \frac{P_{Y_2}(y)^2}{P_{Y_1}(y)}-O(\eps^3) .
\end{align}
We now expand each term in \cref{Taylor_expansion_mutual_information_regime_near_one_half} asymptotically.
First, we have
\begin{align*}
    h(p)&= 1 -O(\eps^2).
\end{align*}
Second, we have
\begin{align}
\sum_{y\geq 1} P_{Y_2}(y) \log P_{Y_1}(y)
\label{weight_input}
&= - \sum_{y\geq 1} P_{Y_2}(y) - \log \varphi \sum_{y\geq 1} P_{Y_2}(y) y + \log \varphi \sum_{\text{even } y, y\geq 2} P_{Y_2}(y)\nonumber\\
&~~~+ 2\eps \sum_{\text{even } y, y\geq 2} P_{Y_2}(y) - 2\eps \sum_{\text{odd } y, y\geq 1} P_{Y_2}(y) -O(\eps^2)\nonumber\\
&=-1 -4\log \varphi-2\log \varphi \cdot \eps - O(\eps^2).
\end{align}
Third, we have
\begin{align*}
    \sum_{y\geq 1}  \frac{P_{Y_2}(y)^2}{P_{Y_1}(y)}
    &=\frac{3\varphi+1}{2(3-\varphi)}+O(\eps).
\end{align*}
The reasoning behind this particular mixture choice is the following. In order to maximize the right-hand side of \cref{Taylor_expansion_mutual_information_regime_near_one_half}, we want to maximize the expectation of $Y_2$ while minimizing $\sum_{\text{even } y, y\geq 2} P_{Y_2}(y)$ in \cref{weight_input}. Since $\mathbb{P}[Y_2 \text{ is even}]=(1-p)\mathbb{P}[X_2 \text{ is even}]+p\cdot \mathbb{P}[X_2 \text{ is odd}]= 1-p +(2p-1) \mathbb{P}[X_2 \text{ is odd}]$ and $p>\frac{1}{2}$, we choose a distribution $P_{X_2}$ with mass only on even numbers.

Recalling that the RLL achievable rate per unit cost is $\log \varphi$ (see \cref{sec:opt-rlc}), we obtain from \cref{Taylor_expansion_mutual_information_regime_near_one_half} that
\begin{align*}
&I(X;Y)\\
&~~~= (1-\gamma \eps)\log \varphi \cdot \E[X_1]-\gamma \eps -\gamma \eps \left( -1 - 4 \log \varphi -2\log \varphi \cdot \eps\right) +\frac{(\gamma \eps)^2}{2\ln 2} \left(1-\frac{3\varphi+1}{2(3-\varphi)}\right) +O(\eps^3)\\
&~~~=(1-\gamma \eps)\sqrt{5} \log \varphi+  4 \gamma \log \varphi \cdot \eps  +(\gamma \eps)^2\left(\frac{5(1-\varphi)}{4(3-\varphi)\ln 2}  + \frac{2\log \varphi}{\gamma} \right)+O(\eps^3),
\end{align*}
where
\begin{align*}
    \E[X_1]&= \sqrt{5}.
\end{align*}
Then, noting that $\E[X]=\sqrt{5}\left(1+\gamma \eps\frac{4-\sqrt{5}}{\sqrt{5}}\right)$,
the achievable rate per unit cost of the mixture is
\begin{align}
\label{achievable_rate_per_unit_cost_regime_near_one_half}
\frac{I(X;Y)}{\E[X]}
&= \log \varphi +\frac{(\gamma \eps)^2}{\sqrt{5}}\left(\frac{5(1-\varphi)}{4(3-\varphi)\ln 2} + 2 \frac{\log \varphi}{\gamma}\right) -O(\eps^3).
\end{align}
It remains to optimize the choice of $\gamma$. 
The optimal choice is $\gamma^*=\frac{4(3-\varphi)\ln 2 \cdot \log \varphi}{5(\varphi-1)}$, which satisfies $\gamma^*\leq 2$ ensuring $\gamma\eps\leq1$ for $0\leq\eps\leq1/2$. Plugging this choice into \cref{achievable_rate_per_unit_cost_regime_near_one_half} leads to the desired statement.
\end{proof}

A plot of the closed-form upper and lower bounds obtained in \cref{thm:UP-p-half} and \cref{prop:LB-p-half} together with a numerical approximation of $C(p)$ is given in \cref{fig:p-half-ublb}.

\subsection{Numerical approximation of $C(p)$ and a conjecture}\label{sec:numerical}

While \cref{thm:cap-below-half} determines $C(p)$ for all $p\leq 1/2$, \cref{thm:cap-ub-largep} provides only an upper bound on $C(p)$ for $p>1/2$.
In this section we discuss a numerical approximation of $C(p)$, which in particular provides evidence that the upper bound from \cref{thm:cap-ub-largep} is the capacity.
More precisely, we formulate the following conjecture.
\begin{conjecture}\label{conj:cap-largep}
    Fix $p\in(1/2,1]$.
    For $\beta\in\R$, denote by $\lambda(p,\beta)$ the unique solution to the equation
    \begin{equation*}
        \sum_{n=1}^\infty 2^{-\lambda(p,\beta)(n-p)-h(p)-\left(-\frac{1-p}{p}\right)^{n-1} \beta}=1.
    \end{equation*}
    Then, $C(p)= \inf_{\beta\in\R}\lambda(p,\beta)$.
\end{conjecture}

Towards numerically approximating $C(p)$, we leverage \cref{lem:ANSC-to-ANMC} and employ a variant of the well-known Blahut–Arimoto algorithm (BAA) \cite{blahut1972,Arimoto72}, specifically the version developed by Jimbo and Kunisawa \cite{jimbo1979iteration}.
Recall that the $\ANMC_{\Ber(p)}$ takes as input a natural number $n>0$ and outputs $n + Z$ with $Z$ following a Bernoulli distribution with success probability $p$. In what follows, we numerically approximate the capacity per unit cost by studying a finite-input, finite-output truncation of this channel.

Namely, we will numerically approximate the capacity of the channel $\ANMC_{\Ber(p),\tau}$ which gets as input only $n\in[\tau]$ and outputs $n + \Ber(p)$. This channel and its capacity is discussed in more detail in \cref{sec:cap-rlc}. 
For the purpose of this section, we observe that this channel is ``equivalent'' to the channel $\ANSC_{\Ber(p), \tau}$ which gets as input $x\in \{0,1\}^n$ having runs of length at most $\tau$ and applies $\ANSC_{\Ber(p)}$ to the input. Formally, if we denote by $C(p,t)$ the capacity of $\ANSC_{\Ber(p), \tau}$, then, by \cref{lem:ANSC-to-ANMC} (whose proof can be found in the appendix), we have that 
\[
C(p,\tau) = \CPUC(\ANMC_{\Ber(p),\tau}) = \sup_{X:\E[X]<\infty} \frac{I(X;Y)}{\E[X]},
\]
where the supremum is taken over all input random variables supported on $[\tau]$, with finite expectation, and $Y$ is the output of $\ANMC_{\Ber(p),\tau}$ on input $X$.

In general, when a channel is a DMC with finite input and output alphabets, one can apply the Jimbo-Kunisawa algorithm~\cite{jimbo1979iteration} to numerically approximate the capacity per unit cost.
This algorithm gets as input the channel's transition law $P(y|x)$, which specifies, for every $x\in\calX$ and $y\in\calY$, the probability that the channel outputs $y\in\calY$ when $x\in\calX$ is transmitted.
The algorithm is initialized with some input distribution $X^{(0)}$ supported on the input alphabet, and at each iteration this distribution is updated. It is known that the resulting sequence of distributions $X^{(0)}, X^{(1)}, X^{(2)}, \ldots$ converges to an input distribution that attains the capacity per unit cost of the corresponding DMC~\cite[Theorem 1]{jimbo1979iteration}.
In our implementation, we also give as input a stopping criteria which will govern the error term to capacity. We summarize it in the following theorem.

\begin{theorem}[{\cite[Theorem 1, Corollary 2]{jimbo1979iteration}}] \label{thm:jk-alg}
    Let $a>0$, $p\in [0,1]$, and $\tau\in \N$. There exists an iterative algorithm that outputs a distribution $P_X$ supported on $[\tau]$ such that
    \[
    \frac{I(X;Y)}{\E[X]} \leq \capBASfinite \leq \frac{I(X;Y)}{\E[X]} + a \;.
    \]
    Moreover, the number of iterations is $O(1/a)$.
\end{theorem}

Our next goal is to show that $C(p,\tau)$ closely approximates $C(p)$, particularly when $\tau$ is large. To establish this, we will rely on the following simple lemma and a claim.

We will make use of a standard result stating that among all random variables on $N$ with a given finite mean, the one with maximal entropy is the geometric random variable having that same mean \cite[Theorem 5.8]{conrad2004probability}.

\begin{lemma} \label{lem:max-ent-geom} 
    Let $X$ be a random variable sampled from a distribution supported on $\{1,2,3,\ldots\}$ with finite mean.
    Then, $H(X) \leq \log(e \cdot \E[X])$.
\end{lemma}
\begin{proof}
    By \cite[Theorem 5.8]{conrad2004probability}, $H(X) \leq H(Y)$ where $Y\sim \textrm{Geom}(1/\E[X])$, and we have 
    \[
    H(Y) = -\log \left(\E[X]\right) - (\E[X] - 1) \log \left( \frac{\E[X]}{\E[X] - 1} \right) \leq \log (e \cdot \E[X]) \;. \qedhere
    \]
\end{proof}
Next, we show that any distribution which achieves capacity on the $\ANMC_{\Ber(p)}$ must have ``small'' expectation.
\begin{claim} \label{clm:expectation-bound-capacity}
    Let $P_X$ be a distribution which achieves capacity on $\ANMC_{\Ber(p)}$. It holds that $\E[X] < 6$.
\end{claim}
\begin{proof}
    Assume that $\E[X] \geq 6$. Since $P_X$ achieves capacity,
    \[
    C(p) = \frac{I(X;Y)}{\E[X]} \leq \frac{H(X)}{\E[X]} \leq \frac{\log(e \E[X])}{\E[X]} \leq \frac{\log(6e)}{6} \leq 0.672 \;,
    \]
    where the second inequality is due to \cref{lem:max-ent-geom} and the third inequality follows since $\log(ex)/x$ is decreasing in the range $x \geq 1$.  
    This contradicts the fact that the capacity of this channel is at least $\log \varphi \geq 0.697$ for every $p\in [0,1]$.
\end{proof}

Using these two elementary bounds and the simple claim above, we now show that  $\capBASfinite$ provides a fairly tight approximation to $\capBAS$. In particular, the approximation error is on the order of $O(\frac{\log \tau}{\tau})$.
\begin{theorem} \label{prop:numer-approx}
    For every $\tau \geq 1$, we have that 
    \[
    \capBASfinite \leq \capBAS \leq \capBASfinite + \frac{7\log (e(\tau+1) ) - 11}{\tau + 1} \;.
    \]
\end{theorem}
\begin{proof}
    The lower bound is immediate, since any distribution supported on $[\tau]$ is also a valid distribution on $\N$. Consequently, the supremum taken over all finite-mean distributions supported on $\N$ is at least as large as the supremum taken over finite-mean distributions supported on $[\tau]$.

    For the upper bound, let $P_X$ be a distribution which achieves capacity on $\ANMC_{\Ber(p)}$ and let $X\sim P_X$ be a random variable. Denote $A = \mathds{1} \{X > \tau \}$, i.e., $A$ is the event that $X$ is larger than $\tau $ and denote $\alpha = \Pr[A]$. 
    Since $A$ is determined by $X$, we have that $I(X;Y) = I(A,X;Y)$. Now, by the chain rule, we have
    \begin{align*}
        I(A,X;Y) &= I(A; Y) + I(X;Y|A) \\
        &= I(A;Y) + \alpha I(X;Y|A=1) + (1 - \alpha)I(X;Y|A=0) \;.
    \end{align*}
    Thus, we have
    \begin{equation} \label{eq:cap-upper-bound-split}
        C(p) \leq \frac{I(A;Y) + \alpha I(X;Y|A=1) + (1 - \alpha)I(X;Y|A=0)}{\E[X]} \;.
    \end{equation}
    We will upper bound each term separately.
    For the first term we have
    \begin{align}
        \frac{I(A;Y)}{\E[X]} \leq h(\alpha) \leq h \left( \frac{6}{\tau + 1}\right) \leq \frac{6 \left( \log(e(\tau+1)) - 2\right)}{\tau + 1} \;. \label{eq:bound-I-Y-A}
    \end{align}
    Here, the first inequality follows since $I(A;Y)\leq H(A) = h(\alpha)$ and we have that $\E[X] \geq 1$. The second inequality uses the Markov inequality which gives that $\alpha = \Pr[X \geq \tau  +1] \leq \frac{\E[X]}{\tau +1} \leq \frac{6}{\tau + 1}$ where the last inequality is due to \cref{clm:expectation-bound-capacity}. The last inequality is by the well-known upper bound $h(x) \leq x\log(e/x)$. 

    Now, the second term can be upper bounded as follows
    \begin{align}
        \frac{\alpha \cdot I(X; Y|A=1)}{\E[X]} &\leq \frac{H(Y|A=1)}{\E[X|A = 1]} \nonumber \\
        &\leq \frac{H(X|A=1) + h(p)}{\E[X|A=1]} \nonumber \\
        &\leq \frac{\log (e (\E[X|A=1]))+h(p)}{\E[X|A=1]} \nonumber \\
        &\leq \frac{\log(e (\tau + 1)) + h(p)}{\tau+1} \label{eq:bound-A-1}\;.
    \end{align}
    The first inequality follows since 
    \[
        \frac{\alpha}{\E[X]} = \frac{\alpha}{(1 - \alpha)\cdot \E[X|A=0] + \alpha \cdot \E[X|A=1]} \leq \frac{1}{\E[X|A=1]} \;.
    \]
    The second inequality follows by recalling that $Y = X + \Ber(p)$ and that $H(X+\Ber(p)) \leq H(X) + H(\Ber(p))$. 
    The third inequality is due to \cref{lem:max-ent-geom} and the fourth inequality follows since $\log(ex)/x$ is a decreasing function for $x\geq 1$ and $\E[X|A=1]\geq \tau+1$.
    
    Finally, the third term can be upper bounded by
    \begin{align}
        \frac{(1 - \alpha)\cdot I(X; Y|A=0)}{\E[X]} \leq \frac{(1 - \alpha) \E[X|A=0] \cdot C(p,\tau)}{\E[X]} \leq C(p,\tau) \;. \label{eq:bound-A-0}
    \end{align}
    The first inequality follows since the distribution of $(X|A=0)$ is supported on $[\tau ]$ and thus $I(X; Y |A=0)/ \E[X|A=0]\leq C(p,\tau)$. The second inequality follows since
    \[
    \frac{(1-\alpha)\cdot \E[X|A=0]}{\E[X]} = \frac{(1-\alpha)\sum_{x\in [\tau ]} x\cdot \frac{P_X(x)}{(1- \alpha)}}{\E[X]} = \frac{\sum_{x\in [\tau ]} x\cdot P_X(x)}{\E[X]}\leq 1\;.
    \] 
    Plugging \cref{eq:bound-I-Y-A,eq:bound-A-1,eq:bound-A-0} into \cref{eq:cap-upper-bound-split} yields the claimed proposition.
\end{proof}

Using Proposition~\ref{prop:numer-approx}, we can apply the Jimbo–Kunisawa algorithm (Theorem~\ref{thm:jk-alg}) to obtain numerical estimates of $C(p)$ and then compare these values with those predicted by our Conjecture~\ref{conj:cap-largep}. 
Specifically, we run the algorithm from Theorem~\ref{thm:jk-alg} with parameters $\tau = 2^{16}$ and $a = 10^{-3}$. This produces a lower bound on $C(p, 2^{16})$ that also serves as a tight approximation of the capacity, with at most $10^{-3}$ additive error. A plot of this lower bound 
is given in \cref{fig:numerical}. 

Moreover, when combined with Proposition~\ref{prop:numer-approx}, this yields a numerical approximation of $C(p)$ whose additive error is at most $3 \cdot 10^{-3}$. For several values of $p$, Table~\ref{table:approx-conj} compares these capacity approximations with our conjectured values. 
Furthermore, \cref{fig:combined} presents a comparison between the approximated capacity and the asymptotic upper and lower bounds obtained in previous sections, in the regimes where $p\approx 1/2$ and $p\approx 1$.

\begin{table}[h!]
\centering
\begin{tabular}{|c||c|c|}
\hline
$p$ & Approximation & UB (\cref{thm:cap-ub-largep}) \\ \hline
$0.55$ & $0.6970746$ & $0.6970751$ \\ \hline
$0.60$ & $0.7051158$ & $0.7051163$ \\ \hline
$0.65$ & $0.7179734$ & $0.7179744$ \\ \hline
$0.70$ & $0.7355886$ & $0.7355888$ \\ \hline
$0.75$ & $0.7581924$ & $0.7581924$ \\ \hline
\end{tabular}
\quad
\begin{tabular}{|c||c|c|}
\hline
$p$ & Approximation & UB (\cref{thm:cap-ub-largep}) \\ \hline
$0.80$ & $0.7863129$ & $0.7863129$ \\ \hline
$0.85$ & $0.8208797$ & $0.8208797$ \\ \hline
$0.90$ & $0.8635515$ & $0.8635517$ \\ \hline
$0.95$ & $0.9178513$ & $0.9178514$ \\ \hline
$0.99$ & $0.9774431$ & $0.9774431$ \\ \hline
\end{tabular}
\caption{A comparison between our numerical estimates of the capacity $C(p)$ and the values obtained from the upper bound in \cref{thm:cap-ub-largep}, which we conjecture (\cref{conj:cap-largep}) to equal the true capacity, shows consistency to at least five decimal digits.}
\label{table:approx-conj}
\end{table}

\begin{figure}
		\centering	\begin{tikzpicture}
		
		\begin{axis}[
            width = 12cm,
            height=8cm,
		scaled y ticks = true,
		tick label style={/pgf/number format/fixed },
		axis lines = left,
		xlabel = $p$,
		ylabel = {Rate},
		xtick={0.1,0.2,0.3,0.4,0.5,0.6,0.7,0.8,0.9,1},
		ytick={0.7,0.8,0.9,1}, 
        xmin=0, xmax=1,
        ymin=0.65, ymax=1,
		legend pos=south east,
		ymajorgrids=true,
		grid style=dashed,
		]
            \addplot[
		color=blue,
            mark size=0.6pt,
		mark=*,
		]
		coordinates {
                (0.01, 0.9648662676429648)
                (0.02, 0.9399823450657127)
                (0.03, 0.9189989594574902)
                (0.04, 0.9005481866200025)
                (0.05, 0.8839774956197947)
                (0.06, 0.8689018826659844)
                (0.07, 0.8550669401647814)
                (0.08, 0.8422920803783178)
                (0.09, 0.8304426009301901)
                (0.1, 0.8194143037269959)
                (0.11, 0.8091243162827199)
                (0.12, 0.7995052715812978)
                (0.13, 0.7905014379178712)
                (0.14, 0.782066044993075)
                (0.15, 0.7741593774066233)
                (0.16, 0.7667473790198729)
                (0.17, 0.7598006082115885)
                (0.18, 0.7532934406863723)
                (0.19, 0.7472034510328082)
                (0.2, 0.7415109260074572)
                (0.21, 0.7361984766605651)
                (0.22, 0.7312507258378717)
                (0.23, 0.7266540540106576)
                (0.24, 0.7223963908490971)
                (0.25, 0.7184670431147644)
                (0.26, 0.7148565517241482)
                (0.27, 0.7115565724989746)
                (0.28, 0.7085597763527098)
                (0.29, 0.7058597655890525)
                (0.3, 0.7034510036925086)
                (0.31, 0.7013287565327979)
                (0.32, 0.6994890433260432)
                (0.33, 0.6979285960269611)
                (0.34, 0.696644826089528)
                (0.35, 0.6956357977454452)
                (0.36, 0.6949002071220213)
                (0.37, 0.6944373666630861)
                (0.38, 0.6942471770729512)
                (0.39, 0.694241913361114)
                (0.4, 0.694241913628305)
                (0.41, 0.6942419136304968)
                (0.42, 0.6942419136305983)
                (0.43, 0.6942419136306112)
                (0.44, 0.6942419136306129)
                (0.45, 0.6942419136306104)
                (0.46, 0.6942419136306001)
                (0.47, 0.6942419136305242)
                (0.48, 0.6942419136292091)
                (0.49, 0.6942419135302156)
                (0.5, 0.6942414858912936)
                (0.51, 0.69436000567872)
                (0.52, 0.6947090281293751)
                (0.53, 0.6952818789104241)
                (0.54, 0.6960722748449032)
                (0.55, 0.6970746414224845)
                (0.56, 0.6982840670533418)
                (0.57, 0.699696275317943)
                (0.58, 0.7013076025141329)
                (0.59, 0.7031149765048882)
                (0.6, 0.705115895702047)
                (0.61, 0.707308408098575)
                (0.62, 0.7096910906984084)
                (0.63, 0.7122630298561664)
                (0.64, 0.7150238030753789)
                (0.65, 0.717973462791648)
                (0.66, 0.7211125226225602)
                (0.67, 0.7244419465189864)
                (0.68, 0.7279631412142799)
                (0.69, 0.7316779523455814)
                (0.70, 0.7355886646189486)
                (0.71, 0.7396980064106677)
                (0.72, 0.7440091592435745)
                (0.73, 0.7485257726531105)
                (0.74, 0.7532519850678772)
                (0.75, 0.7581924514805006)
                (0.76, 0.7633523788864142)
                (0.77, 0.7687375707346347)
                (0.78, 0.7743544819859631)
                (0.79, 0.7802102868388481)
                (0.8, 0.786312961803251)
                (0.81, 0.7926713876388253)
                (0.82, 0.7992954748157927)
                (0.83, 0.8061963187401816)
                (0.84, 0.813386393217187)
                (0.85, 0.8208797938315323)
                (0.86, 0.828692547621916)
                (0.87, 0.8368430124726565)
                (0.88, 0.8453524004897085)
                (0.89, 0.8542454768129525)
                (0.9, 0.8635515134807049)
                (0.91, 0.8733056259551878)
                (0.92, 0.8835507055395343)
                (0.93, 0.8943403222874498)
                (0.94, 0.9057432982093652)
                (0.95, 0.9178513633009069)
                (0.96, 0.9307930484555111)
                (0.97, 0.9447619046513025)
                (0.98, 0.9600845382693225)
                (0.99, 0.9774431732443992)

        }; 
        
        \end{axis}
	\end{tikzpicture}
	\caption{A numerical lower bound on the capacity of $\capBASfinite$ for $\tau=2^{16}$, computed using the Jimbo–Kunisawa algorithm (\cref{thm:jk-alg}) with threshold $a=0.001$. As discussed in \cref{sec:numerical}, these points approximate $C(p)$ with an additive error of less than $0.003$.} \label{fig:numerical}
	\end{figure}
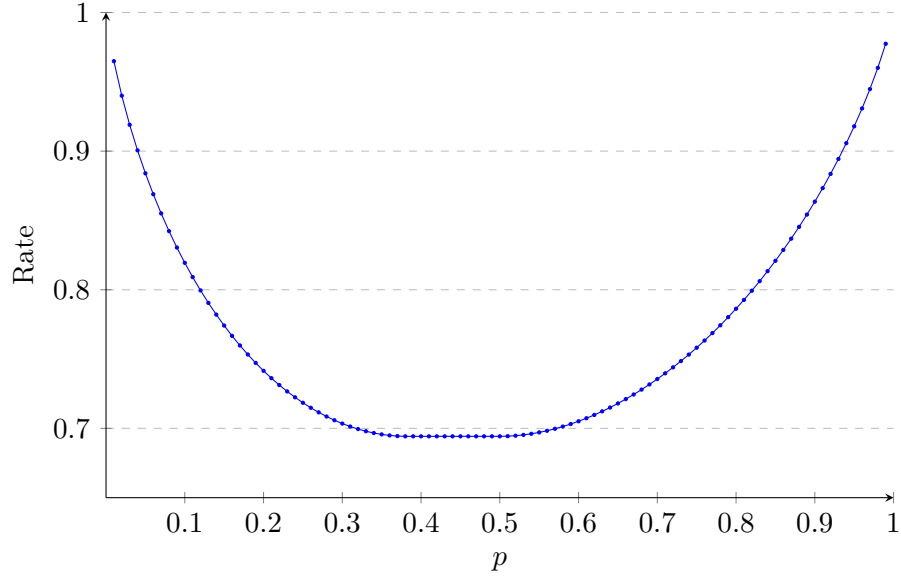

\pgfmathdeclarefunction{entropy}{1}{%
  \pgfmathparse{-(#1)*ln(#1)/ln(2) - (1-(#1))*ln(1-(#1))/ln(2)}%
}

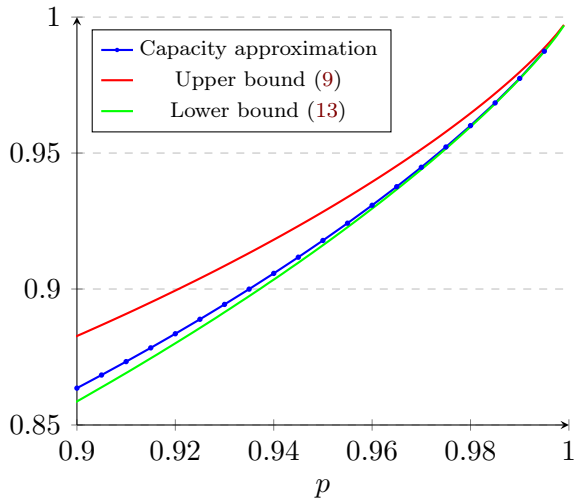
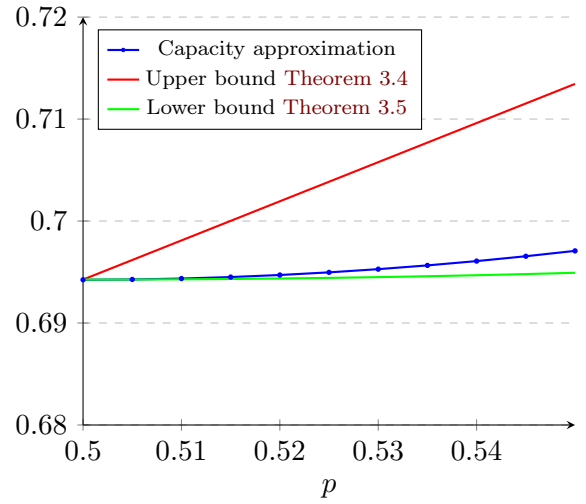
\begin{figure}
    \centering
\begin{subfigure}{0.49\textwidth}
 \centering
        \begin{tikzpicture}
		
		\begin{axis}[
    width=\linewidth,
    scaled y ticks=true,
    tick label style={/pgf/number format/fixed},
    axis lines=left,
    xlabel={$p$},
    xmin=0.9, xmax=1,
    ymin=0.85, ymax=1,
    legend style={
        font=\scriptsize
    },
    legend pos=north west,
    legend image post style={scale=0.7},
    ymajorgrids=true,
    grid style=dashed,
]  
            \addplot[
		color=blue,
        thick,
            mark size=0.6pt,
		mark=*,
		]
		coordinates {
            ( 0.8 , 0.786312961803251 )
            ( 0.805 , 0.7894596122474256 )
            ( 0.81 , 0.7926713876388253 )
            ( 0.815 , 0.795949559924682 )
            ( 0.820 , 0.7992954748157927 )
            ( 0.825 , 0.802710557452374 )
            ( 0.830 , 0.8061963187401816 )
            ( 0.835 , 0.8097543624532411 )
            ( 0.840 , 0.8133863932171871 )
            ( 0.845 , 0.8170942255085453 )
            ( 0.850 , 0.820879793831532 )
            ( 0.855 , 0.8247451642662081 )
            ( 0.860 , 0.828692547621916 )
            ( 0.865 , 0.8327243144799845 )
            ( 0.870 , 0.8368430124726565 )
            ( 0.875 , 0.8410513862251079 )
            ( 0.880 , 0.8453524004897089 )
            ( 0.885 , 0.8497492671336186 )
            ( 0.890 , 0.8542454768129525 )
            ( 0.895 , 0.8588448363935557 )
            ( 0.900 , 0.8635515134807047 )
            ( 0.905 , 0.8683700898281522 )
            ( 0.910 , 0.8733056259551881 )
            ( 0.915 , 0.878363740075885 )
            ( 0.920 , 0.8835507055395343 )
            ( 0.925 , 0.8888735725556983 )
            ( 0.930 , 0.8943403222874501 )
            ( 0.935 , 0.8999600648645246 )
            ( 0.940 , 0.9057432982093652 )
            ( 0.945 , 0.9117022530395353 )
            ( 0.950 , 0.9178513633009066 )
            ( 0.955 , 0.9242079249711804 )
            ( 0.960 , 0.9307930484555113 )
            ( 0.965 , 0.9376330895366252 )
            ( 0.970 , 0.9447619046513027 )
            ( 0.975 , 0.952224628966008 )
            ( 0.980 , 0.9600845382693227 )
            ( 0.985 , 0.9684370023607721 )
            ( 0.990 , 0.9774431732443992 )
            ( 0.995 , 0.9874403913189771 )

        }; 
        \addlegendentry{Capacity approximation};

        \addplot[
            red,thick,domain=0.7,domain=0.7:0.999, samples=200
        ] 
        { 1 - entropy(x) / 4 };
        \addlegendentry{Upper bound \eqref{eq:ub-alphap}};

        \addplot[
            green,thick,domain=0.7,domain=0.7:0.999, samples=200
        ] 
        { (1 + x + entropy((1-x)/2) - entropy(x)) / 2 };
        \addlegendentry{Lower bound \eqref{eq:lb-geom2}};
        
        \end{axis}
	\end{tikzpicture}
	\caption{Our asymptotically optimal closed-form upper and lower bounds (derived in \cref{sec:p-geq-half}) compared to the numerical approximation of the capacity in the large $p$ regime.
    } 
    \label{fig:p-1-ublb}
	\end{subfigure}
\hfill
\begin{subfigure}{0.49\textwidth}
 \centering
        \begin{tikzpicture}
		
		\begin{axis}[
    width=\linewidth,
    scaled y ticks=true,
    tick label style={/pgf/number format/fixed},
    axis lines=left,
    xlabel={$p$},
    xmin=0.5, xmax=0.55,
    ymin=0.68, ymax=0.72,
    xtick={0.5,0.51,0.52,0.53,0.54},
    legend style={
        font=\scriptsize
    },
    legend pos=north west,
    legend image post style={scale=0.7},
    ymajorgrids=true,
    grid style=dashed,
]
            \addplot[
		color=blue,
        thick,
            mark size=0.6pt,
		mark=*,
		]
		coordinates {
                ( 0.5 , 0.6942414858912936 )
                ( 0.505 , 0.6942716107981445 )
                ( 0.51 , 0.69436000567872 )
                ( 0.515 , 0.6945061158609878 )
                ( 0.52 , 0.6947090281293751 )
                ( 0.525 , 0.6949678850576532 )
                ( 0.53 , 0.6952818789104241 )
                ( 0.535 , 0.6956502480730913 )
                ( 0.54 , 0.6960722748449032 )
                ( 0.545 , 0.6965472839474861 )
                ( 0.55 , 0.6970746414224845 )
        }; 
        \addlegendentry{Capacity approximation};
        
        \pgfmathsetmacro{\golden}{(1+sqrt(5))/2}
        \pgfmathsetmacro{\loggolden}{ln(\golden)/ln(2)}
        
        \addplot[
            red,thick,domain=0.5:0.55, samples=200
        ] 
        {\loggolden * (1+ (2 *(x - 0.5) )/ (2+\golden ) )};
        \addlegendentry{Upper bound \cref{thm:UP-p-half}}

        \addplot[
            green,
            thick,
            domain=0.5:0.55,
            samples=200,
        ]
        {
        \loggolden * (
            1
            + \loggolden *
            (4*(3-\golden)*ln(2))/(5*sqrt(5)*(\golden-1))
            * (x - 0.5)^2
        )
        };
        \addlegendentry{Lower bound \cref{prop:LB-p-half}};
        
        \end{axis}
	\end{tikzpicture}
	\caption{Our asymptotic closed-form upper and lower bounds compared to the numerical approximation of the capacity in the $p\approx 1/2$ regime. When plotting these bounds, we omitted the lower order terms.
    } \label{fig:p-half-ublb}
    \end{subfigure}
    \caption{Comparison of our closed-form bounds with numerical capacity approximations in the regimes $p\approx1$ and $\approx1/2$.}
\label{fig:combined}
\end{figure}

\section{The Bernoulli additive-noise sticky channel with bounded runlengths}\label{sec:cap-rlc}

We now consider the Bernoulli additive-noise sticky channel with parameter $p \leq \frac{1}{2}$, and input runlength bounded from above by $\tau < \infty$ (this is equivalent to considering the corresponding additive-noise memoryless channel with finite input alphabet $[\tau]$). We summarize the results of this section below.

\begin{itemize}
    \item \emph{Odd $\tau$}: When $\tau$ is odd, we are able to compute $C(p,\tau)$ exactly for all $p \leq \frac{1}{2}$. We separate our analysis into two cases:
    \begin{enumerate}
        \item \emph{Naive RLC coding optimality}: Taking a similar approach to that of \Cref{thm:opt-rlc}, we show that naive coding over the odd integers is optimal for a region $p \in \left[\frac{1}{\varphi^2}(1 + \eps_{\tau}), \frac{1}{2}\right]$, and that $C(p,\tau) = \log \varphi - \eps_{\tau}'$ in this region, where $\eps_{\tau}, \eps_{\tau}' \to 0$ exponentially fast in $\tau$.

        \item \emph{Remaining values of $p \leq \frac{1}{2}$}: Similarly to \Cref{thm:cap-below-half}, we guess that the optimal output distribution $P_Y$ satisfies $P_Y(y) \propto 2^{-\lambda y}$ for some $\lambda$ depending on $\tau$ and $p$, which leads us to determine the exact capacity for all values of $p$ in the left sub-interval of $\left[0,\frac{1}{2}\right]$ for which naive coding is not optimal.
    \end{enumerate}

    \item \emph{Even $\tau$}: When $\tau$ is even, we are only able to determine the exact capacity $C(p,\tau)$ for values of $p$ in a left sub-interval of $\left[0,\frac{1}{2}\right]$, using the same approach as that of the second case mentioned above for odd $\tau$. In fact, it is easily shown that naive coding is never optimal when $\tau$ is even, and we are unable to obtain the exact capacity for values of $p$ in the remaining right sub-interval of $\left[0,\frac{1}{2}\right]$.

    Towards a better understanding of the capacity-achieving input distribution for the right sub-interval, we prove that the optimal input does not have full support over $[\tau]$, for all $\tau$ large enough.
\end{itemize}

We begin by studying the parameter regime where naive RLC coding is optimal.

\subsection{When is naive RLC coding optimal?}

We find that the naive input distribution over strings with runs of odd length is optimal when $\tau$ is odd, for a region of $p$ which approaches the plateau region $\left[\frac{1}{\varphi^{2}},\frac{1}{2}\right]$ for the infinite input alphabet setting found in \Cref{thm:opt-rlc}. Moreover, both the capacity and the plateau region for $p$ with input alphabet $[\tau]$ approach their infinite-alphabet analogues exponentially fast in $\tau$.
When $\tau$ is even, the naive coding strategy is not optimal. While we do not obtain the exact capacity for this case, we remark that the trivial bounds
\begin{equation*}
    C(p,\tau-1) \leq C(p,\tau) \leq C(p,\tau+1)
\end{equation*}
hold for all $p,\tau$.

\begin{theorem}\label{thm:finite_alphabet_bernoulli_right}
    Fix an odd integer $\tau>1$ and let $\lambda_{\tau} \in (0,1]$ be the unique solution to
    \begin{equation*} \label{eq:right_interval_condition}
        2^{2\lambda} - 2^{\lambda} - 1 + 2^{-\lambda \tau} = 0.
    \end{equation*}
    Then, $C(p,\tau) \geq \lambda_{\tau}$ with equality if and only if $p \in \left[\frac{1}{2^{\lambda_{\tau}} + 1}, \frac{1}{2} \right]$. In particular,
    \begin{equation*}
        C(p,\tau) > \log \varphi - 1.3^{-\tau}
    \end{equation*}
    for $p \in \left[ \frac{1}{\varphi^{2}}(1 + 1.3^{-\tau}), \frac{1}{2} \right]$.
\end{theorem}

\begin{proof}
The proof is analogous to that of \Cref{thm:opt-rlc}. Here, we consider the distribution
\begin{equation*}
    P_X(x) = \begin{cases}
        2^{-\lambda_{\tau} x} & \text{if $x \leq \tau$ is odd,}\\
        0 & \text{otherwise.}
    \end{cases}
\end{equation*}
It is easy to verify that $P_X$ is indeed a valid probability distribution.
Writing $x := 2^{\lambda_{\tau}}$, determining $\lambda_{\tau}$ amounts to finding the roots of
\begin{align*}
    f_{\tau}(x) := x^2 - x - 1 + x^{-\tau}.
\end{align*}
To determine the conditions for which $\lambda_{\tau}$ is the exact capacity, we proceed as in the proof of \Cref{thm:opt-rlc} (i.e., by checking the conditions of \Cref{lem:CPUC-UB-framework}). 
By routine algebraic manipulation, we conclude that this is the case if and only if
\begin{equation} \label{eq:lambda_p_relation_naiveRLC}
    \frac{1}{2^{\lambda_{\tau}}+1} \leq p \leq \frac{1}{2}.
\end{equation}
The convergence statement for $C(p,\tau) = \lambda_{\tau}$ follows from \Cref{lemma:convergence} below. Indeed, since $0 < \varphi - 2^{\lambda_{\tau}} < 1.3^{-\tau}$, it holds that
\begin{equation*}
    \log \varphi - \lambda_{\tau} = \log \left( \frac{\varphi}{2^{\lambda_{\tau}}} \right) < \frac{\varphi}{2^{\lambda_{\tau}}} - 1 < 1.3^{-\tau},
\end{equation*}
where we used that $\log(x+1) \leq x$ for $x > 0$. For the limit point $p = \frac{1}{2^{\lambda_{\tau}}+1}$, we have that
\begin{equation*}
    p \leq \frac{1}{\varphi - 1.3^{-\tau}+1} = \frac{1}{\varphi + 1} + \frac{1.3^{-\tau}}{(\varphi + 1)(\varphi + 1-1.3^{-\tau})} \leq \frac{1}{\varphi^2} + \frac{1.3^{-\tau}}{\varphi^2},
\end{equation*}
where we used the fact that $\varphi + 1 = \varphi^2$.
\end{proof}

\begin{lemma} \label{lemma:convergence}
    Fix an integer $\tau \geq 2$ and define
    \begin{equation*}
        f_{\tau}(x) = x^2 - x - 1 + x^{-\tau}
    \end{equation*}
    and
    \begin{equation*}
        f(x) = x^2 - x - 1.
    \end{equation*}
    Then, in the interval $I = \left( 1, 2 \right)$, $f$ has a unique root $x^{\star} = \varphi$ and $f_{\tau}$ has a unique root $x_{\tau}^{\star}$. Moreover, $|x_{\tau}^{\star} - x^{\star}| < 1.3^{-\tau}$.
\end{lemma}
\begin{proof}
    The claim about the root of $f$ is easy to check. Moreover, for all $\tau$ we have
    \begin{align*}
        f_{\tau}'(x) &= 2x - 1 - \tau x^{-(\tau+1)}
    \end{align*}
    and
    \begin{align*}
        f_{\tau}''(x) = 2 + \tau (\tau + 1) x^{-(\tau+2)}.
    \end{align*}
    Thus, $f_{\tau}'(1)<0$ and $f'_{\tau}$ is strictly increasing in $\left[1,2 \right]$. Since $f_{\tau}(1) = 0$ and $f_{\tau}(\varphi) > f(\varphi) = 0$, it must be the case that $f_{\tau}$ has a unique root $x_{\tau}^{\star} \in \left( 1,\varphi \right)$. Since $f_{\tau}(x) > f_{\tau+1}(x)$ for all $x \in (1,\varphi)$, the sequence $x_{\tau}^{\star}$ is increasing, and thus converges to some value in the interval $[x_2^{\star},\varphi]$.
    Note that $f'(x) > 1$ for all $x \in (1,2)$. Thus, writing $\delta_{\tau} := \varphi - x_{\tau}^{\star}$, we have
    \begin{equation*}
        0 = f_{\tau}(x_{\tau}^{\star}) = f(\varphi - \delta_{\tau}) + \left(x_{\tau}^{\star}\right)^{-\tau} < -\delta_{\tau} + \left(x_{\tau}^{\star}\right)^{-\tau},
    \end{equation*}
    from which we conclude that $\delta_{\tau} < \left(x_{\tau}^{\star}\right)^{-\tau} \leq \left(x_{2}^{\star}\right)^{-\tau}$, where $x_2^{\star} > 1.3$.
\end{proof}

\begin{remark}[On the optimality of naive RLC coding for even $\tau$]
    \em
    If $\tau$ is even, then $Y_{\tau}$ is not absolutely continuous with respect to $Y$, where $Y$ is the output distribution corresponding to the input $X$ defined in the proof of \cref{thm:finite_alphabet_bernoulli_right}, since ${\tau}+1 \in \supp (Y_{\tau})\setminus \supp(Y)$. Thus, $\KL(P_{Y_{\tau}}\|P_Y) = + \infty$, from which we conclude that $X$ is not capacity-achieving by \Cref{lem:AG_necessary}.
\end{remark}

\subsection{Completing the region $p \in \left(0,\frac{1}{2}\right]$}

We now determine the capacity $C(p,\tau)$ of the bounded runlength channel for the region of $p \leq \frac{1}{2}$ not covered by the previous result. Contrarily to before, we are able to determine the exact capacity for both even and odd values of $\tau$.

\begin{theorem}\label{thm:finite_alphabet_bernoulli_left} 
Fix an integer $\tau > 1$ and $p\in(0,1/2]$, and let $\lambda_{\tau,p} \in (0,1)$ be the unique solution to
    \begin{equation} \label{eq:left_interval_condition}
        2^{\lambda p - h(p)}  \cdot \left( 1 - 2^{-\lambda(\tau+1)} \right) - 2^{\lambda} + 1 = 0.
    \end{equation}
    Then, $C(p,\tau) = \lambda_{\tau,p}$ if $p \in \left( 0, \frac{1}{2^{\lambda_{\tau,p}} + 1} \right]$. In particular,
    \begin{equation*}
        C(p,\tau) > \lambda_{p}^{\star} - 1.8 \cdot 1.3^{-\tau}
    \end{equation*}
    for $p \in \left( 0, \frac{1}{2^{\lambda_{p}^{\star}}+1} \left( 1 + 1.8 \cdot 1.3^{-\tau} \right) \right]$, where $\lambda_{p}^{\star}$ is the solution to
    \begin{equation*}
        2^{\lambda p - h(p)} - 2^{\lambda} + 1 = 0.
    \end{equation*}
\end{theorem}
\begin{proof}
    The proof is similar to that of \Cref{thm:cap-below-half}. 
    Defining
    \begin{equation*}
        Q(y) = 2^{-\lambda(y-p) - h(p)}
    \end{equation*}
    for all $y\in[\tau+1]$, one can easily verify that $\sum_{y=1}^{\tau+1} Q(y) = 1$ if and only if $\lambda \in (0,1)$ satisfies \Cref{eq:left_interval_condition}. Note that the map $f(\lambda)= \sum_{y=1}^{\tau+1} Q(y)$ is decreasing in $\lambda$. Moreover, $f(1) = 2^{p-h(p)} (1 - 2^{-(\tau+1)}) < 1$ and $f(0) = (\tau+1)2^{-h(p)}> 1$ whenever $p \leq \frac{1}{2}$ and $\tau > 1$. We then have that $f(\lambda)$ must attain the value 1 for some unique $\lambda_{\tau,p} \in (0,1)$. Equivalently, \Cref{eq:left_interval_condition} has a unique solution $\lambda_{\tau,p} \in (0,1)$. Hence, $Q$ with this choice of $\lambda$ gives rise to a valid probability distribution.
    Analogously to the proof of \Cref{thm:cap-below-half}, we can conclude that
    \begin{equation*}
        \KL(P_{Y_x}\|Q) = \lambda x,
    \end{equation*}
    for every $x \in [\tau]$.
    We now must check that $Q$ is a valid output distribution corresponding to some input distribution $P_X$. 
    This amounts to finding some $P_X$ supported on $[\tau]$ such that, for all $y \in [\tau+1]$,
    \begin{equation*} 
        Q(y) = (1-p)P_X(y) + p P_X(y-1).
    \end{equation*}
    From this equation we can deduce that
    \begin{equation} \label{eq:X_from_Y}
        P_X(x) = \frac{1}{1-p} \sum_{k=0}^{x-1} Q(x-k) \left( -\frac{p}{1-p} \right)^{k}
    \end{equation}
    for all $x\in[\tau]$.
    It suffices to show that $P_X(x) \geq 0$ for all $x \in [\tau]$, as $\sum_{x\in[\tau]} P_X(x)=1$ is 
    guaranteed by construction if $\sum_{y\in[\tau+1]}Q(y)=1$. 
    Substituting $Q(y-k) = Q(1) \cdot 2^{-\lambda_{\tau,p} (y-k-1)}$ and simplifying the geometric sum in \Cref{eq:X_from_Y}, we have that $P_X(x) \geq 0$ if and only if
    \begin{align*}
    \left( -\frac{p 2^{\lambda_{\tau,p}}}{1-p} \right)^{x} \leq 1.
    \end{align*}
    This holds for all $x \in [\tau]$ if and only if $\lambda_{\tau,p} \leq \log \left( \frac{1-p}{p} \right)$.
    We then conclude that $P_Y$ is indeed the output distribution corresponding to a valid input distribution, and that $\lambda_{\tau,p}$ is the exact capacity, as long as
    \begin{equation*}
        0 < p \leq \frac{1}{2^{\lambda_{\tau,p}}+1}.
    \end{equation*}
    We omit the proof of the convergence statement, as the argument is similar to that of \Cref{thm:finite_alphabet_bernoulli_right}. In particular, by a similar argument to that of \Cref{lemma:convergence}, one can show that $2^{\lambda_p^{\star}} - 2^{\lambda_{\tau,p}} < 1.8 \cdot 1.3^{-\tau}$ for all values of $p$ considered.
\end{proof}

We reiterate that \Cref{thm:finite_alphabet_bernoulli_right} and \Cref{thm:finite_alphabet_bernoulli_left} provide the exact capacity $C(p,\tau)$ for all odd $\tau$ and $p \in \left(0,\frac{1}{2}\right]$.
\begin{corollary}
    Let $\tau > 1$ be odd and $p \in \left(0, \frac{1}{2} \right]$. Let $\lambda_{\tau}$ and $\lambda_{\tau,p}$ be as in \Cref{thm:finite_alphabet_bernoulli_right,thm:finite_alphabet_bernoulli_left}, respectively. Then,
    \begin{equation*}
        C(p,\tau) = 
        \begin{cases}
            \lambda_{\tau} & \text{if }p \geq \frac{1}{2^{\lambda_{\tau}}+1} \\
            \lambda_{\tau,p} & \text{otherwise}
        \end{cases}\, .
    \end{equation*}
\end{corollary}
\begin{proof}
    The result follows from \Cref{thm:finite_alphabet_bernoulli_right,thm:finite_alphabet_bernoulli_left} by noting that $\lambda_{\tau} = \lambda_{\tau,p}$ for $p = \frac{1}{2^{\lambda_{\tau}}+1}$.
\end{proof}

\subsection{Properties of the optimal distribution for even-sized alphabets}

Given that the naive RLC coding strategy is never optimal for even-sized input alphabet $[\tau]$, it makes sense to ask whether the capacity-achieving input distribution has full support in the region of $p \leq \frac{1}{2}$ for which we have not determined the exact capacity. We answer this question in the negative, provided that $\tau$ is large enough.

\begin{lemma} \label{lem:support-size-even-tau}
    Let $p \in \left( \frac{1}{\varphi^2}, \frac{1}{2} \right)$. For all $\tau$ large enough, the capacity achieving distribution does not have full support. 
\end{lemma}
\begin{proof}
    Fix $p \in \left( \frac{1}{\varphi^2}, \frac{1}{2} \right)$. For each $\tau$, we define $\eps_{\tau}$ as the quantity such that the interval of valid values of $p$ given by \Cref{thm:finite_alphabet_bernoulli_left} is $\left(0,\frac{1}{\varphi^2} + \eps_{\tau}\right]$, and we denote $\lambda_{\tau} = C(p,\tau)$.
    
    Let $\tau'$ be the smallest odd natural number such that $p > \frac{1}{\varphi^2} + \eps_{\tau'}$. Suppose that $X$ has full support and is capacity-achieving for even-sized input alphabet $[\tau]$ (we show that this leads to a contradiction if $\tau$ is sufficiently large). 
    By \Cref{lem:AG_necessary}, the corresponding output distribution $P_Y$ must satisfy
    \begin{equation*}
        \KL(P_{Y_x}\|P_Y) = \lambda_{\tau} x
    \end{equation*}
    for every $x \in [\tau]$.
    By the same reasoning as that of \Cref{sec:genbound}, writing
    \begin{equation*}
        P_Y(y) = 2^{-\lambda_{\tau} (y-p) - h(p) - g(y)}
    \end{equation*}
    for $y \in [\tau+1]$, it must be the case that
    \begin{equation*}
        g(y) = \left( -\frac{1-p}{p} \right)^{y-1} \beta
    \end{equation*}
    for all $y \in [\tau + 1]$, where $\beta \in \R$ is a constant.

    Since $p > p_{\tau'}$, we have that $p$ belongs to the interval defined in \Cref{thm:finite_alphabet_bernoulli_right} when the input alphabet is $[\tau']$ (i.e., naive RLC coding is optimal). By \Cref{eq:lambda_p_relation_naiveRLC}, for $\tau > \tau'$ it holds that
    \begin{equation} \label{eq:lambda_LWB}
    \lambda_{\tau} > \lambda_{\tau'} > \log \left( \frac{1-p}{p} \right).
    \end{equation}
    We have seen in \Cref{eq:X_from_Y} that, if $P_Y$ is the output distribution induced by $P_X$, then
    \begin{equation*}
        P_X(x) = \frac{1}{1-p} \sum_{k=0}^{x-1} P_Y(n-x) \left( -\frac{p}{1-p} \right)^k.
    \end{equation*}
    Since $P_X$ is a probability distribution with full support, $P_X(2) > 0$. By the above characterization of $P_X$, this is equivalent to the condition
    \begin{align*}
    2^{\frac{\beta}{p} - \lambda_{\tau}} > \frac{p}{1-p}.
    \end{align*}
    In particular, $\beta$ must satisfy
    \begin{equation} \label{eq:beta_condition}
        \beta > p \cdot \left( \lambda_{\tau} - \log \left( \frac{1-p}{p} \right) \right) > 0.
    \end{equation}
    Defining $\delta_p := p \cdot \left( \lambda_{\tau'} - \log \left( \frac{1-p}{p} \right) \right)$, by \Cref{eq:lambda_LWB,eq:beta_condition} we also have that $\beta > \delta_p > 0$, where $\delta_p$ is a constant not depending on $\tau$. On the other hand, for even $\tau$, a necessary condition to ensure that $P_Y(\tau) \leq 1$ is that
    \begin{equation*}
        \beta \leq (\lambda_{\tau}(\tau-p) + h(p)) \cdot \left( \frac{p}{1-p} \right)^{\tau-1}.
    \end{equation*}
    Clearly, if $p < \frac{1}{2}$ then the right-hand side of the equation above is at most $\delta_p$ for $\tau$ large enough. This implies $\beta \leq \delta_p$, a contradiction. 
\end{proof}

\section{General additive-noise sticky channels}

In this section we explore additive-noise sticky channels beyond Bernoulli noise.

\subsection{A universal lower bound on the capacity of any additive-noise sticky channel with prescribed noise mean}\label{sec:universal-LB}

In this section we prove \cref{thm:lower_bound_capacity_noise_infinite_support}, which we restate here to aid the reader.

\lbmean*

\begin{proof}
We begin by showing that $\Capa(\Geom_\mu)=\alpha$, with $\alpha$ the smallest positive solution to \cref{equation_tau_star}.
By \Cref{lem:ANSC-to-ANMC}, $\Capa(\Geom_\mu)$ equals the capacity per unit cost of the memoryless additive-noise channel with noise distribution $\Geom_\mu$.
We will use \cref{lem:CPUC-UB-framework} to determine this capacity per unit cost.

If $Z$ is distributed according to $\Geom_\mu$, then $P_Z(z)=\left(\frac{\mu}{1+\mu}\right)^z\cdot\frac{1}{1+\mu}$ for all $z\in\N$.
With hindsight, consider the probability distribution $Q^\star$ given by
\begin{equation}
\label{capacity_per_unit_cost_achieving_output_geometric}
    Q^\star(y) = 2^{-\alpha\left(y-\mu\right)-(1+\mu)h\left(\frac{1}{1+\mu}\right)}
\end{equation}
for all $y\geq 1$, with $\alpha>0$ the smallest positive solution to the equation
\begin{equation}
\label{normalization_output_geometric_noise}
    \sum_{y=1}^\infty Q^\star(y) = \frac{2^{\alpha\mu-(1+\mu)h\left(\frac{1}{1+\mu}\right)}}{2^\alpha-1} = 1.
\end{equation}
In this case, on input $x\geq 1$ the output $Y_x$ is distributed according to $x+\Geom_\mu$.
Since $H(Y_x)=H(\Geom_\mu)=(1+\mu)h\left(\frac{1}{1+\mu}\right)$ for all $x$, we get that
\begin{equation*}
    D(P_{Y_x}\|Q^\star)= -\sum_{y=0}^\infty P_{Y_x}(y)\log Q^\star(y)  - H(Y_x)=  \alpha x,\quad \forall x\geq 1.
\end{equation*}
Therefore, if we show that $Q^\star$ is realizable as a channel output distribution, then the second part of \cref{lem:CPUC-UB-framework} implies that $\Capa(\Geom_\mu)=\alpha$, as desired.

Similarly to the proof of \cref{thm:cap-below-half}, if indeed $Q^\star$ is realizable as the channel output distribution corresponding to some input distribution $P_X$, then the probability generating function $G_X$ of $P_X$ would satisfy
\begin{equation*}
    G_X(w) = \frac{G_Y(w)}{G_Z(w)},
\end{equation*}
where $G_Y$ and $G_Z$ are the probability generating functions of the output $Y\sim Q^\star$ and noise $Z\sim \Geom_\mu$, respectively.
It can be easily checked that $G_Y(w) = 2^{\alpha\mu - (1+\mu)h\left(\frac{1}{1+\mu}\right)} \cdot \frac{w}{2^\alpha-w}$ and $G_Z(w) = \frac{1}{1+\mu-\mu w}$ in a disk around $w=0$, and so
\begin{equation*}
    G_X(w) =  2^{\alpha\mu - (1+\mu)h\left(\frac{1}{1+\mu}\right)} \cdot \frac{w(1+\mu-\mu w)}{2^\alpha-w}
\end{equation*}
in a disk around $w=0$.
To show that $P_X$ is a valid probability distribution it suffices to check that the coefficients of the Taylor expansion of $\frac{w(1+\mu-\mu w)}{2^\alpha-w}$ around $w=0$, which we denote by $\beta_n$, are non-negative (clearly $\beta_0=0$).

It is easy to check that $\frac{1}{2^\alpha-w} = \sum_{i=0}^\infty 2^{-\alpha(i+1)} w^i$ around $w=0$, and so
\begin{equation*}
    \beta_n=\begin{cases}
        0, &\textrm{if $n=0$,}\\
        (1+\mu) 2^{-\alpha}, &\textrm{if $n=1$,}\\
        (1+\mu)(2^{-\alpha n} -\frac{\mu}{1+\mu} \cdot 2^{-\alpha(n-1)}), &\textrm{if $n>1$.}
    \end{cases}
\end{equation*}
The order-$1$ coefficient is clearly non-negative.
For $n>1$, the order-$n$ coefficient is non-negative if and only if
\begin{equation}\label{eq:bound-alpha-geom}
    \alpha \leq \log\left(\frac{1+\mu}{\mu}\right).
\end{equation}
To verify \cref{eq:bound-alpha-geom}, we first note that, for fixed $\mu>0$, the map $x\mapsto \frac{2^{x\mu - (1+\mu)h\left(\frac{1}{1+\mu}\right)}}{2^x-1}$ is continuous for $x>0$, converges to $+\infty$ as $x\to 0$ from above, is decreasing for $x\in(0,\infty)$ when $\mu\in(0,1]$, and is decreasing for $x\in \left(0,\log\left(\frac{\mu}{\mu-1}\right)\right)$ when $\mu>1$.
Then, setting $x=\log\left(\frac{1+\mu}{\mu}\right)>0$, which is smaller than $\log\left(\frac{\mu}{\mu-1}\right)$ when $\mu>1$, yields
\begin{equation*}
    \frac{2^{x\mu-(1+\mu)h\left(\frac{1}{1+\mu}\right)}}{2^x-1} = \frac{\mu}{1+\mu}<1,
\end{equation*}
and so we conclude that $\alpha\leq \log\left(\frac{1}{1-\delta}\right)$, since $\alpha$ is the smallest positive solution to \cref{normalization_output_geometric_noise}.
As a result, we get that $\Capa(\Geom_\mu)=\alpha$, as desired.

It remains to show that
\begin{equation*}
    \Capa(\cN)\geq \Capa(\Geom_\mu)
\end{equation*}
for any noise distribution $\cN$ with mean $\mu$.
By \Cref{lem:ANSC-to-ANMC}, $\Capa(\cN)$ equals the capacity per unit cost of the memoryless additive-noise channel with noise distribution $\cN$, which we denote by $\CPUC(\cN)$.
Therefore, we can focus on comparing the capacity per unit cost of a memoryless additive-noise channel with noise distribution $\cN$ with that of a memoryless additive-noise channel with noise distribution $\Geom_\mu$.

By the discussion above, we know that the channel output distribution that achieves the capacity per unit cost $\CPUC(\Geom_\mu)$ is $Q^\star$ defined in \cref{capacity_per_unit_cost_achieving_output_geometric}.
Denote its corresponding channel input distribution (with respect to the additive-noise channel with noise distribution $\Geom_\mu$) by $P^\star$.
We will now lower bound $\CPUC(\cN)$ by the rate achieved by the input distribution $P^\star$ to get the desired result.

Below, we use $P_{Y|X}$ to denote the transition probability rule for the additive-noise channel with noise distribution $\cN$ and $Q_{Y|X}$ to denote the transition probability rule for the additive-noise channel with noise distribution $\Geom_\mu$.
More precisely, $P_{Y|X}(x,y)$ is the probability of observing channel output $y$ on input $x$, and likewise for $Q_{Y|X}$.
To simplify exposition, we will also use shorthand such as $\frac{P_{Y|X}}{P_Y}(x,y)$ to mean $\frac{P_{Y|X}(x,y)}{P_Y(y)}$.
The mapping of arguments to the various functions will be clear from context.
For $X\sim P^\star$ and $Y$ its corresponding channel output random variable, we have
\begin{align}
    \CPUC(\cN) &\geq \frac{I(X;Y)}{\E[X]} \label{eq:def-CPUC}\\
    &=\E_{P_{X,Y}}\left[\log\left(\frac{P_{Y|X}}{P_Y}(X,Y)\right)\right]\cdot \frac{1}{\E[X]} \label{eq:def-MI}\\
    &= \E_{P_{X,Y}}\left[\log\left(\frac{P_{Y|X}Q^\star}{Q_{Y|X}P_Y}(X,Y)\right)+ \log\left(\frac{Q_{Y|X}}{Q^\star}(X,Y)\right)\right] \cdot \frac{1}{\E[X]} \label{eq:alg1}\\
    &= \E_{P_{X,Y}}\left[\log\left(\frac{P_{Y|X}Q^\star P^\star}{Q_{Y|X} P_Y P^\star}(X,Y)\right)+ \log\left(\frac{Q_{Y|X}}{Q^\star}(X,Y)\right)\right]\cdot \frac{1}{\E[X]} \label{eq:alg2}\\
    &=\E_{P_{X,Y}}\left[\log\left(\frac{P_{X|Y}}{Q_{X|Y}}(X,Y)\right)+ \log\left(\frac{Q_{Y|X}}{Q^\star}(X,Y)\right)\right]\cdot \frac{1}{\E[X]}\label{eq:bayes}\\
    &=\frac{\E_{Y\sim P_Y}[\KL(P_{X|Y}(\cdot|Y)\|Q_{X|Y}(\cdot|Y))]}{\E[X]}+\E_{P_{X,Y}}\left[\log\left(\frac{Q_{Y|X}}{Q^\star}(X,Y)\right)\right]\cdot \frac{1}{\E[X]}\label{eq:def-KL-well-def}\\
    &\geq \E_{P_{X,Y}}\left[\log\left(\frac{Q_{Y|X}}{Q^\star}(X,Y)\right)\right]\cdot \frac{1}{\E[X]}\label{eq:KL-non-neg}\\
    &= \E_{P_{X,Y}}\left[\log\left(\frac{\left(\frac{\mu}{1+\mu}\right)^{Y-X} \frac{1}{1+\mu}}{2^{-\alpha\left(Y-\mu\right)-(1+\mu)h\left(\frac{1}{1+\mu}\right)}}\right)\right]\cdot \frac{1}{\E[X]}\label{eq:def-Q}\\
    &=\E_{P_{X,Y}}\left[\log\left(\frac{1}{1+\mu}\right)+(Y-X)\log\left(\frac{\mu}{1+\mu}\right)+\alpha \left(Y-\mu\right)+(1+\mu)h\left(\frac{1}{1+\mu}\right)\right]\cdot \frac{1}{\E[X]}\nonumber\\
    &= \frac{\log\left(\frac{1}{1+\mu}\right)+\mu\log\left(\frac{\mu}{1+\mu}\right)+\alpha\E[X]+(1+\mu)h\left(\frac{1}{1+\mu}\right)}{\E[X]}\label{eq:mean-assumption}\\
    &=\alpha\nonumber\\
    &=\CPUC(\Geom_\mu).\label{eq:capgeom}
\end{align}
\cref{eq:def-CPUC,eq:def-MI} follow by the definition of capacity per unit cost and mutual information. \cref{eq:alg1,eq:alg2} follow by algebraic manipulation. \cref{eq:bayes} follows by Bayes' theorem. \cref{eq:def-KL-well-def} follows by the definition of the Kullback-Leibler divergence (and this step is well defined because $Q_{X|Y}$ is fully supported on $\N$). \cref{eq:KL-non-neg} follows by the non-negativity of the Kullback-Leibler divergence. \cref{eq:def-Q} follows by the definition of $Q_{Y|X}$ and $Q^\star$. \cref{eq:mean-assumption} follows from the fact that for $Z\sim \cN$ we have $\E_{P_Y}[Y]=\E[X+Z]=\E[X]+\mu$ (and so, in particular, $\E_{P_{X,Y}}[Y-X]=\mu$). \cref{eq:capgeom} holds by the first part of this proof, which shows that $\CPUC(\Geom_\mu)=\alpha$.
\end{proof}

\subsection{Improved capacity lower bounds for bounded-support noise with prescribed mean}
\label{section_lower_bound_capacity_bounded_noise}

In this section we prove \cref{theorem_lower_bound_capacity_noise_finite_support}.
The proof is divided into two parts: $\mu\geq m/2$ (\cref{thm:RLC-never-optimal} in \cref{sec:RLC-opt-gen}) and $\mu<m/2$ (\cref{thm:bounded-supp-small-mu} in \cref{sec:cap-geom-bounded}).

\subsubsection{A necessary and sufficient criterion for the optimality of naive RLC coding, and applications}\label{sec:RLC-opt-gen}

The first item of \cref{theorem_lower_bound_capacity_noise_finite_support} is a consequence of the following necessary and sufficient criterion for the optimality of naive RLC coding for bounded-support noise distributions $\cN$.

\begin{lemma} \label{lem:naiveRLCoptimality}
    Consider an arbitrary noise distribution $\cN=P_Z$ such that $\supp(\cN) = \{0,1,\dots,m\}$ and define $p_i=\cN(i)$.
    For each integer $x\geq 1$, define $\hat{x}=(x-1)\Mod(m+1)$.
    For $\varphi$ the unique positive solution to $x^{m+1}-x^m-1=0$, consider the input distribution $P_X$ given by
    \begin{equation*}
        P_X(x) := \begin{cases}
            \varphi^{-x}, & \text{if $x \equiv 1\pmod{m+1}$,} \\
            0, & \text{otherwise.}
        \end{cases}
    \end{equation*}
    Then, $\CPUC(\cN)=\log \varphi$ and $P_X$ achieves this capacity per unit cost if and only if
    \begin{equation*}
        \KL \left( (p_0,\dots,p_{m}) \| (p_{\hat{x}}, \dots,p_{m}, p_{0}, \dots, p_{\hat{x}-1})  \right) \leq \log\varphi \cdot \left(\hat{x} - (m+1) \cdot \Pr[Z \geq m+1 - \hat{x}]\right)
    \end{equation*}
    for every $x \geq 1$ such that $x \not\equiv 1 \pmod{m+1}$.

    Furthermore, if $m = \max \supp(\cN)$ and $\supp(\cN)\neq \{0,1,\dots,m\}$, then $P_X$ does not achieve $\CPUC(\cN)$.
\end{lemma}

\begin{proof}
    First, note that the output distribution $P_Y$ corresponding to input $P_X$ is given by
    \begin{equation*}
    P_Y(y) = P_X\left(1+\left\lfloor\frac{y-1}{m+1}\right\rfloor\cdot(m+1)\right) \cdot p_{\hat{y}}.
    \end{equation*}
    Note that $H(P_{Y_x}) < \infty$ for all $x$, as $P_{Y_x}$ has finite support. Moreover, since $P_Y$ has full support, $\KL(P_{Y_x}\|P_Y) < \infty$ for all $x$. Thus, by \Cref{lem:AG_necessary}, $P_X$ is capacity-achieving if and only if there exists $\lambda \geq 0$ such that
    \begin{equation*}
        \KL(P_{Y_x}\|P_Y) \leq \lambda x
    \end{equation*}
    for every $x \in \N_{>0}$, with equality if $x\in\supp(P_X)$ (i.e., if $x \equiv 1 \pmod{m+1}$).
    When $x\in\supp(P_X)$, we have
    \begin{align*}
        \KL(P_{Y_x}\|P_Y) &= \sum_{i=0}^{m} p_i \log \left( \frac{p_i}{P_Y(x+i)} \right) \\
        &= \sum_{i=0}^{m} p_i \log \left( \frac{p_i}{p_i \cdot P_X(x)} \right) \\
        &= x \log \varphi.
    \end{align*}
    Let $\tilde{x} = m+1 - \hat{x}$ and $f(x) = x - \hat{x}$. If $x \not\in\supp(P_X)$ (i.e., $x\not\equiv 1 \pmod{m+1}$), we have
    \begin{align*}
        \KL(P_{Y_x}\|P_Y) &= \sum_{i=0}^{\tilde{x}-1} p_i \log \left( \frac{p_i}{p_{i+\hat{x}} \cdot P_X(f(x))} \right) + \sum_{i=\tilde{x}}^{m} p_i \log \left( \frac{p_i}{p_{i-\tilde{x}} \cdot P_X(f(x)+m+1)} \right) \\
        &= \sum_{i=0}^{\tilde{x}-1} p_i \log \left( \frac{p_i}{p_{i+\hat{x}}} \right) + \sum_{i=\tilde{x}}^{m} p_i \log \left( \frac{p_i}{p_{i-\tilde{x}}} \right) + \sum_{i=0}^{\tilde{x}-1} p_i f(x) \log \varphi + \sum_{i=\tilde{x}}^{m} p_i \log \varphi (f(x)+m+1) \\
        &= \KL \left( (p_0,\dots,p_{m}) \| (p_{\hat{x}}, \dots,p_{m}, p_{0}, \dots, p_{\hat{x}-1})  \right) + \log \varphi \cdot \left( f(x)  + (m+1) \cdot \Pr[Z \geq \tilde{x}] \right).
    \end{align*}
    Replacing $f(x) = x - \hat{x}$ in the above, the inequality $\KL(P_{Y_x}\|P_Y) \leq \log \varphi \cdot x$ holds if and only if
    \begin{equation*}
        \KL \left( (p_0,\dots,p_{m}) \| (p_{\hat{x}}, \dots,p_{m}, p_{0}, \dots, p_{\hat{x}-1})  \right) \leq \log \varphi \cdot \left( \hat{x} - (m+1) \cdot \Pr[Z \geq \tilde{x}] \right),
    \end{equation*}
    as desired.

    We now show the last part of the lemma statement, i.e., that if $\supp(\cN)\neq \{0,..,m\}$, then $P_X$ does not achieve $\CPUC(\cN)$. 
    Indeed, fix some $k \in \{0,...,m-1\} \setminus \supp(\cN)$ (this set is non-empty because we assume that $m=\max\supp(\cN)$). Then, $m+k+2 \in \supp(P_{Y_{k+2}}) \setminus \supp(P_Y)$, with $P_Y$ the output distribution corresponding to input distribution $P_X$. Thus, we get that $\KL(P_{Y_{k+2}}\|P_Y) = \infty$, and by \Cref{lem:AG_necessary} we conclude that $P_X$ is not optimal.
\end{proof}

We are now ready to prove the first part of \cref{theorem_lower_bound_capacity_noise_finite_support}.
\begin{theorem}\label{thm:RLC-never-optimal}
    Fix an integer $m>0$ and a real number $\mu\geq m/2$. 
    Consider an arbitrary noise distribution $\cN$ of mean $\mu$ and such that $\max\supp(\cN)=m$.
    Then, $\Capa(\cN)>\log\varphi$ unless $\cN$ is uniform over $\{0,1,\dots,m\}$. That is, in this regime naive RLC coding is optimal if and only if $\cN$ is uniform.
\end{theorem}
\begin{proof}
    We apply \cref{lem:naiveRLCoptimality}.
    By that lemma, we may assume that $\supp(\cN)=\{0,1,\dots,m\}$, as otherwise naive RLC coding is not optimal.

    If $\cN$ is uniform over $\{0,1,\dots,m\}$, then we have
    \begin{equation*}
        \KL \left( (p_0,\dots,p_{m}) \| (p_{\hat{x}}, \dots,p_{m}, p_{0}, \dots, p_{\hat{x}-1})  \right) = 0,
    \end{equation*}
    for all $x \geq 1$, where $\hat{x}=(x-1)\pmod{m+1}$ is defined as in \cref{lem:naiveRLCoptimality}.
    Thus, the optimality of naive RLC coding follows from that lemma by noting that $\hat{x} - (m+1) \cdot \Pr[Z \geq m + 1 - \hat{x}] = 0$ for all $x \geq 1$, where $Z\sim\cN$.

    Now suppose that $\cN$ is not uniform but $\Capa(\cN) = \log \varphi$. Then, we know that
    \begin{equation*}
        \KL \left( (p_0,\dots,p_{m}) \| (p_{\hat{x}}, \dots,p_{m}, p_{0}, \dots, p_{\hat{x}-1})  \right)>0
    \end{equation*}
    for all $x \not\equiv 1 \pmod{m+1}$.
    Once more invoking \Cref{lem:naiveRLCoptimality}, it holds that
    \begin{equation*}
        \Pr[Z \geq m+1-\hat{x}]< \frac{\hat{x}}{m+1}
    \end{equation*}
    for all $\hat{x}\in\{1,\ldots,m\}$.
    In turn, this implies that
    \begin{equation*}
        \E[Z]=\sum_{x=0}^{m-1} \Pr\left[Z>x\right] =\sum_{x=1}^m \Pr\left[Z\geq m+1-x\right]< \sum_{x=1}^m \frac{\hat{x}}{m+1} =\frac{m}{2}. \qedhere
    \end{equation*}
\end{proof}

\subsubsection{Universal capacity lower bounds for bounded-support noise with mean $\mu<m/2$}\label{sec:cap-geom-bounded}

To complement the previous section, we now focus on noise distributions $\cN$ with $\max\supp(\cN)=m$ and mean $\mu<m/2$, and prove the second and third parts of \cref{theorem_lower_bound_capacity_noise_finite_support}.
In this setting, we establish a capacity lower bound for all such distributions that cannot be improved in general (it is achieved by a ``truncated geometric'' noise distribution).

More precisely, for given $m$ and $\mu<m/2$ we define the \emph{truncated geometric distribution} $\Geom_{\mu,m}$ defined as follows\footnote{Intuitively, the $\Geom_{\mu,m}$ distribution is obtained by truncating (and renormalizing) a geometric distribution to $\{0,1,\dots,m\}$, and setting the $\delta$ parameter appropriately so that the resulting truncated distribution has mean $\mu$.}: for $Z\sim \Geom_{\mu,m}$ we have $P_Z(z)=\frac{1-\delta}{1-\delta^{m+1}}\cdot \delta^z$ for $z\in\{0,1,\dots,m\}$ and $P_Z(z)=0$ otherwise, where $\delta$ is the unique solution in $(0,1)$ to
\begin{equation*}
    \mu=\frac{x}{1-x}-\frac{(m+1)x^{m+1}}{1-x^{m+1}}.
\end{equation*}
We begin by determining $\Capa(\Geom_{\mu,m})$.
We divide the analysis into two cases depending on the value of $\delta$ as a function of $\mu$ and $m$.

\begin{theorem}\label{thm:bounded-supp-RLC}
    Fix an integer $m>0$ and a real number $\mu\in(0,m/2)$.
    Recall that $\varphi$ denote the unique positive solution to $x^{m+1}-x^m-1=0$.
    If $\delta\geq 1/\varphi$, then
    \begin{equation*}
        \Capa(\Geom_{\mu,m})=\log\varphi.
    \end{equation*}
\end{theorem}
\begin{proof}
    As before, we can focus on $\CPUC(\Geom_{\mu,m})$ by \cref{lem:ANSC-to-ANMC}.
    We will show that naive RLC coding is optimal in the regime of the theorem statement via \cref{lem:CPUC-UB-framework}.
    Consider
    \begin{align}
    \label{rlc_coding_geometric_noise_finite_support}
        P_X(x)&=\begin{cases}
                \varphi^{-x}, &\textrm{if $x\equiv 1 \pmod{m+1}$},\\
                0, &\textrm{otherwise}.
                \end{cases}
    \end{align}
    As in the proof of \cref{lem:naiveRLCoptimality}, we will use the notation $\hat{x}=(x-1)\Mod(m+1)$.
    Notice that $P_X$ chosen as in \cref{rlc_coding_geometric_noise_finite_support} induces the output distribution
    \begin{align*}
        P_Y(y)&= P_Z(\hat{y})\cdot \varphi^{-y + \hat{y}}, \qquad \forall y\geq 1.
    \end{align*}
    Therefore, for an arbitrary $x\geq 1$ and $Z\sim\Geom_{\mu,m}$ we have
    \begin{align}
        \KL(P_{Y_x}\|P_Y)
        &= -\sum_{i=0}^m P_{Z}(i) \log(P_Y(i+x)) - H(Z)\nonumber\\
        &= -\sum_{i=0}^m P_{Z}(i) \Bigg[\log\left(P_Z\left(\widehat{i+x}\right)\right) + \left(-(i+x) + \widehat{i+x}\right)\log(\varphi)\Bigg] - H(Z)\nonumber\\
        &= x\log\varphi +\log\varphi\cdot \sum_{i=0}^m  P_{Z}(i)\left(i - \widehat{i+x}\right)  +\sum_{i=0}^m P_{Z}(i) \log\left(\frac{P_{Z}(i)}{P_Z\left(\widehat{i+x}\right)}\right) \nonumber\\
        &= x\log \varphi +\log\varphi\cdot \sum_{i=0}^m P_{Z}(i) \left(i -\widehat{i+x}\right)  +\sum_{i=0}^m P_{Z}(i) \log\left(\delta^{i-\widehat{i+x}}\right) \nonumber\\
        &= x\log\varphi  + \left(\log\varphi  + \log\delta\right) \sum_{i=0}^m P_{Z}(i)\cdot \left(i - \widehat{i+x}\right).\label{eq:GeomMuM-simplified-KL}
    \end{align}
    In particular, when $x\in\supp(P_X)$ (i.e., $x \equiv 1 \pmod{m+1}$,
    we get $\widehat{i+x}= (i+x-1) \Mod{(m+1)} = i$, and so \cref{eq:GeomMuM-simplified-KL} becomes
    \begin{align}
    \label{last_step_KKT_support}
        \KL(P_{Y_x}\|P_Y)&=x\log \varphi.
    \end{align}
    When $x\not\in\supp(P_X)$, from \cref{eq:GeomMuM-simplified-KL} we get
    \begin{align}
        \KL(&P_{Y_x}\|P_Y)= x\log\varphi + \left(\log\varphi  + \log\delta\right) \sum_{i=0}^{m-\hat{x}} P_{Z}(i) \left(-\hat{x}\right) + \left(\log\varphi + \log\delta\right) \sum_{i=m+1-\hat{x}}^{m} P_{Z}(i) \left(m+1 - \hat{x}\right)\nonumber\\
        &= x\log\varphi + \left(\log\varphi  + \log\delta\right) \left(- \hat{x} + (m+1)\sum_{i=m+1-\hat{x}}^{m} P_{Z}(i) \right)\nonumber\\
        &=x\log\varphi + \left(\log\varphi  + \log\delta\right) \left(- \hat{x}+ (m+1)\frac{\sum_{i=m+1- \hat{x}}^{m} \delta^i}{1-\delta^{m+1}} \right)\nonumber\\
        \label{last_step_KKT_no_support}
        &\leq x\log\varphi,
    \end{align}
    where \cref{last_step_KKT_no_support} holds because, on the one hand, our assumption that $\delta\geq 1/\varphi$ implies that $\log\varphi+\log\delta\geq 0$ and, on the other hand, 
    \begin{align*}
        -\hat{x}+ (m+1)\frac{\sum_{i=m+1- \hat{x}}^{m} \delta^i}{1-\delta^{m+1}} = \left(\sum_{i=m+1- \hat{x}}^{m} \delta^i\right) \left(-\frac{\hat{x}}{\sum_{i=m+1- \hat{x}}^{m} \delta^i}+ \frac{m+1}{\sum_{i=0}^{m} \delta^i}\right)\leq 0
    \end{align*}
    which, in turn, holds because $\delta<1$ implies that the average of the negative terms is smaller than the average of the positive terms.
    Combining \cref{last_step_KKT_support,last_step_KKT_no_support}
    yields the desired result by \Cref{lem:CPUC-UB-framework}.
\end{proof}

The next theorem determines $\Capa(\Geom_{\mu,m})$ in the remaining regime where $\delta<1/\varphi$.
\begin{theorem}\label{thm:bounded-supp-other}
    Fix an integer $m>0$ and a real number $\mu\in(0,m/2)$.
    If $\delta< 1/\varphi$, then
    \begin{equation*}
        \Capa(\Geom_{\mu,m})=\alpha
    \end{equation*}
    with $\alpha$ the smallest positive solution to
    \begin{equation}
    \label{equation_capacity_finite_support}
        \frac{2^{\alpha\mu+\log\left(\frac{1-\delta}{1-\delta^{m+1}}\right)+\mu\log\delta}}{2^{\alpha}-1}=1.
    \end{equation}
    Furthermore, $\alpha>\log\varphi$.
\end{theorem}
\begin{proof}
First, we will show that the capacity of the additive-noise sticky channel with $\Geom_{\mu,m}$ noise is $\alpha$ where $\alpha$ satisfies \cref{equation_capacity_finite_support}.
By \Cref{lem:ANSC-to-ANMC}, $\Capa(\Geom_{\mu,m})$ equals the capacity per unit cost of the memoryless additive-noise channel with noise distribution $\Geom_{\mu,m}$.
We will determine this capacity per unit cost via \cref{lem:CPUC-UB-framework}.
Recall that $Z$ is distributed according to $\Geom_{\mu,m}$, then $P_Z(z)=\frac{1-\delta}{1-\delta^{m+1}}\delta^z$ for $z\in\{0,1,\dots,m\}$. and consider the probability distribution $Q^\star$ given by
\begin{equation}
\label{optimal_output_distribution_finite_support}
    Q^\star(y) = 2^{-\alpha\left(y-\mu\right)-H(Z)} = 2^{-\alpha\left(y-\mu\right)+\log\left(\frac{1-\delta}{1-\delta^{m+1}}\right)+\mu\log\delta}, \quad y\in\N{>0},
\end{equation}
with $\alpha>0$ the smallest positive solution to the equation
\begin{equation*}
    \sum_{y=1}^\infty Q^\star(y) = \frac{2^{\alpha\mu+\log\left(\frac{1-\delta}{1-\delta^{m+1}}\right)+\log(\delta)\mu}}{2^\alpha-1} = 1,
\end{equation*}
which yields \cref{equation_capacity_finite_support}.
Then, for all $x\in\N_{>0}$,
\begin{align*}
    \KL(P_{Y_x}\|Q^\star)&=-\sum_{y=x}^{x+m}P_{Y_x}(y)\log\left(Q^\star(y)\right)-H(Z)\nonumber\\
    &=-\sum_{y=0}^{m}P_{Z}(y)\log\left(Q^\star (x+y)\right)-H(Z)\nonumber\\
     &=\alpha x.
\end{align*}
Thus, if $Q^\star$ is induced as the output of the channel by some input distribution, \cref{lem:CPUC-UB-framework} implies that $\alpha$ is the capacity per unit cost of the additive-noise memoryless channel, which equals the capacity of the additive-noise sticky channel.
Next, we check that such an input $P^\star$ exists by studying its probability generating function. Note that
\begin{align*}
    G_Z(w)
    &=\frac{1-\delta}{1-\delta^{m+1}} \frac{1-\delta^{m+1}w^{m+1}}{1-\delta w},
\end{align*}
then
\begin{align}
\label{probability_generating_function_geometric_finite}
    G_X(w)&=\frac{G_Y(w)}{G_Z(w)}
    =\frac{2^{\alpha\mu-H(Z)-\alpha}(1-\delta^{m+1})}{1-\delta} g_A(w) g_B(w),
\end{align}
where
\begin{align*}
    g_A(w)=\frac{w(1-\delta w)}{1-2^{-\alpha}w}
\end{align*}
and
\begin{align*}
    g_B(w)=\frac{1}{1-\delta^{m+1}w^{m+1}}.
\end{align*}
We look for conditions to ensure that all coefficients in the power series expansion of $G_X$ are non-negative for $P^\star$ to be correctly defined.
Notice that
\begin{align*}
    g_A(w)&=w (1-\delta w) \sum_{k\in\N} 2^{-\alpha k} w^k\nonumber\\
    &= w\left(1+\sum_{k\in\N_{>0}} \left(2^{-\alpha k}-\delta 2^{-\alpha (k-1)}\right) w^{k}\right)
\end{align*}
and
\begin{align*}
    g_B(w)
    &=\sum_{k\in\N}\delta^{k(m+1)}w^{k(m+1)}.
\end{align*}
We have $g_A(0)=0$, $g_A^{(1)}(0)=1$, and $g_A^{(n)}(0)=n! 2^{-\alpha (n-1)}\left(1-\delta 2^{\alpha}\right)$ for $n\geq 2$.
Furthermore,
\begin{align*}
    g_B^{(n)}(0)&=\begin{cases}
        n! \delta^{n},  &\textrm{if $n \equiv 0 \pmod{m+1}$},\\
        0, &\textrm{otherwise}.
    \end{cases}
\end{align*}
Since $G_X^{(n)}(0)$ is a linear combination of $g_A^{(n)}(0)$ and $g_B^{(n)}(0)$, $P^\star$ is a valid distribution if and only if $\alpha\leq -\log(\delta)$. We check this, by recalling that $\alpha$ is the smallest solution to \cref{equation_capacity_finite_support} and letting
\begin{equation}
\label{function_solution_alpha_finite_support}
f(x)=\frac{2^{x\mu+\log\left(\frac{1-\delta}{1-\delta^{m+1}}\right)+\log(\delta)\mu}}{2^x-1}    
\end{equation}
be the content in \cref{equation_capacity_finite_support}. We verify the existence and uniqueness of $\alpha$ by noticing that
\begin{equation*}
    \lim\limits_{x\rightarrow 0}f(x)=+\infty \textnormal{ and } f(-\ln(\delta))=\frac{\delta}{1-\delta^{m+1}}<\frac{\frac{1}{\varphi}}{1-\frac{1}{\varphi^{m+1}}}=1,
\end{equation*}
while for all $x\leq-\ln(\delta)$,
\begin{align*}
f'(x)&=f(x)\left(\mu-\frac{1}{2^x-1}\right)\ln 2\nonumber\\
&\leq f(x)\left(\mu-\frac{1}{\frac{1}{\delta}-1}\right)\ln 2 \nonumber\\
&=-\frac{(m+1)\delta^{m+1}}{1-\delta^{m+1}}\nonumber\\
&<0.\nonumber
\end{align*}

Second, we will show that the capacity $\alpha>\log(\varphi)$.
Notice that for $\delta=1/\varphi$, the choice $\alpha=-\log\delta=\log\varphi$ makes $P^\star$ with probability generating function \cref{probability_generating_function_geometric_finite} to be the naive RLC coding \cref{rlc_coding_geometric_noise_finite_support}. Notice also, that we have shown that \cref{function_solution_alpha_finite_support} is decreasing with $\alpha$, thus, we can show that the smallest solution to \cref{equation_capacity_finite_support} is strictly greater than $\log(\varphi)$ for $0<\delta<1/\varphi$ by showing that 
\begin{equation*}
\label{function_solution_delta_finite_support}
g(\delta)=\frac{2^{\alpha\mu+\log\left(\frac{1-\delta}{1-\delta^{m+1}}\right)+\mu\log\delta}}{2^\alpha-1}    
\end{equation*}
strictly decreases with $\delta$ on $0\leq \delta< 1/\varphi$, for any $\alpha$.
To this end, we study the behavior of $g$ via the derivative
\begin{align}
\label{derivative_function_solution_delta_finite_support}
    g'(\delta)
    &= \ln(2)g(\delta)\left(\alpha+\ln(\delta)\right)\frac{\dd \mu}{\dd \delta}.
\end{align}
Recall that $\alpha\leq-\ln\delta$. To know the sign of $g'$, we are left to study
\begin{align*}
\frac{\dd \mu}{\dd \delta}&=\frac{1}{(1-\delta)^2}-\frac{(m+1)^2\delta^m}{(1-\delta^{m+1})^2}.
\end{align*}
Notice that $\dd \mu/\dd \delta>0$ is equivalent to $1/(1-\delta)-(m+1)\delta^{m/2}/(1-\delta^{m+1})>0$, which in turn is equivalent to
\begin{equation}\label{eq:AMGM}
    \frac{1}{m+1}\frac{1-\delta^{m+1}}{1-\delta} -\delta^{m/2}
    = \frac{\sum_{i=0}^m \delta^i}{m+1} - \left(\prod_{i=0}^m \delta^i\right)^{\frac{1}{m+1}}> 0.
\end{equation}
The last step of \cref{eq:AMGM} follows by the AM-GM inequality.
Therefore for any $0\leq \delta< 1/\varphi$, \cref{derivative_function_solution_delta_finite_support} is strictly negative.
This concludes the proof.
\end{proof}

Finally, to obtain the main result we show that truncated geometric noise is the worst-case in the setting of this section, and then invoke \cref{thm:bounded-supp-RLC,thm:bounded-supp-other}.
The argument is very similar to that in the proof of \cref{thm:lower_bound_capacity_noise_infinite_support}.
\begin{theorem}\label{thm:bounded-supp-small-mu}
    Fix an integer $m>0$ and a real number $\mu\in(0,m/2)$.
    Consider an arbitrary noise distribution $\cN$ of mean $\mu$ and such that $\max\supp(\cN)=m$.
    Recall that $\delta$ is the unique solution in $(0,1)$ to
    \begin{equation*}
        \mu=\frac{x}{1-x}-\frac{(m+1)x^{m+1}}{1-x^{m+1}}
    \end{equation*}
    and that $\varphi$ is the unique positive solution to $x^{m+1}-x^m-1=0$.
    Then,
    \begin{equation*}
        \Capa(\cN)\geq \Capa(\Geom_{\mu,m}) = \begin{cases}
            \log\varphi, &\textrm{if $\delta\geq 1/\varphi$,}\\
            \alpha, &\textrm{if $\delta<1/\varphi$,}
        \end{cases}
    \end{equation*}
    where $\alpha$ is the smallest positive solution to
    \begin{equation*}
        \frac{2^{\alpha\mu+\log\left(\frac{1-\delta}{1-\delta^{m+1}}\right)+\mu\log\delta}}{2^{\alpha}-1}=1.
    \end{equation*}
    Furthermore, $\alpha>\log\varphi$ when $\delta<1/\varphi$, and so naive RLC coding is never optimal in this regime. 
\end{theorem}
\begin{proof}
    By \cref{lem:ANSC-to-ANMC} it suffices to focus on the capacity per unit cost.
    Since we always have $\Capa(\cN)\geq \log\varphi$ with $\varphi$ the unique positive solution to $x^{m+1}-x^m-1=0$ and $\Capa(\Geom_{\mu,m})=\log\varphi$ when $\delta\geq 1/\varphi$, we can focus only on choices of $m$ and $\mu$ such that $\delta<1/\varphi$. It remains to show that for $\delta<1/\varphi$ we have
    \begin{equation*}
        \Capa(\cN)\geq \Capa(\Geom_{\mu,m})
    \end{equation*}
for any noise distribution $\cN$ with mean $\mu$.
Recall that \Cref{lem:ANSC-to-ANMC}, $\Capa(\cN)$ equals the capacity per unit cost of the memoryless additive-noise channel with noise distribution $\cN$, which we denote by $\CPUC(\cN)$, which we lower bound.

In the proof of \cref{thm:bounded-supp-other}, we show that the channel output distribution that achieves the capacity per unit cost $\CPUC(\Geom_{\mu,m})$ is $Q^\star$ defined in \cref{optimal_output_distribution_finite_support}.
Denote its corresponding channel input distribution (with respect to the additive-noise channel with noise distribution $\Geom_{\mu,m}$) by $P^\star$.

We denote $P_{Y|X}$ to denote the transition probability rule for the additive-noise channel with noise distribution $\cN$ and $Q_{Y|X}$ to denote the transition probability rule for the additive-noise channel with noise distribution $\Geom_{\mu,m}$.
To simplify simplify the writing, we shorten $\frac{P_{Y|X}(x,y)}{P_Y(y)}$ as $\frac{P_{Y|X}}{P_Y}(x,y)$.
We will now lower bound $\CPUC(\cN)$ by the rate achieved by the input distribution $P^\star$:
    \begin{align}
        \CPUC(\cN)&\geq\frac{I(X;Y)}{\E[X]}\nonumber\\
        \label{use_size_finite_support}
        &= \frac{1}{\E[X]}\left(\KL(P_{X|Y}\|Q_{X|Y}|P_Y)+\E_{P_{X,Y}}\left[\log\left(\frac{Q_{Y|X}}{Q_{Y}}(X,Y)\right)\right]\right)\\
        &\geq \frac{1}{\E[X]} \E_{P_{X,Y}}\left[\log\left(\frac{Q_{Y|X}}{Q_{Y}}(X,Y)\right)\right]\nonumber\\
        &=\frac{1}{\E[X]}\left(\E_{P_{X,Y}}\left[\log\left(\frac{1-\delta}{1-\delta^{m+1}} \delta^{Y-X}\right)\right]- \E_{P_{X,Y}}\left[\log\left(2^{-\alpha(Y-\mu)-H(Z)}\right)\right]\right)\nonumber\\
        &=\frac{\log\left(\frac{1-\delta}{1-\delta^{m+1}}\right)+\mu \log\left(\delta\right)+\E_{P_{X,Y}}\left[\alpha(Y-\mu)+H(Z)\right]}{\E[X]}\nonumber\\
        &=\alpha,\nonumber
    \end{align}
    where \cref{use_size_finite_support} is well defined since $\Geom_{\mu,m}$ has support $\{0,1,\dots,m\}$.
\end{proof}

\section{Acknowledgements}

JR thanks Daniella Bar-Lev, Serge Kas Hanna, Olgica Milenkovic, and Ligong Wang for insightful discussions.

R.\ Con’s work on this paper was carried out while he was a postdoctoral researcher in the Faculty of Computer Science at the Technion -- Israel Institute of Technology. During that period, he was partially supported by the European Union (DiDAX, 101115134).
The work of C.\ Bouette, S.\ Pearson, and J.\ Ribeiro was funded by the European Union (LESYNCH, 101218842).
Views and opinions expressed are however those of the authors only and do not necessarily reflect those of the European Union or the European Research Council Executive Agency. Neither the European Union nor the granting authority can be held responsible for them.
The work of C.\ Bouette, S.\ Pearson, and J.\ Ribeiro was also funded by national funds through FCT – Fundação para a Ciência e a Tecnologia, I.P., and, when eligible, co-funded by EU funds under project/support UID/50008/2025 – Instituto de Telecomunicações, with DOI \href{https://doi.org/10.54499/UID/50008/2025}{10.54499/UID/50008/2025}.

\bibliographystyle{alpha}
\bibliography{refs}

@ARTICLE{AG93,
author={Khaled A. S. {Abdel-Ghaffar}},
journal={Electronics Letters},
title={Capacity per unit cost of a discrete memoryless channel},
year={1993},
volume={29},
number={2},
pages={142-144},
doi={10.1049/el:19930096},
ISSN={0013-5194},
}

@ARTICLE{Ver90,
  author={Verdú, Sergio},
  journal={IEEE Transactions on Information Theory}, 
  title={On channel capacity per unit cost}, 
  year={1990},
  volume={36},
  number={5},
  pages={1019-1030},
  doi={10.1109/18.57201}}

@ARTICLE{BA98,
  author={Bedekar, A.S. and Azizoglu, M.},
  journal={IEEE Transactions on Information Theory}, 
  title={The information-theoretic capacity of discrete-time queues}, 
  year={1998},
  volume={44},
  number={2},
  pages={446-461},
  doi={10.1109/18.661496}}

@ARTICLE{PG03,
  author={Prabhakar, B. and Gallager, R.},
  journal={IEEE Transactions on Information Theory}, 
  title={Entropy and the timing capacity of discrete queues}, 
  year={2003},
  volume={49},
  number={2},
  pages={357-370},
  doi={10.1109/TIT.2002.807287}}

@article{Arimoto72,
  title={An algorithm for computing the capacity of arbitrary discrete memoryless channels},
  author={Arimoto, Suguru},
  journal={IEEE Transactions on Information Theory},
  volume={18},
  number={1},
  pages={14--20},
  year={1972},
  publisher={IEEE}
}

@article{blahut1972,
  title={Computation of channel capacity and rate-distortion functions},
  author={Blahut, Richard E},
  journal={IEEE Transactions on Information Theory},
  volume={18},
  number={4},
  pages={460--473},
  year={1972}
}

@article{jimbo1979iteration,
  title={An iteration method for calculating the relative capacity},
  author={Jimbo, Masakazu and Kunisawa, Kiyonori},
  journal={Information and Control},
  volume={43},
  number={2},
  pages={216--223},
  year={1979},
  publisher={Elsevier}
}

@ARTICLE{CR21,
  author={Cheraghchi, Mahdi and Ribeiro, João},
  journal={IEEE Transactions on Information Theory}, 
  title={An Overview of Capacity Results for Synchronization Channels}, 
  year={2021},
  volume={67},
  number={6},
  pages={3207-3232},
  doi={10.1109/TIT.2020.2997329}}

@inproceedings{WGGH23,
  title={Embracing errors is more effective than avoiding them through constrained coding for {DNA} data storage},
  author={Weindel, Franziska and Gimpel, Andreas L. and Grass, Robert N. and Heckel, Reinhard},
  booktitle={2023 59th Annual Allerton Conference on Communication, Control, and Computing (Allerton)},
  pages={1--8},
  year={2023},
  organization={IEEE}
}

@article{YGM17,
  title={Portable and error-free {DNA}-based data storage},
  author={Yazdi, S. M. Hossein Tabatabaei and Gabrys, Ryan and Milenkovic, Olgica},
  journal={Scientific reports},
  volume={7},
  number={1},
  pages={5011},
  year={2017},
  publisher={Nature Publishing Group UK London}
}

@misc{CMRS26,
      author = {Yuan-Pon Chen and Olgica Milenkovic and João Ribeiro and Jin Sima},
      title = {Correcting Contextual Deletions in {DNA} Nanopore Readouts},
      note = {Preprint. \url{https://arxiv.org/abs/2602.05072}},
      year={2026}
}

@article{Mit09,
	title={A survey of results for deletion channels and related synchronization channels},
	author={Mitzenmacher, Michael},
	journal={Probability Surveys},
	volume={6},
	pages={1--33},
	year={2009},
	publisher={The Institute of Mathematical Statistics and the Bernoulli Society}
}

@article{Mit08,
	title={Capacity bounds for sticky channels},
	author={Mitzenmacher, Michael},
	journal={IEEE Transactions on Information Theory},
	volume={54},
	number={1},
	pages={72--77},
	year={2008},
	publisher={IEEE},
    doi={10.1109/TIT.2007.911291}
}

@ARTICLE{MTL12,
author={Hugues Mercier and Vahid Tarokh and Fabrice Labeau},
journal={IEEE Transactions on Information Theory},
title={Bounds on the capacity of discrete memoryless channels corrupted by synchronization and substitution errors},
year={2012},
volume={58},
number={7},
pages={4306-4330},
doi={10.1109/TIT.2012.2191682},
}

@article{ISW16,
	title={On the capacity of channels with timing synchronization errors},
	author={Iyengar, Aravind R. and Siegel, Paul H. and Wolf, Jack Keil},
	journal={IEEE Transactions on Information Theory},
	volume={62},
	number={2},
	pages={793--810},
	year={2016},
	publisher={IEEE},
    doi={10.1109/TIT.2015.2504358}
}

@ARTICLE{CR19b, 
author={Mahdi {Cheraghchi} and Jo{\~a}o {Ribeiro}}, 
journal={IEEE Transactions on Information Theory}, 
title={Sharp Analytical Capacity Upper Bounds for Sticky and Related Channels}, 
year={2019}, 
volume={65}, 
number={11}, 
pages={6950-6974}, 
doi={10.1109/TIT.2019.2920375}, 
ISSN={1557-9654}, 
month={Nov}}

@ARTICLE{CR26,
  author={Con, Roni and Ribeiro, João},
  journal={IEEE Transactions on Information Theory}, 
  title={Channels With Input-Correlated Synchronization Errors}, 
  year={2026},
  volume={72},
  number={8},
  pages={5442-5472},
  doi={10.1109/TIT.2026.3694251}}

@misc{conrad2004probability,
  title={Probability distributions and maximum entropy},
  author={Conrad, Keith},
  year={2004},
  note={Available at \url{https://kconrad.math.uconn.edu/blurbs/analysis/entropypost.pdf}}
}

@article{RKR18,
  title={From squiggle to basepair: computational approaches for improving nanopore sequencing read accuracy},
  author={Rang, Franka J. and Kloosterman, Wigard P. and {De Ridder}, Jeroen},
  journal={Genome biology},
  volume={19},
  number={1},
  pages={90},
  year={2018},
  publisher={Springer}
}

@ARTICLE{LM03,
  author={Lapidoth, Amos and Moser, Stefan M.},
  journal={IEEE Transactions on Information Theory}, 
  title={Capacity bounds via duality with applications to multiple-antenna systems on flat-fading channels}, 
  year={2003},
  volume={49},
  number={10},
  pages={2426-2467},
  doi={10.1109/TIT.2003.817449}}

@article{BSBHT13,
  title={Shining a light on dark sequencing: characterising errors in {Ion Torrent PGM} data},
  author={Bragg, Lauren M. and Stone, Glenn and Butler, Margaret K. and Hugenholtz, Philip and Tyson, Gene W.},
  journal={PLoS computational biology},
  volume={9},
  number={4},
  pages={e1003031},
  year={2013},
  publisher={Public Library of Science San Francisco, USA}
}

@article{CheRib19,
author = {Cheraghchi, Mahdi and Ribeiro, Jo\~{a}o},
title = {Improved Upper Bounds and Structural Results on the Capacity of the Discrete-Time Poisson Channel},
year = {2019},
issue_date = {July 2019},
publisher = {IEEE Press},
volume = {65},
number = {7},
issn = {0018-9448},
url = {https://doi.org/10.1109/TIT.2019.2896931},
doi = {10.1109/TIT.2019.2896931},
journal = {IEEE Transactions on Information Theory},
month = jul,
pages = {4052–4068},
numpages = {17}
}

@misc{Kov26,
      title={The Capacity of a Family of Sticky Channels}, 
      author={Mladen Kovačević},
      year={2026},
      note={\url{https://arxiv.org/abs/2607.28281}}, 
}

\appendix

\section{Capacity of additive-noise sticky channels and capacity per unit cost of additive-noise memoryless channels} \label{sec:ANSC-to-ANMC}

We now formalize the intuition that the capacity of additive-noise sticky channels equals the capacity of the corresponding additive-noise memoryless channel over the positive integers. We begin by recalling the framework of Verdú \cite{Ver90}, particularized to the cost function $c(x):=x$.

\begin{definition}[$(n,M,\nu,\eps)$ code]
\label{def:cost-code}
    Let $\Ch$ be a channel with input alphabet $\cA \subseteq \N_{>0}$. We say that $\cC \subseteq \cA^{n}$ is an \emph{$(n,M,\nu,\eps)$ code} with respect to $\Ch$ if it has $|\cC|=M$ with each codeword $c=(c_1,\dots,c_n) \in \cC$ satisfying
    \begin{equation*}
        \sum_{i=1}^{n} c_i \leq \nu,
    \end{equation*}
    and where the average decoding error probability (taken over the uniformly random choice of $c\in\cC$) is at most $\eps$.
\end{definition}

\begin{definition}[Achievable rate per unit cost]
    A rate $R > 0$ is \emph{$\eps$-achievable per unit cost} with respect to $\Ch$ if for every $\gamma > 0$ there exists $\nu_0 > 0$ such that, for all $\nu \geq \nu_0$, there is an $(n,M,\nu,\eps)$ code with respect to $\Ch$ satisfying $\log M > \nu (R-\gamma)$. We say that $R$ is achievable per unit cost with respect to $\Ch$ if it is $\eps$-achievable per unit cost for all $\eps > 0$.
\end{definition}

\begin{definition}[Capacity per unit cost]
    The \emph{capacity per unit cost of $\Ch$}, denoted by $\CPUC(\Ch)$, is the supremum of the achievable rates per unit cost with respect to $\Ch$.
\end{definition}

 It turns out that we may assume that $\nu$ scales linearly with the blocklength $n$, a fact that will be useful below.
\begin{lemma}[{\cite[Corollary of Theorem 2]{Ver90}}] \label{lem:cost-linear-in-blocklength}
    A rate $R > 0$ is \emph{achievable per unit cost} with respect to $\Ch$ if and only if for every $\eps,\gamma > 0$ there exist $c,\nu_0 > 0$ such that, for all $\nu \geq \nu_0$, there is an $(n,M,\nu,\eps)$ code with respect to $\Ch$ satisfying $\log M > \nu (R-\gamma)$ and $\nu \leq cn$.
\end{lemma}

\begin{remark}
    In fact, \cite[Corollary of Theorem 2]{Ver90} only states that one can restrict the blocklength $n$ to grow at most linearly with the cost $\nu$, and not the converse. However, inspection of the proof easily reveals that the converse also holds (i.e., one can restrict $\nu$ to grow at most linearly in $n$, as described in \Cref{lem:cost-linear-in-blocklength} and implicitly noted by Mitzenmacher \cite{Mit08}).
\end{remark}

\subsection{Proof of \Cref{lem:ANSC-to-ANMC}} \label{proof:ANSC-to-ANMC}

    As mentioned before, this argument is analogous to that of Mitzenmacher~\cite[proof of Theorem 2.1]{Mit08}.
    We reproduce it here for completeness.

    We begin by noting that without loss of generality we may assume that all inputs to the $\ANSC_{\cN}$ channel begin with a $0$, as this does not change the channel capacity.
    Under this condition, there is a straightforward one-to-one correspondence between inputs to each channel that follows by 
    mapping an input $x=0^{\ell_1} 1^{\ell_2} 0^{\ell_3}\dots$ to the $\ANSC_{\cN}$ to its ``runlength-encoding'' $(\ell_1,\ell_2,\ell_3,\dots)$, which will be sent through the $\ANMC_{\cN}$. 
    Then, the decoding error probability of a code over the $\ANSC_{\cN}$ equals the decoding error probability of its runlength-encoding over the $\ANMC_{\cN}$.

    We now show that any rate $R$ achievable per unit cost in $\ANMC_{\cN}$ is also achievable (in the usual sense) in $\ANSC_{\cN}$. Let $R$ be achievable per unit cost in $\ANMC_{\cN}$. 
    By \Cref{lem:cost-linear-in-blocklength}, for every $\eps,\gamma > 0$, there is a constant $c > 0$ such that, for all $\nu$ large enough, there exists an $(n,M,\nu,\eps)$ code $\cC_{\N}$ for $\ANMC_{\cN}$ satisfying $\log M > \nu (R - \gamma)$ and $n \leq \nu \leq c n$. By the one-to-one correspondence mentioned above, $\cC_{\N}$ is then mapped to a binary code $\cC$ with codeword lengths between $n$ and $\nu\leq c n$ and average error probability $\eps$ in $\ANSC_{\cN}$. We now define $\cC'$ to contain the codewords of $\cC$ with decoding error probability at most $2 \eps$. An averaging argument gives that $M' := |\cC'| \geq M/2$. We further restrict ourselves to the set $\cC''$ of codewords of $\cC'$ of a fixed length $\ell$ with $n \leq \ell \leq \nu$, where $\ell$ is chosen to maximize $\log |\cC''|$. We then have that
    \begin{equation*}
        \frac{\log |\cC''|}{\ell} \geq \frac{\log M'}{\ell} - \frac{\log (\nu - n + 1)}{\ell} \geq \frac{\log M}{\ell} - \frac{1 + \log (\nu - n + 1)}{\ell} > R-\gamma-o(1).
    \end{equation*}
    Thus, for $n$ large enough, there exists a code for $\ANSC_{\cN}$ with rate at least $R - 2\gamma$ and maximum decoding error probability $2 \eps$, as desired.

    By a similar argument, we can show that any rate $R$ achievable in $\ANSC_{\cN}$ is also achievable per unit cost in $\ANMC_{\cN}$. To this end, let $R$ be an achievable rate in $\ANSC_{\cN}$. Then, for every $\eps, \gamma > 0$, for $n$ large enough, there exists a code $\cC$ for $\ANSC_{\cN}$ of size $M$ with blocklength $n$, maximum decoding error probability $\eps$, and such that $\log M > n(R-\gamma)$. 
    Consider the code $\cC_{\N}$ obtained from $\cC$ via runlength-encoding as described above.
    Then, $\cC_{\N}$ has average decoding error at most $\eps$ in $\ANMC_{\cN}$, constant codeword cost $n$, and the length of each codeword is at most $n$. Thus, we take $\cC_{\N}'$ to be the subset of $\cC_{\N}$ obtained by choosing the codewords of fixed length $\ell \leq n$, where $\ell$ is chosen to maximize $\log |\cC_{\N}'|$. We now have that
    \begin{equation*}
        \frac{\log |\cC_{\N}'|}{n} \geq \frac{\log M -\log n}{n} > R-\gamma-o(1),
    \end{equation*}
    and the result follows. \qed

\section{Proof of \Cref{lem:AG_necessary}} \label{sec:AG-necessary}

In our proof, we shall invoke definitions and results from \cite[Appendix B]{CheRib19}. For completeness, we now state the necessary tools, specialized to our use case. We fix a cost function $c$, and let $\Omega$ denote the space of all probability distributions over the positive integers with finite expected cost (note that $\Omega$ is convex). For a functional $f: \Omega \to \R$, the \emph{weak derivative} of $f$ at $P_X \in \Omega$ in the direction of the distribution $Q$ is given by
    \begin{equation*}
        f_{P_X}'(Q) := \lim_{\theta \to 0^+} \frac{f((1-\theta)P_X + \theta Q) - f(P_X)}{\theta}.
    \end{equation*}
\begin{lemma}[{\cite[Lemma 18]{CheRib19}}]\label{lem:maximizer-derivatives}
    If $f: \Omega \to \R$ is concave, achieves its maximum at $P_X \in \Omega$, and $f_{P_X}'(Q)$ exists, then
    \begin{equation*}
        f_{P_X}'(Q) \leq 0.
    \end{equation*}
\end{lemma}

For a fixed channel $\Ch$, we define the mutual information functional $I : \Omega \to [0,\infty]$ as the map $P_X \mapsto \sum_{x \in \supp(P_X)} P_X(x) \cdot \KL(P_{Y_x}\|P_Y)$, where $P_Y$ is the output distribution corresponding to input $P_X$ given by the channel law for $\Ch$.

\begin{lemma}[{\cite[Lemma 20, extended]{CheRib19}}]\label{lem:MI-derivative}
    Let $\Ch$ be a memoryless channel with input and output alphabets $\cX,\cY \subseteq \N_{>0}$, respectively, and with mutual information functional $I(\cdot)$. Assume that $\E_{X \sim Q} [H(P_{Y_X})] < \infty$, where $Q \in \Omega$. Let $P_X \in \Omega$ with corresponding output distribution $P_Y$ be such that $I(P_X) < \infty$.
    Then, if
    \begin{equation*}
        \sum_{x \in \supp(Q)} Q(x) \KL(P_{Y_x}\|P_Y) < \infty,
    \end{equation*}
    $I_{P_X}'(Q)$ exists and
    \begin{equation*}
        I_{P_X}'(Q) = \sum_{x \in \supp(Q)} Q(x) \KL(P_{Y_x}\|P_Y) - I(P_X).
    \end{equation*}
    Moreover, if $Q = \delta_n$ is a point mass on $n \in \cX$, then $I_{P_X}'(\delta_n) = \infty$ whenever $\KL(P_{Y_n}||P_Y) = \infty$.
\end{lemma}
\begin{proof}
    We only sketch the proof of the last statement, since the rest has been shown in \cite[Lemma 20]{CheRib19}.
    Assume $\KL(P_{Y_n}||P_Y)=\infty$, and let $Q_{\theta} = (1-\theta)P_Y + \theta P_{Y_n}$ be the output distribution corresponding to input $(1-\theta)P_X + \theta \delta_n$, for $\theta \in [0,1]$. We note that
    \begin{align} \label{eq:I_convex_decomp}
        \frac{I((1-\theta)P_X + \theta \delta_n) - I(X)}{\theta} &= \KL(P_{Y_n}||Q_{\theta}) - I(P_X) \notag \\
        &+ \frac{1-\theta}{\theta} \left[ \sum_{x \in \supp(P_X)} P_X(x) \left( \KL(P_{Y_x}||Q_{\theta}) - \KL(P_{Y_x}||P_Y) \right) \right].
    \end{align}
    Moreover, considering a sequence $\{\theta_k\}_{k \in \N_{>0}}$ such that $\theta_k \to 0^+$ as $k \to \infty$, and letting $f(y) = \log \left( \frac{1}{P_Y(y)} \right)$ and $f_k(y) = \log \left( \frac{1}{Q_{\theta_k}(y)} \right)$, we have $f_k \to f$ pointwise (in particular, we have $f(y) = \liminf_{k \to \infty} f_k(y) \in [0,\infty]$) for all $y \in \supp(P_{Y_n})$.
    By Fatou's lemma,
    \begin{equation} \label{eq:Fatou}
        \sum_{y \in \supp(P_{Y_n})} P_{Y_n}(y) \log (1/P_Y(y)) \leq \liminf_{k \to \infty} \sum_{y \in \supp(P_{Y_n})} P_{Y_n}(y) \log (1/Q_{\theta_k}(y)).
    \end{equation}
    Since the left-hand side of \Cref{eq:Fatou} is infinite (because $\KL(P_{Y_n}||P_Y) = \infty$ and $H(P_{Y_n})<\infty$), we conclude that
    \begin{equation} \label{eq:infiniteLHS}
        \KL(P_{Y_n}||Q_{\theta}) = -H(P_{Y_n}) + \sum_{y \in \supp(P_{Y_n})} P_{Y_n}(y) \log (1/Q_{\theta}(y)) \to \infty
    \end{equation}
    as $\theta \to 0^+$.
    Since $I(P_X)<\infty$, by \Cref{eq:I_convex_decomp} it is enough to show that
    \begin{equation*}
       \lim_{\theta \to 0^+} \frac{1-\theta}{\theta} \left[ \sum_{x \in \supp(P_X)} P_X(x) \left( \KL(P_{Y_x}||Q_{\theta}) - \KL(P_{Y_x}||P_Y) \right) \right] = 0.
    \end{equation*}
    The proof of this fact is analogous to the proof of~\cite[Equation (44)]{CheRib19}, and we do not repeat it here.
\end{proof}
\begin{remark}
    The original statement of \cite[Lemma 20]{CheRib19} does not include the condition that $\E_{X \sim Q} [H(P_{Y_X})] < \infty$. However, upon inspection of the proof, this condition is implicitly assumed in~\cite[Equation (41)]{CheRib19}, and appears to be necessary for the proof to go through.
\end{remark}

\begin{proof}[Proof of \Cref{lem:AG_necessary}]    
    Let $E : \Omega \rightarrow \R$ denote the functional given by $E(Q) = \sum_{x \in \N_{>0}} Q(x) c(x)$, and let $I$ be the mutual information functional for $\Ch$. Since $P_X$ is capacity-achieving, we have that $P_X$ maximizes the concave functional $f : \Omega \to \R$ defined as $f(Q) = I(Q) - \lambda E(Q)$, with $\lambda = \CPUC_c(\Ch)$, and that $f(P_X) = 0$ (the concavity of $f$ follows from the concavity of $I$ and the linearity of $E$).

    By \Cref{lem:MI-derivative}, if $\KL(P_{Y_x}||P_Y)=\infty$ for any $x \in \cX$, then $I_{P_X}'(\delta_x) = \infty$, which implies $f_{P_X}'(\delta_x)= \infty$, and so by \Cref{lem:maximizer-derivatives} $P_X$ is not capacity-achieving (a contradiction). Thus, we may assume that $\KL(P_{Y_x}||P_Y) < \infty$ for all $x \in \cX$.
    
    Now, take any $x \in \cX$ and let $\delta_x$ be a point mass on $x$.
    Since $\KL(P_{Y_x}\|P_Y) < \infty$ and $H(P_{Y_x}) < \infty$, \Cref{lem:MI-derivative} implies that
    \begin{equation*}
        I_X'(\delta_x) = \KL(P_{Y_x}\|P_Y) - I(P_X) = \KL(P_{Y_x}\|P_Y) - \lambda E(P_X),
    \end{equation*}
    where we used the fact that $I(P_X) = \lambda E(P_X)$ since $P_X$ is capacity-achieving.
    Moreover,
    \begin{align*}
        E_{P_X}'(\delta_x) &= \lim_{\theta \to 0^+} \frac{E(P_X(1-\theta) + \theta \delta_x) - E(P_X)}{\theta} \\
        &= \lim_{\theta \to 0^+} \frac{\theta c(x) - \theta E(P_X)}{\theta} \\
        &= c(x) - E(P_X).
    \end{align*}
    Thus,
    \begin{equation*}
        f_{P_X}'(\delta_x) = I_{P_X}'(\delta_x) - \lambda E_{P_X}'(\delta_x) = \KL(P_{Y_x}\|P_Y) - \lambda c(x).
    \end{equation*}
    Invoking \Cref{lem:maximizer-derivatives}, we conclude that
    \begin{equation} \label{eq:DKL_leq_lambda_c(x)}
        \KL(P_{Y_x}\|P_Y) \leq \lambda c(x)
    \end{equation}
    for all $x \in \cX$.

    Now, let $x \in \supp(P_X)$. If $P_X(x) = 1$, then 
    $\lambda = I(P_X)/E(P_X) = \KL(P_{Y_x}\|P_Y)/c(x) = 0$,
    so we may assume $P_X(x) < 1$. Let $P_{X_x}$ be the distribution satisfying $P_{X_x}(x) = 0$ and $P_{X_x}(n) = P_X(n)/(1-P_X(x))$ for $n \neq x$ (intuitively, $P_{X_x}$ is obtained from $P_X$ by removing its mass at $x$ and distributing it evenly across the remaining support points).
    Note that
    \begin{align*}
        \sum_{n \in \supp(P_{X_x})} P_{X_x}(n) \KL(P_{Y_n}\|P_Y) =& \frac{1}{1-P_X(x)} \sum_{n \in \supp(P_X)} P_{X}(n) \KL(P_{Y_n}\|P_Y) \\
        &- \frac{P_X(x)}{1-P_X(x)} \KL(P_{Y_x}\|P_Y) \\
        =& \frac{1}{1-P_X(x)} \cdot \left( I(P_X) - P_X(x) \KL(P_{Y_x}\|P_Y) \right) <\infty.
    \end{align*}
    Since $\sum_{x \in \cX} P_X(x) H(P_{Y_x}) < \infty$, it easily follows that $\sum_{x \in \cX} P_{X_x}(x) H(P_{Y_x}) < \infty$. Once more applying \Cref{lem:MI-derivative}, we have
    \begin{equation} \label{eq:MI-part}
        I_{P_X}'(P_{X_x}) = \frac{P_X(x)}{1-P_X(x)} \cdot \left( I(P_X) - \KL(P_{Y_x}\|P_Y) \right).
    \end{equation}
    A straightforward computation yields $E(P_{X_x}) = \frac{E(P_X) - P_X(x) c(x)}{1-P_X(x)}$. Thus,
    \begin{align}\label{eq:E-part}
        E_{P_X}'(P_{X_x}) &= \lim_{\theta \to 0^+} \frac{E(P_X(1-\theta) + \theta P_{X_x}) - E(P_X)}{\theta} \notag \\
        &= \lim_{\theta \to 0^+} \frac{\theta E(P_{X_x}) - \theta E(P_X)}{\theta} \notag \\
        &= \frac{P_X(x)}{1-P_X(x)} \cdot \left( E(P_X) - c(x) \right).
    \end{align}
    Joining \Cref{eq:MI-part,eq:E-part} and recalling that $I(P_X) = \lambda E(P_X)$, we obtain
    \begin{equation*}
        f_{P_X}'(P_{X_x}) = \frac{P_X(x)}{1-P_X(x)} \cdot \left( \lambda c(x) - \KL(P_{Y_x}\|P_Y) \right).
    \end{equation*}
    By \Cref{lem:maximizer-derivatives}, it must hold that
    \begin{equation} \label{eq:reverse-AG}
        \lambda c(x) \leq \KL(P_{Y_x}\|P_Y).
    \end{equation}
    Since we have seen in \Cref{eq:DKL_leq_lambda_c(x)} that the reverse inequality must also hold, \Cref{eq:reverse-AG} holds with equality for $x \in \supp(P_X)$, as desired.
\end{proof}

\end{document}